\documentclass[a4paper, final]{article}

\usepackage{a4wide}
\usepackage[utf8]{inputenc}
\usepackage{amsfonts}
\usepackage{amsmath, amsthm}
\usepackage{enumitem}
\usepackage{microtype}

\usepackage{hyperref}

\usepackage{pgf}
\usepackage{tikz}

\theoremstyle{definition}

\newtheorem{theorem}{Theorem}[section]
\newtheorem{lemma}[theorem]{Lemma}
\newtheorem{proposition}[theorem]{\textbf{Proposition}}
\newtheorem{corollary}[theorem]{Corollary}
\newtheorem{example}[theorem]{Example}
\newtheorem{definition}[theorem]{Definition}
\newtheorem{remark}{Remark}

\newcommand{\cA}{\mathcal{A}}
\newcommand{\cB}{\mathcal{B}}
\newcommand{\cC}{\mathcal{C}}
\newcommand{\cO}{\mathcal{O}}
\newcommand{\cP}{\mathcal{P}}
\newcommand{\cQ}{\mathcal{Q}}
\newcommand{\cS}{\mathcal{S}}
\newcommand{\cV}{\mathcal{V}}

\newcommand{\fA}{\mathfrak{A}}
\newcommand{\fB}{\mathfrak{B}}

\newcommand{\fUA}{\mathsf{A}}
\newcommand{\fUB}{\mathsf{B}}
\newcommand{\fUD}{\mathsf{D}}

\newcommand{\hS}{{\widehat{S}}}
\newcommand{\hcS}{\widehat{\cS}}

\newcommand{\hfUD}{\widehat{\mathsf{D}}}

\newcommand{\<}{\langle}
\renewcommand{\>}{\rangle}

\newcommand{\vars}{\text{\upshape{vars}}}
\newcommand{\len}{\text{\upshape{len}}}

\newcommand{\subst}[1]{\bigl[#1\bigr]}

\newcommand{\semequiv}{\mathrel{|}\joinrel\Relbar\joinrel\mathrel{|}}

\newcommand{\At}{\text{\upshape{At}}}
\newcommand{\Sk}{\text{\upshape{Sk}}}

\newcommand{\fa}{\mathsf{a}}
\newcommand{\fb}{\mathsf{b}}
\newcommand{\fc}{\mathsf{c}}
\newcommand{\fd}{\mathsf{d}}
\newcommand{\fe}{\mathsf{e}}

\newcommand{\va}{\bar{\fa}}
\newcommand{\vb}{\bar{\fb}}
\newcommand{\vc}{\bar{\fc}}
\newcommand{\vd}{\bar{\fd}}
\newcommand{\ve}{\bar{\fe}}

\newcommand{\vu}{\bar{\mathrm{u}}}
\newcommand{\vv}{\bar{\mathrm{v}}}
\newcommand{\vw}{\bar{\mathrm{w}}}
\newcommand{\vx}{\bar{\mathrm{x}}}
\newcommand{\vy}{\bar{\mathrm{y}}}
\newcommand{\vz}{\bar{\mathrm{z}}}

\newcommand{\twoup}[2]{{2^{\uparrow #1}(#2)}}

\newcommand{\idx}{\text{\upshape{idx}}}

\newcommand{\Mapsto}{{\mathop{\mapsto}}}

\newcommand{\EXPTIME}{\textup{\textsc{ExpTime}}}
\newcommand{\NEXPTIME}{\textup{\textsc{NExpTime}}}

\newcommand{\SF}{\text{\upshape{SF}}}
\newcommand{\LGF}{\text{\upshape{LGF}}}
\newcommand{\GNFO}{\text{\upshape{GNFO}}}

\newcommand{\MFOeq}{\text{\upshape{MFO$_\approx$}}}
\newcommand{\FOk}{\text{\upshape{FO}$^k$}}
\newcommand{\FOtwo}{\text{\upshape{FO}$^2$}}
\newcommand{\SFOk}{\text{\upshape{SFO}$^k$}}
\newcommand{\SFOkm}{\text{\upshape{SFO}$^{k,m}$}}
\newcommand{\SFOtwo}{\text{\upshape{SFO}$^2$}}

\newcommand{\FL}[1]{\text{\upshape{FL}$^{(#1)}$}}

\begin{document}

\title{Separateness of Variables \\ --- \\ A Novel Perspective on Decidable First-Order Fragments}

\author{
	\begin{tabular}{l}
		Marco Voigt\\
		\small\textit{Max Planck Institute for Informatics, Saarland Informatics Campus, Saarbr\"ucken, Germany,}\\
		\small\textit{Saarbr\"ucken Graduate School of Computer Science}
	\end{tabular}
}	
\date{}
\maketitle

\begin{abstract}
The classical decision problem, as it is understood today, is the quest for a delineation between the decidable and the undecidable parts of first-order logic based on elegant syntactic criteria. In this paper, we treat the concept of separateness of variables and explore its applicability to the classical decision problem. Two disjoint sets of first-order variables are separated in a given formula if variables from the two sets never co-occur in any atom of that formula. This simple notion facilitates extending many well-known decidable first-order fragments significantly and in a way that preserves decidability. We will demonstrate that for several prefix fragments, several guarded fragments, the two-variable fragment, and for the fluted fragment. Altogether, we will investigate nine such extensions more closely. Interestingly, each of them contains the relational monadic first-order fragment without equality. Although the extensions exhibit the same expressive power as the respective originals, certain logical properties can be expressed much more succinctly. In three cases the succinctness gap cannot be bounded using any elementary function.
\end{abstract}


\section{Introduction}
\label{section:Introduction}

In the early twentieth century David Hilbert initiated his famous program striving for a formalization of the foundations of mathematics.
At its core lay the \emph{classical decision problem} of first-order logic: Find an algorithm that determines the validity of any given first-order sentence.
Nowadays, the classical decision problem is understood as the problem of classifying first-order logic into fragments with a decidable or undecidable satisfiability problem.\footnote{For convenience, we will be less precise every now and then and speak of \emph{(un)decidable fragments}.}
This quest has produced a wealth of positive and also negative results, see \cite{Dreben1979, Borger1997, Fermuller2001} for references.
We will give definitions of a number of well-known decidable fragments in later sections. 
More recently identified decidable fragments that are neither covered in the present paper nor in the mentioned texts are treated in~\cite{Segoufin2013, Kieronski2014, Mogavero2015, Segoufin2017, Bova2017}.
A much more detailed but still brief overview of past and recent developments in the area is given in the beginning of Chapter~3 of~\cite{Voigt2019PhDthesis}.

In the present paper we explore the concept of \emph{separateness of variables} in the context of the classical decision problem.
Two sets $X, Y$ of variables are \emph{separated in a formula}, if there are no co-occurrences of any $x \in X$ and $y \in Y$ in any atom.
Put more formally, we get the following definition.
\begin{definition}[Separateness of variables]
	Let $\varphi$ be any first-order formula and let $X, Y$ be two disjoint sets of first-order variables.
	We say that $X$ and $Y$ are \emph{separated in $\varphi$}
	 if for every atom $A$ occurring in $\varphi$ we have $\vars(A) \cap X = \emptyset$  or $\vars(A) \cap Y = \emptyset$ or both.
\end{definition}
This simple concept enables us to elegantly define nontrivial extensions of well-known decidable fragments of first-order logic.
In essence, the definitions of the new fragments are careful combinations of the concepts used in the original definitions with the concept of separateness of quantified variables.
All new fragments still have a decidable satisfiability problem.
Figure~\ref{figure:KnownAndNovelFOLfragments} provides a schematic overview of the most important fragments we will introduce.
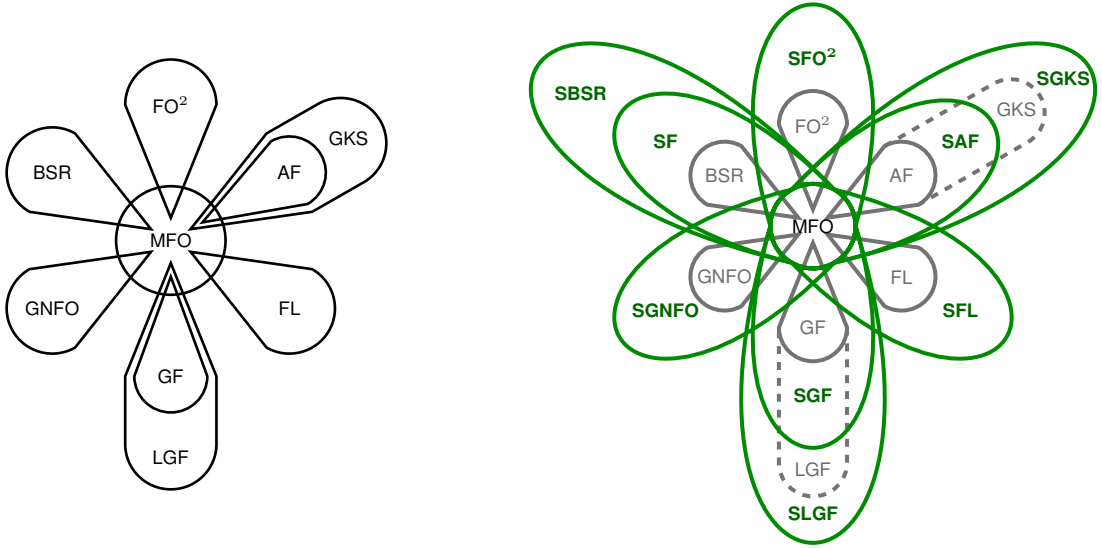
\begin{figure}[h]
	\hspace{-2ex}
	\begin{tabular}{@{}c@{\hspace{-7ex}}c@{}}
		\begin{tabular}{c}
			{\fontfamily{phv}\selectfont
				\begin{tikzpicture}[scale=0.75]	
					\begin{scope}[line width=1pt, color=black, scale=0.8]
						\draw (0,0) circle [radius=1.2]  node[color=black] at (0,0) {\scriptsize MFO};
						\draw [rotate=0] (0,-0.5) -- (1,-3) -- (1,-4.5) arc[start angle=0, end angle=-180, radius=1] -- (-1,-3) -- cycle node at (0,-4.8) {\scriptsize LGF};
						\draw [rotate=120] (0,-0.5) -- (1,-3) -- (1,-4.3) arc[start angle=0, end angle=-180, radius=1] -- (-1,-3) -- cycle node at (0,-4.5) {\scriptsize GKS};
						\draw [rotate=0] (0,-0.8) -- (0.8,-3) arc[start angle=0, end angle=-180, radius=0.8] -- cycle node at (0,-3) {\scriptsize GF};
						\draw [rotate=60]  (0,-0.5) -- (1,-3) arc[start angle=0, end angle=-180, radius=1] -- cycle node at (0,-3) {\scriptsize FL};
						\draw [rotate=120]  (0,-0.8) -- (0.8,-3) arc[start angle=0, end angle=-180, radius=0.8] -- cycle node at (0,-3) {\scriptsize AF};
						\draw [rotate=180]  (0,-0.5) -- (1,-3) arc[start angle=0, end angle=-180, radius=1] -- cycle node at (0,-3) {\scriptsize \text{\upshape{FO}$^2$}};
						\draw [rotate=240]  (0,-0.5) -- (1,-3) arc[start angle=0, end angle=-180, radius=1] -- cycle node at (0,-3) {\scriptsize BSR};
						\draw [rotate=300]  (0,-0.5) -- (1,-3) arc[start angle=0, end angle=-180, radius=1] -- cycle node at (0,-3) {\scriptsize GNFO};
					\end{scope}
				\end{tikzpicture}		
			}
		\end{tabular}
		&
		\begin{tabular}{c}
			{\fontfamily{phv}\selectfont
				\begin{tikzpicture}[scale=0.56]
					\clip (-9.8,-7.7) rectangle (9.8,5.4);
					\begin{scope}[line width=1.5pt, color=gray!90!black, scale=0.8]
						\begin{scope}[color=gray!90!black, dashed]
							\draw [rotate=0] (1,-3) -- (1,-7) arc[start angle=0, end angle=-180, radius=1] -- (-1,-3) node[color=gray!85!black] at (0,-7.2) {\scriptsize LGF};
							\draw [rotate=120] (1,-3) -- (1,-6.7) arc[start angle=0, end angle=-180, radius=1] -- (-1,-3) node[color=gray!85!black] at (0,-6.9) {\scriptsize GKS};
						\end{scope}
						\draw [rotate=0] (0,-0.5) -- (1,-3) arc[start angle=0, end angle=-180, radius=1] -- cycle node[color=gray!85!black] at (0,-3) {\scriptsize GF};
						\draw [rotate=60]  (0,-0.5) -- (1,-3) arc[start angle=0, end angle=-180, radius=1] -- cycle node[color=gray!85!black] at (0,-3) {\scriptsize FL};
						\draw [rotate=120]  (0,-0.5) -- (1,-3) arc[start angle=0, end angle=-180, radius=1] -- cycle node[color=gray!85!black] at (0,-3) {\scriptsize AF};
						\draw [rotate=180]  (0,-0.5) -- (1,-3) arc[start angle=0, end angle=-180, radius=1] -- cycle node[color=gray!85!black] at (0,-3) {\scriptsize \text{\upshape{FO}$^2$}};
						\draw [rotate=240]  (0,-0.5) -- (1,-3) arc[start angle=0, end angle=-180, radius=1] -- cycle node[color=gray!85!black] at (0,-3) {\scriptsize BSR};
						\draw [rotate=300]  (0,-0.5) -- (1,-3) arc[start angle=0, end angle=-180, radius=1] -- cycle node[color=gray!85!black] at (0,-3) {\scriptsize GNFO};
					\end{scope}
			
					\begin{scope}[line width=1.5pt, color=green!55!black]
						\draw (0,0) circle [radius=1]  node[color=black] at (0,0) {\scriptsize MFO};
						\draw [rotate=0] (1,0) arc [start angle=0, end angle=180, radius=1] .. controls (-3,-7) and (3,-7) .. (1,0) node[color=green!40!black] at (0,-4) {\scriptsize \textbf{SGF}};
						\draw [rotate=0] (1,0) arc [start angle=0, end angle=180, radius=1] .. controls (-4,-10) and (4,-10) .. (1,0) node[color=green!40!black] at (0,-6.8) {\scriptsize \textbf{SLGF}};
						\draw [rotate=60] (1,0) arc [start angle=0, end angle=180, radius=1] .. controls (-3,-7) and (3,-7) .. (1,0) node[color=green!40!black] at (0,-4) {\scriptsize \textbf{SFL}};
						\draw [rotate=120] (1,0) arc [start angle=0, end angle=180, radius=1] .. controls (-3,-6.5) and (3,-6.5) .. (1,0) node[color=green!40!black] at (0,-4) {\scriptsize \textbf{SAF}};
						\draw [rotate=120] (1,0) arc [start angle=0, end angle=180, radius=1] .. controls (-4,-10) and (4,-10) .. (1,0) node[color=green!40!black] at (0.1,-6.85) {\scriptsize \textbf{SGKS}};
						\draw [rotate=180] (1,0) arc [start angle=0, end angle=180, radius=1] .. controls (-3,-7) and (3,-7) .. (1,0) node[color=green!40!black] at (0,-4) {\scriptsize \textbf{SFO$^{\boldsymbol{2}}$}};
						\draw [rotate=240] (1,0) arc [start angle=0, end angle=180, radius=1] .. controls (-3,-7) and (3,-7) .. (1,0) node[color=green!40!black] at (0,-4) {\scriptsize \textbf{SF}};
						\draw [rotate=240] (1,0) arc [start angle=0, end angle=180, radius=1] .. controls (-4,-10) and (4,-10) .. (1,0) node[color=green!40!black] at (0,-6.3) {\scriptsize \textbf{SBSR}};
						\draw [rotate=300] (1,0) arc [start angle=0, end angle=180, radius=1] .. controls (-3,-7) and (3,-7) .. (1,0) node[color=green!40!black] at (0,-4) {\scriptsize \textbf{SGNFO}};
					\end{scope}
				\end{tikzpicture}
			}
		\end{tabular}
		\\
		\begin{tabular}{c@{\,}c@{\;}l}
			MFO &--& monadic first-order fragment\\
			BSR &--& Bernays--Sch\"onfinkel--Ramsey \\
				&& fragment\\
			\FOtwo &--& two-variable fragment\\
			AF &--& Ackermann fragment\\
			GKS &--& G\"odel--Kalm\'ar--Sch\"utte fragment\\
			FL &--& fluted fragment\\
			GF &--& guarded fragment\\
			LGF &--& loosely guarded fragment\\
			GNFO &--& guarded negation fragment
		\end{tabular}
		&
		\begin{tabular}{c@{\,--\;}l}
			SF & separated fragment\\
			SBSR & separated BSR\\
			\SFOtwo & separated \FOtwo\\
			SAF & separated AF\\
			SGKS & separated GKS\\
			SFL & separated FL\\
			SGF & separated GF\\
			SLGF & separated LGF\\
			SGNFO & separated GNFO	
		\end{tabular}
	\end{tabular}
	
	\caption[Overview of known and novel decidable fragments treated in the present thesis]{
			Left-hand side: Schematic overview of well-known decidable fragments of first-order logic. 
			Only the partial overlaps between MFO and the other fragments are depicted. 
			We neglect any other partial overlaps. 
			Moreover, the containment of AF in GKS and of GF in LGF is shown.
			Right-hand side: Schematic overview of the extended fragments (in green) that are treated in the present paper. 
			Notice that MFO is properly contained in all extended fragments. 
			The focus is again on the overlaps with MFO and on the proper containment relations. 
			Other depicted overlaps might be unsubstantiated.}
	\label{figure:KnownAndNovelFOLfragments}
\end{figure}
All of them enjoy the \emph{finite model property}, i.e.\ every satisfiable sentence in such a fragment has a finite model.
This is a sufficient condition for decidability of the associated satisfiability problem.
If we can derive a computable upper bound regarding the size of smallest models, we speak of a \emph{small model property}.

From a qualitative point of view, the extended fragments do not come with an increased expressiveness compared to the original fragments.
This is witnessed by the existence of equivalence-preserving translations from each extended fragment into the respective original.
However, we will show for several cases that there may be significant gaps regarding the length of shortest formulas that express one and the same property.
For example, the succinctness gap between the well-known \emph{Bernays--Sch\"onfinkel--Ramsey fragment} (the $\exists^* \forall^*$ prefix class) and an extension of it called the \emph{separated fragment}~\cite{Voigt2016} (all relational sentences $\exists \vz \forall\vx_1 \exists\vy_1 \ldots \forall\vx_n \exists\vy_n.\, \varphi$ in which $\vx_1 \cup \ldots \cup \vx_n$ and $\vy_1 \cup \ldots \cup \vy_n$ are separated) is as follows~\cite{Voigt2017a}.
For every positive natural number $n$ we can find some property that can be expressed in the separated fragment with a formula of length $k\cdot n^2$ for some fixed natural number $k$, whereas expressing the very same property in the Bernays--Sch\"onfinkel--Ramsey fragment requires a formula whose length is at least 
	\[ \left. \begin{array}{c}2^{2^{\vdots^{2^2}}} \end{array} \right\} \text{height } n ~. \]
Hence, some of the extended fragments enable us to describe certain logical properties much more succinctly and elegantly.
Table~\ref{table:SuccinctnessGapsSummary} summarizes the succinctness gaps that we will derive.
\begin{table*}[ht]
	\caption[Summary of the succinctness gaps the are explored in the present thesis]{Summary of the succinctness gaps the are explored in the present paper. 
			The abbreviations for fragments are spelled out in Figure~\ref{figure:KnownAndNovelFOLfragments}.
			Notice that SF $\subseteq$ SBSR.
			The gap between SAF and AF is conditional on $\EXPTIME \neq \NEXPTIME$. 
			All other gaps are unconditional.
		}
	\label{table:SuccinctnessGapsSummary}	
	\centerline{
		\begin{tabular}{cccl}
			More succinct fragments	&	Less succinct fragments		&	Succinctness gaps 	&	Reference \\[0.5ex]
			\hline\\[-1ex]
			SF / SBSR		&	BSR						&	non-elementary	&	Theorem~\ref{theorem:LengthSmallestBSRsentences} \\
			SF / SBSR		&	Gaifman-local first-				&	non-elementary 	&	Prop.~\ref{theorem:SFtoGaifmanLowerBound} \\[-0.5ex]
						&	order fragment \\
			SAF			&	AF						& 	super-polynomial 	&	Prop.~\ref{proposition:SuccinctnessGapForSAFvsAF} \\
			SGKS			&	GKS						&	exponential		&	Theorem~\ref{theorem:LengthSmallestGKSsentences} \\
			SGF			&	LGF	 					&	non-elementary	&	Theorem~\ref{theorem:LengthSmallestLGFsentences} \\
			SGNFO		&	GNFO	 					&	non-elementary	&	Theorem~\ref{theorem:LengthSmallestGNFOsentences} \\
			\SFOtwo		&	\FOtwo					&	exponential		&	Theorem~\ref{theorem:LengthSmallestSFOtwoSentences} 	
		\end{tabular}
	}
\end{table*}


Non-trivial cases of separateness appear, for instance, in formulas where universal and existential quantifiers are nested and the variables they bind are separated.
Consider the sentence $\varphi := \forall x \exists y.\, P(x) \leftrightarrow Q(y)$ in which the singleton sets $\{x\}$ and $\{y\}$ are obviously separated. 
It expresses a certain symmetry in structures $\cA$.
For every domain element $\fa$ there is some element $\fb$ such that $\fa$ belongs to $P^\cA$ if and only if $\fb$ belongs to $Q^\cA$.
It turns out that the same property can be expressed without any nesting of alternating quantifiers.
Indeed, the distributivity laws of Boolean algebra and quantifier shifting (cf.\ Proposition~\ref{lemma:Miniscoping}) facilitate a transformation of $\varphi$ into the equivalent sentence 
$\psi := \bigl( (\exists x.\, P(x)) \rightarrow (\exists y_1.\, Q(y_1)) \bigr) \wedge \bigl( (\exists x.\, \neg P(x)) \rightarrow (\exists y_2.\, \neg Q(y_2)) \bigr)$ --- the symbol $\semequiv$ denotes semantic equivalence:
	\begin{align*}
		\varphi \;=\;\;\;
		&\forall x \exists y.\, P(x) \leftrightarrow Q(y)\\
		&\semequiv\; \forall x \exists y.\, \bigl(\neg P(x) \vee Q(y)\bigr) \wedge \bigl( P(x) \vee \neg Q(y) \bigr) \\
		&\semequiv\; \forall x \exists y.\, \bigl( \neg P(x) \wedge \neg Q(y) \bigr) \vee \bigl(Q(y) \wedge P(x)\bigr) \\
		&\semequiv\; \forall x.\, \bigl( \neg P(x) \wedge (\exists y_2.\, \neg Q(y_2)) \bigr) \vee \bigl((\exists y_1.\, Q(y_1)) \wedge P(x)\bigr) \\
		&\semequiv\; \bigl( (\forall x.\, \neg P(x)) \vee (\exists y_1.\, Q(y_1)) \bigr) \wedge \bigl( (\exists y_2.\, \neg Q(y_2)) \vee (\forall x.\, P(x)) \bigr) \\
		&\semequiv\; \bigl( (\exists x.\, P(x)) \rightarrow (\exists y_1.\, Q(y_1)) \bigr) \wedge \bigl( (\exists x.\, \neg P(x)) \rightarrow (\exists y_2.\, \neg Q(y_2)) \bigr)
		\;\;=\; \psi
	\end{align*}	
We could even shift quantifiers outwards again and finally obtain an equivalent sentence with a $\exists \exists \forall$ quantifier prefix: $\psi' := \exists y_1 y_2 \forall x.\, \bigl(P(x) \rightarrow Q(y_1)\bigr) \wedge \bigl( \neg P(x) \rightarrow \neg Q(y_2) \bigr)$.
In this example we can not only transform nested quantification of separated variables into quantification that is not nested.
In addition, we can replace $\forall \exists$ alternations in exchange for $\exists \forall$ alternations, or vice versa.
In general, succinct representations of certain logical properties can be unfolded into more verbose ones that require a lower \emph{quantifier rank} (i.e.\ a lower nesting depth of quantifiers) or even use fewer quantifier alternations.
The sentence $\psi$ can, using a $\forall \exists$ alternation, represent the symmetry property of structures more succinctly than the sentence $\psi'$ can with an $\exists \forall$ quantifier alternation.
This is a key feature that we will observe for several of the extended fragments. 
We will also see that such succinctness gaps can become $k$-fold exponential, if we start from $k$ nested $\forall \exists$ alternations.


\medskip\noindent
The following list summarizes the main contributions of the present paper:
\begin{enumerate}[label=(\arabic{*}), ref=(\arabic{*})]
	\item \label{enum:ContributionsOfThesis:II} Eight novel decidable fragments of first-order logic are introduced that extend well-known decidable first-order fragments: 
		the \emph{separated Bernays--Sch\"onfinkel--Ramsey fragment~(SBSR)}, the \emph{separated Ackermann fragment~(SAF)}, the \emph{separated G\"odel--Kalm\'ar--Sch\"utte fragment (SGKS)}, the \emph{separated guarded fragment~(SGF)}, the \emph{separated loosely guarded fragment~(SLGF)}, the \emph{separated guarded-negation fragment (SGNFO)}, the \emph{separated two-variable frag\-ment~(SFO$^\mathit{2}$\!)}, and the \emph{separated fluted fragment~(SFL)}.
		Moreover, it is proved that the qualitative expressiveness of each extended fragment coincides with the respective original.
	
	\item \label{enum:ContributionsOfThesis:III} Significant gaps regarding succinctness are derived for several of the extended fragments: SBSR, SGKS, SGF, SLGF, SGNFO, \SFOtwo, cf.\ Table~\ref{table:SuccinctnessGapsSummary}.
		This evidently shows that several of the extended fragments constitute a substantial quantitative improvement regarding expressiveness compared to the original fragments.
\end{enumerate}		

\textbf{Note to editors and reviewers:}
The present paper is (except for Theorem~\ref{theorem:LengthSmallestGNFOsentences}) a condensation of material taken from the author's PhD thesis~\cite{Voigt2019PhDthesis}, mainly Chapter~3.
Moreover, results from previous conference publications are recapitulated~\cite{Voigt2016, Voigt2017a} (with citations).


\section{Preliminaries}
\label{section:PreliminariesPartI}

We mainly rely on the standard notions from first-order logic.
Nevertheless, we need to agree on some notation.

\subsubsection*{Syntax}

A \emph{vocabulary} $\Sigma$ comprises a finite set of predicate symbols and function symbols, each equipped with its \emph{arity}. 
A vocabulary $\Sigma$ is \emph{relational} if it exclusively contains predicate symbols.
We define \emph{$\Sigma$-terms} and $\Sigma$-formulas as usual, allowing the logical connectives $\neg, \wedge, \vee$ and the first-order quantifiers $\forall, \exists$.
The two additional connectives $\rightarrow, \leftrightarrow$ are used as shortcuts: $\varphi \rightarrow \psi$ abbreviates $(\neg \varphi) \vee \psi$ and $\varphi \leftrightarrow \psi$ abbreviates $(\varphi \rightarrow \psi) \wedge (\psi \rightarrow \varphi)$.
If not explicitly excluded, we allow equality.
A $\Sigma$-sentence is a \emph{closed} $\Sigma$-formula, i.e.\ one without free variables.
A $\Sigma$-formula / $\Sigma$-sentence is \emph{relational}, if $\Sigma$ is relational.
If there is no danger of confusion, we drop the explicit reference to $\Sigma$ and just speak of \emph{terms}, \emph{formulas}, and \emph{sentences}.
In order to save parentheses, we follow the convention that negation binds strongest, that conjunction binds stronger than disjunction, and that all of the aforementioned bind stronger than implication and equivalence. 
Equivalence, in turn, binds weaker than implication.
The scope of quantifiers will stretch as far to the right as admitted by parentheses.
In all formulas we tacitly assume, if not explicitly stated otherwise, that no variable occurs freely and bound at the same time and that all distinct occurrences of quantifiers bind distinct variables.
We use $\varphi(v_1, \ldots, v_m)$ to denote a formula $\varphi$ whose free first-order variables form a subset of $\{v_1, \ldots, v_m\}$.
The variables $v_1, \ldots, v_m$ are assumed to be pairwise distinct.
For any formula $\varphi$ we denote by $\vars(\varphi)$ the set of all variables occurring freely or bound in $\varphi$.

Given a set $\Phi$ of $\Sigma$-formulas, we call a $\Sigma$-formula $\varphi$ a \emph{Boolean combination of formulas from $\Phi$}, if $\varphi$ consist of formulas from $\Phi$, possibly connected via the Boolean connectives $\neg, \wedge, \vee, \rightarrow, \leftrightarrow$.
If we restrict the set of Boolean connectives even further to $\wedge, \vee$, we speak of a \emph{$\wedge$-$\vee$-combination of formulas from $\Phi$}.

A \emph{quantifier block} is a maximal sequence $\cQ \vv = \cQ v_1. \cQ v_2. \ldots \cQ v_n$ of quantifiers of the same kind occurring in a given formula. 
For convenience, we often identify tuples $\vv$ of variables with the set containing all the variables that occur in $\vv$.
We occasionally also use regular expressions to describe sequences of quantifiers.
For example, for any positive integer $k$ the expression $\exists^*\forall^k\exists\exists$ stands for the set of all prefixes of the form $\exists y_1 \ldots y_m \forall x_1 \ldots x_k \exists z_1 z_2$, where $m = 0$ is allowed. 

A formula is in \emph{prenex normal form} if all quantifiers are lined up in front of the formula, i.e.\ it has the shape $\cQ_1 v_1 \ldots \cQ_n v_n.\, \psi$ with quantifier-free $\psi$ and $\cQ_i \in \{\forall, \exists\}$.
A formula is in \emph{negation normal form} if it exclusively contains the connectives $\wedge, \vee, \neg$ and every negation sign occurs immediately in front of an atom; quantifiers are of course admitted.

We use a standard notion of term and formula length, denoted $\len(\varphi)$.
Notice that we have $\len(\varphi \rightarrow \psi) = \len(\neg \varphi \vee \psi)$ and $\len(\varphi \leftrightarrow \psi) = \len\bigl( (\varphi \rightarrow \psi) \wedge (\psi \rightarrow \varphi) \bigr)$.

\subsubsection*{Semantics} 

As usual, we interpret ($\Sigma$-)formulas with respect to \emph{($\Sigma$-)structures} $\cA$, consisting of a nonempty \emph{domain} $\fUA$ and \emph{interpretations} $f^\cA$ and $P^\cA$ of all function and predicate symbols in the underlying vocabulary.
Given a term $s$, a structure $\cA$, and a variable assignment $\beta$ over $\cA$'s domain, we denote the \emph{evaluation of $s$ under $\cA$ and $\beta$} by $\cA(\beta)(s)$.
It is defined as usual. 
We simply write $\cA(s)$ if $s$ is variable free.
We write $\cA, \beta \models \varphi$ if $\varphi$ is \emph{satisfied under $\cA$ and $\beta$} in the usual sense.
When there is no danger of confusion, we conveniently abbreviate expressions of the form $\cA, [v_1 \Mapsto \fa_1, \ldots,  v_m \Mapsto \fa_m] \models \varphi(v_1, \ldots, v_m)$ by $\cA \models \varphi(\fa_1, \ldots, \fa_m)$.
We write $\cA \models \varphi$ if $\cA, \beta \models \varphi$ holds for every variable assignment $\beta$ over $\cA$'s domain.
In such cases, we say that $\cA$ is a \emph{model} of $\varphi$.
For sentences $\varphi$ we say that $\cA$ \emph{satisfies} $\varphi$ if $\cA, \beta \models \varphi$ for any $\beta$.
Two sentences $\varphi$ and $\psi$ are considered \emph{equisatisfiable} if $\varphi$ has a model if and only if $\psi$ has one.

We also use the symbol $\models$ to denote \emph{semantic entailment} of two $\Sigma$-formulas, i.e.\ we have $\varphi \models \psi$ whenever for every $\Sigma$-structure $\cA$ and every variable assignment $\beta$, $\cA,\beta \models \varphi$ implies $\cA,\beta \models \psi$.
The symbol $\semequiv$ denotes \emph{semantic equivalence} of formulas, i.e.\ $\varphi \semequiv \psi$ holds whenever $\varphi \models \psi$ and $\psi \models \varphi$.

We will frequently use the fact that quantifiers can be \emph{shifted} in certain ways within formulas under preservation of the formula's semantics.
\begin{proposition}[Quantifier shifting]\label{lemma:Miniscoping}
	Let $\varphi, \psi, \chi$ be formulas, and assume that $x$ and $y$ do not occur freely in $\chi$.
	We have the following equivalences, where $\circ \in \{\wedge, \vee\}$:
	\begin{center}
	\begin{tabular}{cl@{\hspace{2ex}}c@{\hspace{2ex}}l@{\hspace{7ex}}cl@{\hspace{2ex}}c@{\hspace{2ex}}l}
			(i)	&	$\exists y.\, (\varphi \vee \psi)$		&$\semequiv$&		$(\exists y.\, \varphi) \vee (\exists y.\, \psi)$
		&	(ii)	&	$\forall x.\, (\varphi \wedge \psi)$	&$\semequiv$&		$(\forall x.\, \varphi) \wedge (\forall x.\, \psi)$ \\
			(iii)	&	$\exists y.\, (\varphi \circ \chi)$		&$\semequiv$&		$(\exists y.\, \varphi) \circ \chi$
		&	(iv)	&	$\forall x.\, (\varphi \circ \chi)$		&$\semequiv$&		$(\forall x.\, \varphi) \circ \chi$ \\
			(v)	&	$\exists y_1 \exists y_2 .\, \varphi$	&$\semequiv$&		$\exists y_2 \exists y_1 .\, \varphi$
		&	(vi)	&	$\forall x_1 \forall x_2 .\, \varphi$	&$\semequiv$&		$\forall x_2 \forall x_1 .\, \varphi$
	\end{tabular}
	\end{center}	
\end{proposition}

A structure $\cA$ is a \emph{substructure} of a structure $\cB$ (over the same vocabulary) if (1) $\fUA \subseteq \fUB$, (2) $c^\cA = c^\cB$ for every constant symbol $c$, (3) $P^\cA = P^\cB \cap \fUA^m$ for every $m$-ary predicate symbol $P$, and (4) $f^\cA(\va) = f^\cB(\va)$ for every $m$-ary function symbol $f$ and every $m$-tuple $\va \in \fUA^m$.
Given a structure $\cA$ and some subset $S$ of $\cA$'s domain, the \emph{substructure of $\cA$ induced by $S$} is the unique substructure $\cB$ of $\cA$ with the domain $\fUB := S$.
The following is a standard lemma (see, e.g., \cite{Ebbinghaus1994}, Lemma 5.7  in Chapter~III).
\begin{lemma}[Substructure Lemma]\label{lemma:SubstructureLemma}
	Let $\varphi$ be a first-order sentence without existential quantifiers and in which no universal quantifier lies within the scope of any negation sign --- we treat any subformula $\varphi_1 \rightarrow \varphi_2$ as abbreviation for $\neg \varphi_1 \vee \varphi_2$ and any subformula $\varphi_1 \leftrightarrow \varphi_2$ as abbreviation for $(\neg \varphi_1 \vee \varphi_2) \wedge (\varphi_1 \vee \neg \varphi_2)$ to account for implicit negation signs as well.
	Moreover, let $\cA$ be a substructure of $\cB$.
	If $\cB \models \varphi$, then $\cA \models \varphi$.
\end{lemma}

\subsubsection*{Additional Notation}

We use the notation $[k]$ to abbreviate the set $\{1, \ldots, k\}$ for any positive integer $k$.
The \emph{power set} of a set $S$, i.e.\ the set of all subsets of $S$, is denoted by $\cP(S)$.
The iterated application of $\cP$ is given by $\cP^0(S) := S$ and $\cP^{k+1}(S) := \cP^k (\cP(S))$ for $k \geq 0$.
Furthermore, we also define the \emph{tetration operation} inductively by $\twoup{0}{m} := m$ and $\twoup{k+1}{m} := 2^{\left(\twoup{k}{m}\right)}$.


\section{The Separated Fragment (SF)}
\label{section:SeparatedFragment}

Before we define the separated fragment, we first recall some details about the two well-known decidable fragments it extends.
%
The \emph{monadic first-order fragment (MFO)}\label{para:DefinitionMFO}
comprises all relational first-order sentences without equality that contain only unary predicate symbols.
When we refer to the \emph{monadic first-order fragment with equality}, we use the abbreviation \MFOeq.
The decidability of the respective satisfiability problems MFO-Sat and \MFOeq-Sat was proved in~\cite{Lowenheim1915, Skolem1919, Behmann1922}.
Several decades later, L\"ob~\cite{Lob1967} and Gurevich~\cite{Gurevich1969} extended the positive result to monadic first-order sentences with unary function symbols but without equality, the \emph{L\"ob--Gurevich fragment}.
MFO and \MFOeq\ possess the exponential model property: every satisfiable \MFOeq\ sentence $\varphi$ has a model whose domain size is at most exponential in $\varphi$'s length.
Moreover, satisfiability for MFO and \MFOeq\ is \NEXPTIME-complete, cf.~\cite{Borger1997}.

The \emph{Bernays--Sch\"onfinkel--Ramsey fragment (BSR)}
comprises all relational first-order sentences in prenex normal form with an $\exists^*\forall^*$ quantifier prefix and with equality.
Bernays and Sch\"onfinkel~\cite{Bernays1928} showed that satisfiability for the relational $\exists^*\forall^*$ prefix class without equality is decidable.
Following up, Ramsey~\cite{Ramsey1930} obtained a positive decidability result in the presence of equality. 
This extended class is known to posses a linear model property (cf.\ Proposition~\ref{proposition:SmallModelsBSR}), and BSR-Sat is known to be complete for \NEXPTIME.

Now we are ready for defining the \emph{separated fragment}, \emph{SF} for short.
Technically, it is defined as a class of prenex sentences, but this is not an essential property.
The defining principle of SF sentences is simply that co-occurrences of universally and existentially quantified variables in atoms are forbidden.
Existential variables quantified by \emph{leading} existential quantifiers are exempt from this rule.
We consider an existential quantifier \emph{leading} if it does not lie within the scope of any universal quantifier.

\begin{definition}[Separated fragment (SF), \cite{Voigt2016}]\label{definition:SeparatedFragment}
	The \emph{separated fragment (SF)} consists of all relational first-order sentences $\varphi$ with equality of the form 
		$\exists \vz\, \forall\vx_1 \exists\vy_1 \ldots \forall\vx_n \exists\vy_n.\, \psi$
	with quantifier-free $\psi$ in which the sets $\vx := \vx_1 \cup \ldots \cup \vx_n$ and $\vy := \vy_1 \cup \ldots \cup \vy_n$ are separated.
	The tuples $\vz$ and $\vy_n$ may be empty, i.e.\ the quantifier prefix does not have to start with an existential quantifier and it does not have to end with an existential quantifier either.
\end{definition}
Recall that $\vx$ and $\vy$ are separated in $\varphi$ if and only if for every atom $A$ occurring in $\varphi$ we either have $\vars(A) \cap \vx = \emptyset$  or $\vars(A) \cap \vy = \emptyset$.
Moreover, notice that the variables in $\vz$ are not subject to any restriction regarding their occurrences.

It is not hard to see that SF is a proper syntactic extension of BSR.
Clearly, the quantified variables in every BSR sentence $\varphi := \exists \vz \forall \vx.\, \psi$ with quantifier-free $\psi$ trivially satisfy the separateness conditions imposed by Definition~\ref{definition:SeparatedFragment}, as $\vy$ is empty.
Similarly, every monadic sentence without equality in prenex normal form trivially satisfies the conditions of Definition~\ref{definition:SeparatedFragment}, because any monadic atom contains at most one first-order variable.
Since any MFO sentence can easily be transformed into an equivalent sentence in prenex normal form, it is fair to say that SF also contains MFO.
On the other hand, it is an easy task to find sentences that belong to SF but neither to BSR nor to MFO, e.g.\ $\forall x_1 x_2 \exists y_1 y_2.\, P(x_1, x_2) \leftrightarrow Q(y_1, y_2)$.
\begin{proposition}\label{theorem:InclusionSF}
	SF properly contains BSR and MFO.
\end{proposition}

Another interesting question is whether \MFOeq\ is subsumed by SF.
For instance, the sentence $\forall x \exists y. \,x \approx y$ is in \MFOeq\ but violates the separateness conditions of SF.
Therefore, from the syntactic point of view, there are \MFOeq\ sentences whose variables are not sufficiently separated for SF.
However, the sentence $\forall x \exists y. \,x \approx y$ is equivalent to $\forall x.\, x \approx x$, which even belongs to BSR.
Similarly, we have the \MFOeq\ sentence $\forall x \exists y.\, x \not\approx y$, which is not in SF but equivalent to the BSR sentence $\exists y_1 y_2.\, y_1 \not\approx y_2$.
The following proposition witnesses that this is by no means a coincidence.
As one consequence, speaking in terms of expressiveness, \MFOeq\ is subsumed by BSR and, hence, also by SF.

\begin{proposition}\label{theorem:FromMonadicWithEqualityToBSR}
	For every \MFOeq\ sentence there is an equivalent BSR sentence.
\end{proposition}
The proof of this result can be found in~\cite{Voigt2019PhDthesis} (Theorem~3.1.5).
It is based on techniques described by Behmann~\cite{Behmann1922}
in the context of second-order quantifier elimination for the monadic second-order fragment --- see~\cite{Wernhard2015b}, Section~13.2, for a modern account.


\subsection{Translation of SF into BSR: Upper and Lower Bounds}
\label{section:TranslationSFintoBSR}

It was first proved in~\cite{Voigt2016} that every SF sentence can be transformed into an equivalent BSR sentence.
We will present such a translation for the more general case of SBSR sentences in Section~\ref{section:TranslationGBSRintoBSR} (Lemma~\ref{lemma:SBSRquantifierShifting} and Theorem~\ref{theorem:TransformationSBSRintoBSR}).
For SF we simply state the result here.
\begin{theorem}[\cite{Voigt2017a}, Lemma~12]\label{theorem:TranslationFromSFintoBSR}
	Every SF sentence $\varphi$ with $k$ $\forall\exists$ alternations is equivalent to some BSR sentence whose length is at most $k$-fold exponential in $\varphi$'s length.
\end{theorem}
It is worth noting that~\cite{Voigt2017a} uses a much more fine-grained measure than the number of $\forall\exists$ alternations. Even more details can be found in~\cite{Voigt2019PhDthesis}, Section~3.2.

Translations like the one from SF sentences to equivalent BSR sentences lead to large blowups in the worst case.
In general, the length of formulas has a significant effect on the size of smallest models.
For BSR this relation is linear.
\begin{proposition}[cf.\ Proposition~6.2.17 in~\cite{Borger1997}]\label{proposition:SmallModelsBSR}
	Let $\varphi := \exists \vz\, \forall \vx.\, \psi$ be a satisfiable BSR sentence with quantifier-free $\psi$, containing (at most) $k$ constant symbols.
	There is a model $\cA \models \varphi$ such that $|\fUA| \leq \max \bigl( |\vz| + k, 1 \bigr)$.
\end{proposition}

It was proved in~\cite{Voigt2017a} (Lemma~24) that for satisfiable SF sentences $\varphi$ the size of smallest models cannot be bounded by any tower of exponents $2^{{\vdots^{2^{\len(\varphi)}}}}$ of a fixed height.
In other words, the asymptotic growth of the size of smallest models is non-elementary in the length of the regarded SF sentence.

The analysis conducted in~\cite{Voigt2017a} yields matching upper and lower bounds that are not formulated in terms of the number of quantifier alternations.
\marginpar{\label{margin:fineGrainedMeasure}}
Rather, the upper bounds are based on a measure of how much quantified variables interact in atoms.
The motivation for this fine-grained analysis is that the upper bound is intended to also give a tight estimate for MFO sentences.
It is well known that the size of small models for satisfiable MFO sentences is exponential in their length no matter how many quantifier alternations are present (cf.\ Proposition~6.2.1 in~\cite{Borger1997}).
Due to space limitations, we will not elaborate any further on this topic.
The interested reader will find more details in~\cite{Voigt2017a} and in \cite{Voigt2019PhDthesis}, Sections~3.2 and~3.5.

Put together, Theorem~\ref{theorem:TranslationFromSFintoBSR} and Proposition~\ref{proposition:SmallModelsBSR} immediately entail the following small model property for SF.

\begin{proposition}[Small model property for SF]\label{theorem:SmallModelPropertyForSF}
	Every satisfiable SF sentence~$\varphi$ with $k \geq 1$ $\forall\exists$ alternations has a model with at most 
	$\len(\varphi) + k \cdot \bigl( \len(\varphi) \bigr)^2 \cdot \bigl(\twoup{k}{\len(\varphi)}\bigr)^{k}$
	domain elements.
\end{proposition}

It is well known that a small model property leads to decidability of the associated satisfiability problem (see, e.g.\ Proposition~6.0.4 in~\cite{Borger1997}).
Since BSR enjoys a linear model property, even if constant symbols are allowed in the syntax, the separated fragment immediately inherits this property.
Hence, the satisfiability problem for SF (\emph{SF-Sat}) is decidable.

\begin{proposition}\label{theorem:DecidabilitySF}
	SF-Sat is decidable, even if we allow constant symbols in SF sentences.
\end{proposition}

The unrestricted presence of function symbols in SF would lead to an undecidable satisfiability problem.
Nevertheless, SF could easily be extended so far that it also subsumes the L\"ob-Gurevich fragment (see the beginning of Section~\ref{section:SeparatedFragment} for the definition) while retaining decidability (see~\cite{Voigt2019PhDthesis}, Section~3.14.1).


Next, we complement the upper bound from Theorem~\ref{theorem:TranslationFromSFintoBSR} with a corresponding non-elementary lower bound.
\begin{theorem}[\cite{Voigt2017a}, Theorem~16]\label{theorem:LengthSmallestBSRsentences}
	There is a class of satisfiable SF sentences such that for every positive integer $n$ the class contains a sentence $\varphi$ with $n$ $\forall \exists$ quantifier alternations and with a length polynomial in $n$ for which any equivalent BSR sentence contains at least $\sum_{k=1}^{n}\twoup{k}{n}$ leading existential quantifiers.
\end{theorem}
\begin{proof}
	Let $n \geq 1$ be some positive integer.
	Consider the following first-order sentence in which the sets $\{x_1, \ldots, x_n\}$ and $\{y_1, \ldots, y_n\}$ are separated:
	\begin{align*}
		\varphi := \forall &x_n \exists y_n \ldots \forall x_1 \exists y_1. \bigwedge_{i=1}^{4n} \bigl( P_i(x_1, \ldots, x_n) \leftrightarrow Q_i(y_1, \ldots, y_n) \bigr) ~.
	\end{align*}	
	Notice the orientation of the indices in the quantifier prefix!
	Moreover,  recall that $[m]$ with $m \geq 1$ denotes the set $\{ 1, \ldots, m \}$.
	In order to construct a particular model of $\varphi$, we inductively define the following sets:
		$\cS_1 := \bigl\{ S \in \cP([4n]) \bigm| |S| = 2n \bigr\}$, $\cS_{k+1} := \bigl\{ S \in \cP(\cS_k) \bigm| |S| = \tfrac{1}{2} \cdot |\cS_k| \bigr\}$ for every $k > 1$.
	Then, we observe that
		\begin{description}
			\item $|\cS_1| \;=\; {{4n} \choose {2n}} \;\geq\; \bigl( \frac{4n}{2n} \bigr)^{2n} \;=\; 2^{2n}$,
			\item $|\cS_2| \;=\; {{|\cS_1|} \choose {|\cS_1|/2}} \;\geq\; \bigl( \frac{|\cS_1|}{|\cS_1|/2} \bigr)^{|\cS_1|/2} \;=\; 2^{|\cS_1|/2} \;\geq\;  2^{2^{2n-1}}$,
			\item[] $\strut\qquad\vdots$
			\item $|\cS_n| \;=\; {{|\cS_{n-1}|} \choose {|\cS_{n-1}|/2}} \;\geq\; 2^{|\cS_{n-1}|/2} \;\geq\; 2^{2^{2^{\vdots^{2^{2n-1}-1}}-1}} \;\geq\; \twoup{n}{2n-(n-1)} \;=\; \twoup{n}{n+1}$,
		\end{description}	
	where the inequality ${n \choose k} \geq (n/k)^{k}$ can be found in~\cite{Cormen2001} (page~1097), for example.
	
	Having the sets $\cS_k$, we now define the structure $\cA$ as follows:
		\begin{description}
			\item $\fUA := \bigl\{ \fa^{(k)}_{S}, \fb^{(k)}_{S} \bigm| 1 \leq k \leq n \text{ and } S \in \cS_k \bigr\}$, 
			\item $P_i^\cA := \bigl\{ \<\fa^{(1)}_{S_1}, \ldots, \fa^{(n)}_{S_n}\> \in \fUA^n \bigm| i \in S_1 \in S_2 \in \ldots \in S_n \bigr\}$ for $i = 1, \ldots, 4n$, and
			\item $Q_i^\cA := \bigl\{ \<\fb^{(1)}_{S_1}, \ldots, \fb^{(n)}_{S_n}\> \in \fUA^n \bigm| i \in S_1 \in S_2 \in \ldots \in S_n \bigr\}$ for $i = 1, \ldots, 4n$.
		\end{description}		
	For any choice of $S_1, \ldots, S_n$  and every $i$, $1\leq i\leq 4n$, we then have 
		\begin{align*}
			\cA \models P_i\bigl( \fa^{(1)}_{S_1}, \ldots, \fa^{(n)}_{S_n} \bigr) \leftrightarrow Q_i\bigl( \fb^{(1)}_{S_1}, \ldots, \fb^{(n)}_{S_n} \bigr) ~.
		\end{align*}	
	For any other choice of tuples $\< \fc_1, \ldots, \fc_n \>$, i.e.\ if there do not exist sets $S_1 \in \cS_1, \ldots, S_n \in \cS_n$ such that $\< \fc_1, \ldots, \fc_n \>$ equals $\< \fa^{(1)}_{S_1}, \ldots, \fa^{(n)}_{S_n} \>$ or $\< \fb^{(1)}_{S_1}, \ldots, \fb^{(n)}_{S_n} \>$, we observe $\cA \not\models P_i(\fc_1, \ldots, \fc_n)$ and $\cA \not\models Q_i(\fc_1, \ldots, \fc_n)$ for every $i$.
	Hence, 
		\begin{align*}
			\cA \models \bigwedge_{i = 1}^{4n} P_i(\fc_1, \ldots, \fc_n) \leftrightarrow Q_i(\fc_1, \ldots, \fc_n) ~.
		\end{align*}
	Consequently, $\cA$ is a model of $\varphi$.
		
	Consider the following simple two-player game with Players $\fA$ and $\fB$.
	In the first round $\fA$ moves first by picking some domain element $\fa^{(n)}_{S_{\fA, n}}$ for some set $S_{\fA,n} \in \cS_n$. 
	Player $\fB$ answers by picking a domain element $\fb^{(n)}_{S_{\fB, n}}$ for some set $S_{\fB,n} \in \cS_n$. 
	The game continues for $n-1$ more rounds, where in every round Player $\fA$ picks a domain element $\fa^{(j)}_{S_{\fA, j}}$ with $S_{\fA,j} \in S_{\fA, j+1}$ and $\fB$ answers by picking some $\fb^{(j)}_{S_{\fB,j}} \in S_{\fB,j+1}$.	
	Hence, in the last round the chosen domain elements $\fa^{(1)}_{S_{\fA, 1}}$ and $\fb^{(1)}_{S_{\fB, 1}}$ are such that $S_{\fA, 1}$ and $S_{\fB, 1}$ are both nonempty subsets of $[4n]$.
	Player $\fA$ wins if and only if 
		{$\cA \not\models P_i\bigl( \fa^{(1)}_{S_1}, \ldots, \fa^{(n)}_{S_n} \bigr) \leftrightarrow Q_i\bigl( \fb^{(1)}_{S_1}, \ldots, \fb^{(n)}_{S_n} \bigr)$}
	for some $i \in [4n]$,
	and Player $\fB$ wins if and only if 
		{$\cA \models P_i\bigl( \fa^{(1)}_{S_1}, \ldots, \fa^{(n)}_{S_n} \bigr) \leftrightarrow Q_i\bigl( \fb^{(1)}_{S_1}, \ldots, \fb^{(n)}_{S_n} \bigr)$}
	for every $i \in [4n]$.
	Since $\cA$ is a model of $\varphi$, there must exist a winning strategy for $\fB$.
	
	\begin{description}
		\item \underline{Claim I:}
			There is exactly one winning strategy for $\fB$, namely, for every $j = n, \ldots, 1$ Player $\fB$ picks the element $\fb^{(j)}_{S_{\fA,j}}$ in round $n-j+1$, i.e.\ for every $j$ we have $S_{\fB,j} = S_{\fA,j}$.
				
		\item \underline{Proof:}
			It is easy to see that the described strategy is a winning strategy for $\fB$.
			
			Assume $\fB$ deviates from this strategy.
			More precisely, suppose there exists some $j_*$, $1 \leq j_* \leq n$, such that $\fB$ did not adhere to the described strategy in the $(n-j_*+1)$st round, i.e.\ $S_{\fB, j_*} \neq S_{\fA, j_*}$.
			It can be shown by induction on $j_*$ that $\fA$ has a winning strategy from this deviation point on.
		
			\strut\hfill$\Diamond$		
	\end{description}	
	Claim~I would still hold true if we allowed $\fB$ to freely pick any element of the domain $\fUA$ at every round. 
	The reason is that for any choice of elements $\fa^{(n)}_{S_{\fA,n}}, \ldots, \fa^{(1)}_{S_{\fA,1}}$ made by $\fA$ with $S_{\fA,1} \in \ldots \in S_{\fA,n} \in \cS_n$ we know that $S_{\fA,1}$ is nonempty.
	Hence, we can always find some $i_* \in S_{\fA,1}$ such that $\cA \models P_{i_*}\bigl( \fa^{(1)}_{S_{\fA,1}}, \ldots, \fa^{(n)}_{S_{\fA,n}} \bigr)$.
	On the other hand, for any sequence $\fc_n, \ldots, \fc_1$ picked by $\fB$ that does not comply with the rules of the described game, we have $\cA \not\models Q_{i_*} ( \fc_1, \ldots, \fc_n )$.
	This result leads to the following observation.  
	\begin{description}
		\item \underline{Claim II:} 
				For any of the $\fb^{(k)}_{S}$ the substructure of $\cA$ induced by the domain $\fUA \setminus \{ \fb^{(k)}_{S} \}$ does not satisfy $\varphi$.
		\item \underline{Proof:} 
				The reason is simply that in this case player $\fA$ can always prevent $\fB$ from reaching a state of the game where $\fB$ can apply the described winning strategy.
			\hfill$\Diamond$
	\end{description}
	
	We have already analyzed the size of the sets $\cS_k$. Due to the observed lower bounds, we know that $\fUA$ contains at least
		$\sum_{k=1}^{n}\twoup{k}{n}$
	elements of the form $\fb^{(k)}_{S}$.
	
	Next, we argue that any $\exists^* \forall^*$-sentence $\varphi_*$ that is semantically equivalent to $\varphi$ must contain at least $\sum_{k=1}^{n}\twoup{k}{n}$ leading existential quantifiers.
	 Let \\
		 \centerline{$\varphi_* := \exists y_1 \ldots y_m \forall x_1 \ldots x_\ell.\, \chi_*(y_1, \ldots, y_m, x_1, \ldots, x_n)$}
	 with quanti\-fier-free $\chi_*$ be a sentence with minimal $m$ that is semantically equivalent to $\varphi$. 
	 Since $\cA$ is also a model of $\varphi_*$, we know that there is a sequence of elements $\fc_1, \ldots, \fc_m$ taken from the domain $\fUA$ such that $\cA \models \forall x_1 \ldots x_\ell.\, \chi_*(\fc_1, \ldots, \fc_m, x_1, \ldots, x_n)$.
	Consequently, we can extend $\cA$ to a model $\cA_*$ (over the same domain) of the Skolemized sentence $\varphi_{\text{Sk}} := \forall x_1 \ldots x_\ell.\, \chi_*\subst{y_1/d_1, \ldots, y_m/d_m}$ by adding $d_j^{\cA_*} := \fc_j$ for $j = 1, \ldots, m$. On the other hand, every model of $\varphi_{\text{Sk}}$ is also a model of $\varphi_*$.
	The vocabulary underlying $\varphi_{\text{Sk}}$ comprises exactly the constant symbols $d_1, \ldots, d_m$ and does not contain any other function symbols. 
	Suppose $m < \sum_{k=1}^{n}\twoup{k}{n}$. Hence, there is some $\fb^{(k)}_{S}$ with $S \in \cS_k$ such that for none of the $d_j$ we have $d_j^{\cA_*} = \fb^{(k)}_{S}$. 
	Then, by the Substructure Lemma, the following substructure $\cB$ of $\cA_*$ constitutes a model of $\varphi_{Sk}$: $\fUB := \fUA_* \setminus \{\fb^{(k)}_{S}\}$, $P_i^\cB := P_i^{\cA_*} \cap \fUB^n = P_i^{\cA_*}$ and $Q_i^\cB := Q_i^{\cA_*} \cap \fUB^n$ for every $i$, and $d_j^\cB := d_j^{\cA_*}$ for every $j$.
	But then $\cB$ must also be a model of both $\varphi_*$ and $\varphi$, since every model of $\varphi_{\text{Sk}}$ is a model of $\varphi_*$ which we, in turn, assume to be equivalent to $\varphi$.
	This contradicts Claim~II, and thus we must have $m \geq \sum_{k=1}^{n}\twoup{k}{n}$.
\end{proof}

Theorem~\ref{theorem:LengthSmallestBSRsentences} entails that there is no elementary upper bound on the length of the BSR sentences that result from any equivalence-preserving transformation of SF sentences into BSR. 
On the other hand, there is an elementary upper bound, if we only consider SF sentences with a bounded number of quantifier alternations (cf.\ Theorem~\ref{theorem:TransformationSBSRintoBSR}).
A special case of Theorem~\ref{theorem:LengthSmallestBSRsentences} highlights the difference in succinctness between BSR and MFO.
The following proposition states that, in the worst case, there is an unavoidable exponential gap between the two fragments.
\begin{proposition}\label{proposition:LengthSmallestBSRsentencesFromMFOSentences}
	There is a class of MFO sentences such that for every positive integer $n$ the class contains a sentence $\varphi$ of a length polynomial in $n$ for which any equivalent BSR sentence contains at least $2^n$ leading existential quantifiers.
\end{proposition}
One witnessing class of sentences comprises 
	$\forall x \exists y.\, \bigwedge_{i=1}^{2n} \bigl( P_i(x) \leftrightarrow Q_i(y) \bigr)$
for $n \geq 1$.


\subsection{Expressiveness of SF}
\label{section:ExpressivenessOfSF}

We have already seen that SF is at least as expressive as BSR, MFO, and \MFOeq\ (cf.\ Propositions~\ref{theorem:InclusionSF} and~\ref{theorem:FromMonadicWithEqualityToBSR}).
Moreover, every $\wedge$-$\vee$-combination of sentences from BSR and/or \MFOeq\ is equivalent to some SF sentence.
On the other hand, Theorem~\ref{theorem:LengthSmallestBSRsentences} shows that SF sentences can be considerably more succinct than their BSR equivalents.


Whenever it is possible to restrict our attention to bounded-size models, then SF is as expressive as full (relational) first-order logic.
This alone is not a very interesting result, as already the \emph{existential fragment} of relational first-order logic, i.e.\ the class of relational $\exists^*$ prefix sentences, possesses this property (universal quantification can be replaced by finite conjunctions).
What makes the case of SF special is that the incurred blowup in formula length is not linear in the bound but may be significantly lower.

Due to space limitations, we can only state the result here.
The interested reader will find the full details in~\cite{Voigt2019PhDthesis}, Section~3.3.3.
Abstractly speaking, when restricted to models of the size $\twoup{n}{m}$, any first-order sentence can be translated into an equisatisfiable SF sentence whose length is polynomial in $n$, $m$, and the length of the original sentence.
\begin{proposition}\label{proposition:TranslationFOLwithBoundedModelsIntoSF}
	Let $m \geq 1$ and $n \geq 2$ be two integers.
	There exists an efficiently computable SF sentence $\chi_{m,n}$ and an effective translation $T_{m,n}$ mapping any relational first-order sentence $\varphi$ to some SF sentence $T_{m,n}(\varphi)$ that satisfies the following properties. \\
			(a) $\cA \models \chi_{m,n}$ entails $|\fUA| \leq \twoup{n}{m}$. \\
			(b) $\chi_{m,n} \wedge \varphi$ is equivalent to $\chi_{m,n} \wedge T_{m,n}(\varphi)$. \\
			(c) The formula length of $T_{m,n}(\varphi)$ is at most $p(m,n) \cdot \len(\varphi)$ for some polynomial $p(m,n)$. \\
			(d) $T_{m,n}(\varphi)$ is computable in time $q(m, n, \len(\varphi))$ for some polynomial $q(m, n,k)$.
\end{proposition}

Proposition~\ref{proposition:TranslationFOLwithBoundedModelsIntoSF} entails that SF-Sat is computationally at least as hard as the satisfiability problem for any first-order fragment that enjoys a small model property with an elementary upper bound on the size of small models.
For instance, the \emph{Ackermann fragment} (cf.\ Section~\ref{section:GeneralizedAF}), the \emph{G\"odel--Kalm\'ar--Sch\"utte fragment} (cf.\ Section~\ref{section:generalizedGKS}), the \emph{guarded fragment} (cf.\ Section~\ref{section:SeparatenessAndGuardedFormulas}), the \emph{guarded-negation fragment} (cf.\ Section~\ref{section:SeparatenessAndGuardedNegationFormulas}), and the \emph{two-variable fragment} (cf.\ Section~\ref{section:FiniteVariableLogicAndSeparateness}) fall into this category.
Even the satisfiability problem for first-order fragments enjoying a small model property with bounds $\twoup{\lceil c \cdot \len(\varphi) \rceil}{\lceil d \cdot \len(\varphi) \rceil}$ for constants $c,d$, such as the \emph{fluted fragment} (cf.\ Section~\ref{section:SeparatenessAndFlutedFormulas}), can be polynomially reduced to SF-Sat.
Although this latter observation already yields a non-elementary lower bound regarding the computational complexity of SF-Sat, a more accurate lower bound was presented in~\cite{Voigt2017a} by encoding \emph{bounded domino problems} (see also~\cite{Voigt2019PhDthesis}, Section~5.3).

\begin{proposition}\label{theorem:TranslationFMPintoSF}
	Let $\cC$ be any class of relational first-order sentences for which we know constants $c,d \geq 1$ such that every satisfiable $\varphi$ in $\cC$ has a model whose domain size is bounded by $\twoup{\lceil c \cdot \len(\varphi) \rceil}{\lceil d \cdot \len(\varphi) \rceil}$. 
	Then, $\cC$-Sat is polynomial-time reducible to SF-Sat.
\end{proposition}
The proof of this result employs the translation $T_{m,n}$ from Proposition~\ref{proposition:TranslationFOLwithBoundedModelsIntoSF} for $m := \lceil d \cdot \len(\varphi) \rceil$ and $n := \lceil c \cdot \len(\varphi) \rceil$ for any given $\cC$-sentence $\varphi$.

Employing the ideas underlying Proposition~\ref{proposition:TranslationFOLwithBoundedModelsIntoSF}, one can also derive other lower bounds, similar to Theorem~\ref{theorem:LengthSmallestBSRsentences}, regarding the length of sentences that are equivalent to SF sentences but adhere to certain syntactic restrictions.
A classical result by Gaifman~\cite{Gaifman1981} states that every first-order formula is equivalent to some formula that is \emph{local} in a certain sense (see, e.g.~\cite{Libkin2004}).
It has been shown~\cite{Dawar2007a} that there is a non-elementary gap between the length of first-order sentences and their shortest equivalents in \emph{Gaifman normal form}.
Using a proof approach inspired by the one in~\cite{Dawar2007a}, one can show that this gap also applies to SF.
\begin{proposition}\label{theorem:SFtoGaifmanLowerBound}
	There is some vocabulary $\Sigma$ and some polynomial $p(h)$ such that for every $h \geq 0$ there is an SF $\Sigma$-sentence $\varphi_{\SF, h}$ of length $p(h)$ satisfying the following property.
	Every first-order $\Sigma$-sentence $\psi$ in Gaifman normal form that is equivalent to $\varphi_{\SF, h}$ has length at least $\twoup{h}{1}$.
\end{proposition}
Again, the interested reader will find the full details in~\cite{Voigt2019PhDthesis}, Section~3.3.3.


\section{The Separated Bernays--Sch\"onfinkel--Ramsey Fragment}
\label{section:GeneralizedBSR}

In this section we introduce the separated extension of the Bernays--Sch\"onfinkel--Ramsey fragment (BSR).
The separated fragment can be conceived as an intermediate step, as it lies properly between BSR and the \emph{separated Bernays--Sch\"onfinkel--Ramsey fragment (SBSR)}.

Recall that SF contains relational sentences $\exists \vz \forall \vx_1 \exists \vy_1 \ldots \forall \vx_n \exists \vy_n.\, \psi$ in which the sets $\vx_1 \cup \ldots \cup \vx_n$ and $\vy_1 \cup \ldots \cup \vy_n$ are separated.
The exemption of the leading existential quantifier block from the separateness conditions may lead to \emph{benign} co-occurrences of existentially and universally quantified variables in atoms, which do not pose an obstacle to deciding the satisfiability problem.
Extrapolating this emerging pattern of benign co-occurrences leads to the definition of the \emph{separated Bernays--Sch\"onfinkel--Ramsey fragment}. 
Intuitively speaking, an SBSR sentence $\varphi$ has the form $\forall \vx_1 \exists \vy_1 \ldots \forall \vx_n \exists \vy_n.\, \psi$ with quantifier-free $\psi$ that may contain equality and possesses the following properties.
Each atom in $\varphi$ only contains variables from an $\exists^* \forall^*$ subsequence of $\varphi$'s quantifier prefix.
If two atoms share a universally quantified variable, the same quantifier subsequence is used for both atoms.

\begin{definition}[Separated Bernays--Sch\"onfinkel--Ramsey fragment (SBSR)]\label{definition:SBSRaxiomatic:alternative}	
	Let $Y, X_1,$ $X_2, X_3, \ldots$ be pairwise-disjoint, countable sets of first-order variables.
	The \emph{separated Ber\-nays--Sch\"onfinkel--Ramsey fragment (SBSR)} comprises all relational first-order sentences with equality having the shape
		$\varphi := \forall \vx_1 \exists \vy_1 \ldots \forall \vx_n \exists \vy_n.\, \psi$ 
	with quantifier-free $\psi$ where
	(a) $\vx_i \subseteq X_1 \cup \ldots \cup X_i$ and $\vy_i \subseteq Y$ for every $i$, and
	(b) for every atom $A$ in $\varphi$ there is some $i$, $0 \leq i \leq n$, such that
		$\vars(A) \subseteq \vy_1 \cup \ldots \cup \vy_i \cup X_{i+1}$.
	(Notice that this entails separateness of $X_1, \ldots, X_n$, and for every $A$ there is some $i$ with $\vars(A) \subseteq \vy_1 \cup \ldots \cup \vy_i \cup \vx_{i+1} \cup \ldots \cup \vx_n$.)
\end{definition}

Suppose we are given an arbitrary first-order sentence for which we intend to check membership in SBSR without knowing a-priori how the set of occurring variables is to be partitioned.
Checking for the existence of such a suitable partition can be done deterministically in polynomial time.
A corresponding procedure can be based on a graph algorithm (see~\cite{Voigt2019PhDthesis}, Theorem~3.4.3).

The main difference between SF and SBSR lies in the concession policy regarding benign co-occurrences of existential and universal variables.
The following example gives a first impression of SBSR sentences and how they can be translated into BSR. 

\begin{example}\label{example:FirstSBSRsentence}
	Consider the first-order sentence 
		$\varphi := \exists u \forall x \exists y \forall z.\, \bigl(P(u,z) \wedge Q(u,x)\bigr) \vee \bigl(P(y,z) \wedge Q(u,y) \bigr)$.
	It belongs to SBSR, as witnessed by the following partition of its variables:
		$Y = \{u,y\}$,
		$X_1 = \{x\}$,
		$X_2 = \{z\}$,
		$X_3 = \emptyset$.
	Obviously, $\varphi$ neither belongs to BSR nor to SF.
	As universal quantification does not distribute over disjunction, the quantifier $\forall z$ cannot be shifted inwards with the standard quantifier shifting rules from Lemma~\ref{lemma:Miniscoping} alone.
	However, it turns out that the transformation methods that we have first met in Section~\ref{section:Introduction} also facilitate translations of SBSR sentences into BSR sentences.
	We will elaborate on this in Section~\ref{section:TranslationGBSRintoBSR}.
	For $\varphi$ we get the equivalent BSR sentence
	\begin{align*}
		\varphi' := 
			\exists u y \forall x z v.\,	
			&\bigl(\bigl(P(u,x) \vee P(y,x)\bigr) \wedge P(u,x) \wedge Q(u,x) \bigr) \\
			&\vee \bigl(\bigl(P(u,z) \vee P(y,z)\bigr) \wedge Q(u,y) \wedge Q(u,z) \bigr) \\
			&\vee \bigl(\bigl( P(u,v) \vee P(y,v)\bigr) \wedge Q(u,y) \wedge P(y,v) \bigr) ~.
	\end{align*}	
\end{example}

We have advertised SBSR as an extension of SF and, hence, also of BSR and MFO.
Indeed, given an SF sentence $\chi := \exists \vv_1 \forall \vu_2 \exists \vv_2 \ldots \forall \vu_n \exists \vv_n.\, \chi'$, we can partition the set of $\chi$'s variables into 
	$Y := \vv_1 \cup \ldots \cup \vv_n$ and $X_1 := \emptyset$,  $X_2 := \vu_2 \cup \ldots \cup \vu_n$, and $X_j := \emptyset$ for $j \geq 3$.
We then observe for every atom $A$ in $\chi$ that either $\vars(A) \subseteq Y$ or $\vars(A) \subseteq \vv_1 \cup X_2$. 
This partition complies with Definition~\ref{definition:SBSRaxiomatic:alternative}.
%
On the other hand, the sentence $\varphi$ from Example~\ref{example:FirstSBSRsentence} belongs to SBSR but not to SF. 
Hence, SBSR is a proper extension of SF.
\begin{proposition}\label{theorem:InclusionSBSR}
	SBSR properly contains SF and, hence, BSR and MFO.
\end{proposition}
By Proposition~\ref{theorem:FromMonadicWithEqualityToBSR}, SBSR in addition semantically subsumes \MFOeq, like SF does.


\subsection{Translation of SBSR into BSR}
\label{section:TranslationGBSRintoBSR}

We have already advertised several times that there is an effective equivalence-preserving translation from SBSR into BSR.
It is mainly based on the standard axioms of Boolean algebra and quantifier shifting.
Roughly speaking, we iteratively (re-)transform a given SBSR sentence into particular syntactic shapes and apply quantifier shifting so that we eventually obtain a formula in which no existential quantifier occurs within the scope of any universal quantifier.
We then shift all quantifiers outwards again --- existential quantifiers first ---, renaming bound variables as necessary. 
The final result is a BSR sentence.
Since SBSR contains SF, Theorem~\ref{theorem:LengthSmallestBSRsentences} entails that there is no elementary upper bound on the blowup that we incur in any equivalence-preserving translation from SBSR into BSR.
On the other hand, the blowup for SBSR-BSR translations is not significantly worse than in the case of SF-BSR translations.
It seems that in this sense SBSR does not offer much more succintness when describing logical properties than SF does.

The following lemma formulates the technical invariants for transforming SBSR sentences into BSR sentences via shifting inwards quantifier blocks one after the other.
\begin{lemma}\label{lemma:SBSRquantifierShifting}
	Let $Y, X_1,$ $X_2, \ldots$ be pairwise-disjoint, countable sets of first-order variables.
	Fix a quantifier sequence $\forall \vx_1 \exists \vy_1 \ldots \forall \vx_n \exists \vy_n$ (without double occurrences of variables) such that
		$\vx_i \subseteq X_1 \cup \ldots \cup X_i$ and $\vy_i \subseteq Y$ for every $i$.
	Moreover, fix a set of atoms $\At$ over the variables in $\vx_1, \vy_1, \ldots, \vx_n, \vy_n$ such that
		for every atom $A \in \At$ there is some $i$, $0 \leq i \leq n$, with
			$\vars(A) \subseteq \vy_1 \cup \ldots \cup \vy_i \cup X_{i+1}$.
	
	Let $\varphi := \cQ v.\, \psi$ be any first-order formula over the atoms in $\At$ such that
	\begin{enumerate}[label=(\roman{*}), ref=(\roman{*})]
		\item \label{enum:SBSRtranslation:I} 
			$\varphi$ is in negation normal form, no variable is bound by two distinct quantifier occurrences, no variable occurs freely and bound,
		\item \label{enum:SBSRtranslation:II} 
			any sequence of nested quantifiers $\cQ_1 u_1 \ldots \cQ_k u_k$ occurring in $\varphi$ (from left to right) is a subseqence of $\forall \vx_1 \exists \vy_1 \ldots \forall \vx_n \exists \vy_n$,
		\item \label{enum:SBSRtranslation:III} 
			for every subformula $\forall x.\, \chi$ of $\psi$ we have $\vars(\chi) \subseteq \vy_1 \cup \ldots \cup \vy_{j_*-1} \cup X_{j_*}$ for some $j_*$,
			
		\item \label{enum:SBSRtranslation:IV} 
			for every subformula $\exists y.\, \chi$ in $\psi$ with $y \in \vy_i$ we have that $\chi$ does not contain any occurrences of variables from $\vx_1 \cup \ldots \cup \vx_i$, and
		\item \label{enum:SBSRtranslation:V} 
			any sequence of nested quantifiers $\cQ_1 u_1 \ldots \cQ_k u_k$ occurring in $\psi$ (from left to right) is a subseqence of $\exists \vy_1 \ldots \exists \vy_i \forall \vx_{i+1} \ldots \forall \vx_n$ for some $i$.
	\end{enumerate}	
	Then, $\varphi$ can be transformed into an equivalent formula~$\varphi'$ such that Conditions~\ref{enum:SBSRtranslation:I} and~\ref{enum:SBSRtranslation:II} still apply to $\varphi'$ and Conditions~\ref{enum:SBSRtranslation:III} to~\ref{enum:SBSRtranslation:V} apply to $\varphi'$ instead of $\psi$ only.
\end{lemma}
\begin{proof}
	First of all, we remove from $\varphi$ (and any later formulas) any quantifiers that bind variables which do not occur in any atom.
	A \emph{basic formula} is any atom and any subformula $(\cQ u. \ldots)$ in $\psi$ that does not lie within the scope of any quantifier in $\psi$.
				
				Suppose $\cQ v$ is a universal quantifier.
				Then, there are indices $i_*, j_*$ such that $v \in \vx_{i_*} \cap X_{j_*}$ and $j_* \leq i_*$.
				We transform $\psi$ into an equivalent conjunction of disjunctions of negated or non-negated basic formulas.
				This is always possible.
				Due to our assumptions, the constituents of the $k$-th disjunction can be grouped into $i_*+1$ parts:
				$\chi_{k,1} \vee \ldots \vee \chi_{k,i_*} \vee \eta_{k}$ with $\vars(\chi_{k,\ell}) \subseteq \vy_1 \cup \ldots \cup \vy_\ell \cup X_{\ell+1}$ for $0 \leq \ell \leq i_*-1$ and $\vars(\eta_{k}) \subseteq \vy_1 \cup \ldots \cup \vy_n \cup X_{i_* + 1} \cup \ldots \cup X_n$.
				Due to Conditions~\ref{enum:SBSRtranslation:II} and~\ref{enum:SBSRtranslation:IV}, $\eta_k$ can be defined so that it contains all basic formulas from $\psi$ that have the form $\exists y.\, \eta'$ for any $y \subseteq \vy_{i_*} \cup \ldots \cup \vy_n$.
				Hence, $\varphi$ is equivalent to some formula of the form
					\begin{align*} 
						\forall v.\, \bigwedge_k \chi_{k,1}( X_1 ) \;\vee\; \chi_{k,2}( \vy_1, X_2 ) \;\vee \ldots \;&\vee\; \chi_{k,j_*}( \vy_1, \ldots, \vy_{j_*-1}, X_{j_*} ) \;\vee \ldots \\ 
							&\vee\; \chi_{k,i_*}( \vy_1, \ldots, \vy_{i_*-1}, X_{i_*} )  \;\vee\; \eta_{k} ( \vy_1, \ldots, \vy_{i_*-1} ) ~.
					\end{align*}	
				We shift the universal quantifier $\forall v$ inwards so that it only binds the (sub-)con\-junc\-tions $ \chi_{k,j_*}(\vy_1, \ldots,$ $\vy_{j_*-1}, \vx_{j_*})$.
				The resulting formula 
					\begin{align*} 
						\bigwedge_k \chi_{k,1}( X_1 ) \;\vee\; \chi_{k,2}( \vy_1, X_2 ) \;\vee \ldots \;&\vee\; \bigl( \forall v.\, \chi_{k,j_*}( \vy_1, \ldots, \vy_{j_*-1}, X_{j_*} ) \bigr) \;\vee \ldots \\[-2ex] 
							&\vee\; \chi_{k,i_*}( \vy_1, \ldots, \vy_{i_*-1}, X_{i_*} )  \;\vee\; \eta_{k} ( \vy_1, \ldots, \vy_{i_*-1} ) ~.
					\end{align*}	
				is the sought $\varphi'$.
				It is easy to see that $\varphi'$ (after renaming bound variables) still satisfies Conditions~\ref{enum:SBSRtranslation:I} and~\ref{enum:SBSRtranslation:II}.				
				Moreover, we now have established Conditions~\ref{enum:SBSRtranslation:III},~\ref{enum:SBSRtranslation:IV}, and~\ref{enum:SBSRtranslation:V} for the whole formula $\varphi'$ instead of only the subformula $\psi$.
				
				Now suppose $\cQ v$ is an existential quantifier.
				Then, there is an index $i_*$ such that $v \in \vy_{i_*}$.
				We transform $\psi$ into an equivalent disjunction of conjunctions of negated or non-negated basic formulas.
				Due to our assumptions, the constituents of the $k$-th conjunction can be grouped into two parts $\chi_{k}$ and $\eta_{k}$
				such that the variables occurring freely in $\eta_k$ are a subset of $\vy_1 \cup \ldots \cup \vy_{i_*}$ and 
				$\chi_{k}$ may contain free occurrences of variables from $\vx_1, \ldots, \vx_{i_*-1}$ but none from $\vy_{i_*}$.
				After shifting the quantifier $\exists v$ inwards, we obtain the following formula $\varphi'$ that is equivalent to the original $\varphi$:
					\begin{align*} 
						\bigvee_k \chi_{k}( \vx_1, \ldots, \vx_{i_*-1}, \vy_1, \ldots, \vy_{i_*-1}) \;\wedge\; \bigl( \exists v.\, \eta_{k} ( \vy_1, \ldots, \vy_{i_*-1}, \vy_{i_*} ) \bigr) ~.
					\end{align*}	
				It is easy to see that $\varphi'$ (after renaming bound variables) still satisfies Conditions~\ref{enum:SBSRtranslation:I} and~\ref{enum:SBSRtranslation:II}.				
				Moreover, we now have established the Conditions~\ref{enum:SBSRtranslation:III},~\ref{enum:SBSRtranslation:IV}, and~\ref{enum:SBSRtranslation:V} for the whole formula $\varphi'$ instead of only the subformula $\psi$.				
\end{proof}

Lemma~\ref{lemma:SBSRquantifierShifting} provides a tool to shift single quantifiers in SBSR sentences inwards in a fashion that, after applying the lemma iteratively, yields a sentence in which no existential quantifier lies inside the scope of any universal quantifier.
Afterwards quantifiers may be shifted outward again so that we in the end obtain a BSR sentence. 

\begin{theorem}\label{theorem:TransformationSBSRintoBSR}
	Every SBSR sentence $\varphi := \forall \vx_1 \exists \vy_1 \ldots \forall \vx_n \exists \vy_n.\, \psi$ is equivalent to some BSR sentence whose length is at most $n$-fold exponential in the length of $\varphi$.
\end{theorem}
\begin{proof}[Proof Sketch]
	Let $\cQ_1 v_1 \cQ_2 v_2$ be the two rightmost quantifiers in $\varphi$'s quantifier prefix.
	It is easy to check that $\cQ_2 v_2.\, \psi$ satisfies the prerequisites of Lemma~\ref{lemma:SBSRquantifierShifting} and, hence, can be transformed into an equivalent formula $\psi'$ in accordance with that same lemma.
	Then, it is again easy to verify that $\cQ_1 v_1.\, \psi'$ satisfies the lemma again.
	Proceeding further, quantifier by quantifier, we in the end obtain a formula in which all sequences of nested quantifiers have the form $\exists^* \forall^*$.
	Notice that the nesting depth of distinct quantifier blocks in the final result is at most $n$.
	Since the set of occurring atoms does not change (modulo variable renaming), the final formula cannot contain more than $n$-fold exponentially many distinct, non-redundant subformulas $\cQ \vv.\, \chi$ that do not occur within the scope of another quantifier.
	Consequently, shifting outwards all quantifier blocks in an existential-quantifiers-first fashion yields a BSR sentence equivalent to $\varphi$ whose length is at most $n$-fold exponential in the length of $\varphi$.
\end{proof}

On page~\pageref{margin:fineGrainedMeasure} we mentioned that an analysis of the complexity of the SF-BSR translation process can be based on a measure that is more fine-grained than the number of occurring quantifier alternations.
There is a similar measure for the setting of SBSR, which yields an improved upper bound.
Due to space limitations, we will not go any further into the details at this point.
The interested reader will find more material in \cite{Voigt2019PhDthesis}, Sections~3.2 and~3.5.

Theorem~\ref{theorem:TransformationSBSRintoBSR} also holds in the presence of constant symbols:
every SBSR sentence $\varphi$ with constant symbols is equivalent to some BSR sentence $\varphi'$ with the same constant symbols.
Due to Proposition~\ref{proposition:SmallModelsBSR}, Theorem~\ref{theorem:TransformationSBSRintoBSR} entails the following small model property for SBSR.

\begin{corollary}\label{corollary:SmallModelPropertyForSBSR}
	Every satisfiable SBSR sentence $\varphi$ with $n$ $\forall\exists$ quantifier alternations and with or without constant symbols has a model whose size is at most $n$-fold exponential in the length of $\varphi$.
	Hence, \emph{SBSR-Sat} is decidable nondeterministically in $n$-fold exponential time, even if we allow constant symbols to occur.
\end{corollary}


Regarding lower bounds, the result formulated in Theorem~\ref{theorem:LengthSmallestBSRsentences} immediately entails that there are SBSR sentences that inevitably lead to a non-elementary blowup when translating them into equivalent BSR sentences.
Moreover, Proposition~\ref{theorem:SFtoGaifmanLowerBound} is also relevant for SBSR.
It means that for every natural number $k$ there are SBSR sentences whose shortest equivalent in Gaifman normal form is $k$-fold exponentially longer than the original.


\subsection{Taking Boolean Structure into Account}
\label{section:GBSRandBooleanStructure}

It is possible to liberalize the definition of SBSR while retaining decidability of the satisfiability problem, if one takes Boolean structure into account.
For instance, since every $\wedge$-$\vee$-combination of BSR sentences is equivalent to some BSR sentence, we also observe that every $\wedge$-$\vee$-combination of SBSR sentences is equivalent to some SBSR sentence and, ultimately, also to some BSR sentence.
Beyond such trivial observations, one may use approximations of conjunctive and disjunctive normal forms in order to predict when quantifier shifting combined with CNF/DNF-like normal-form transformations of the respective quantifier scopes would result in $\wedge$-$\vee$-combinations of SBSR sentences.
\begin{example}
	Consider the sentence \\
		\centerline{$\varphi := \forall x_1 x_2 \exists y \forall z_1 z_2.\, \bigl( \bigl( P(x_1, z_1) \vee P(z_2, x_2) \bigr) \wedge P(y, z_1) \bigr) \vee \bigl( P(x_2, z_2) \wedge P(x_1, x_2) \bigr)$,}
	which does not satisfy the conditions of SBSR, as $x_1, z_1$ and $z_1, y$ co-occur in $P(x_1, z_1)$, $P(y, z_1)$, respectively.
	However, the sentence $\varphi$ can be transformed into the equivalent \\
	\centerline{$
		\begin{array}{l}
			\Bigl(\forall x_1 x_2 z_1 z_2.\, \bigl( P(x_1, z_1) \vee P(z_2, x_2) \vee P(x_2, z_2) \bigr) \,\wedge\, \bigl( P(x_1, z_1) \vee P(z_2, x_2) \vee P(x_1, x_2) \bigr) \Bigr) \\
			\wedge \Bigl( \forall x_1 x_2 \exists y \forall z_1 z_2.\, \bigl( P(y, z_1) \vee P(x_2, z_2) \bigr) \,\wedge\, \bigl( P(y, z_1) \vee P(x_1, x_2) \bigr) \Bigr) ~.
		\end{array}
	$}
	Evidently, each of the two constituents of the topmost conjunction is an SBSR sentence. 
\end{example}
More details regarding the liberalization of SBSR based on a suitable analysis of Boolean structure can be found in~\cite{Voigt2019PhDthesis}, Section~3.6.


\section{The Separated Ackermann Fragment (SAF)}
\label{section:GeneralizedAF}

The \emph{Ackermann fragment (AF)}\label{para:DefinitionAF}
comprises all relational first-order sentences in prenex normal form with an $\exists^*\forall\exists^*$ quantifier prefix and without equality.
Ackermann derived the finite model property for AF in~\cite{Ackermann1928}.
Ackermann's decidability proof in~\cite{Ackermann1954} proceeds via a reduction to MFO-Sat without any reference to AF's finite model property.
The satisfiability problem for AF is \EXPTIME-complete (see~\cite{Borger1997}, Section~6.3).
In~\cite{Dreben1979} the finite model property of AF with equality is derived.

Gurevich~\cite{Gurevich1973} and Maslov and Orevkov~\cite{Maslov1972} studied AF sentences with arbitrary function symbols, which yields the \emph{Gurevich--Maslov--Orevkov fragment}.
While Gurevich proved the finite model property for this fragment, Orevkov and Maslov took a proof-theoretic route based on Maslov's \emph{inverse method}.
Another extension of AF is the \emph{Shelah fragment}: $\exists^* \forall \exists^*$-sentences with equality and a single unary function symbol~\cite{Shelah1977}.
This class contains \emph{infinity axioms} and, hence, does not possess the finite model property.
A more detailed version of Shelah's proof can be found in Section~7.3 in~\cite{Borger1997}.

In the present section, we generalize AF to the \emph{separated Ackermann fragment (SAF)}.
Moreover, we devise an effective translation procedure from SAF sentences into equivalent AF sentences.
It will turn out that this procedure is compatible with function symbols and equality.
That is, our results will show that 
	SAF with equality is equivalent to AF with equality,
	SAF with arbitrary function symbols is equivalent to the Gurevich--Maslov--Orevkov fragment, and
	SAF with equality and a single unary function symbol in equivalent to the Shelah fragment.
Hence, all these extensions of SAF are decidable.

Intuitively speaking, an SAF sentence is of the form $\varphi := \forall \vx_1 \exists \vy_1 \vu_1 \ldots \forall \vx_n \exists \vy_n \vu_n.\, \psi$ with quantifier-free $\psi$ and it satisfies the following properties.
Each atom in $\varphi$ contains only variables from some subsequence of $\varphi$'s quantifier prefix of the form $\exists^* \forall \exists^*$.
If two atoms share a universally quantified variable or some variable from the trailing $\exists^*$-block of their respective quantifier subsequence, then they have the same $\exists^* \forall \exists^*$-subsequence as source of all their variables.
\marginpar{\label{page:IndexOfVariables}}
For the formal definition, we define the \emph{index of a variable $v \in \vars(\varphi)$} to be $\idx_\varphi(v) := k$ if and only if $v \in \vx_k \cup \vy_k \cup \vu_k$.
We set $\idx_\varphi(v) := \infty$ for all variables not occurring in $\varphi$.
For convenience, we drop the reference to $\varphi$, if $\varphi$ is clear from the context.

\begin{definition}[Separated Ackermann fragment (SAF)]\label{definition:SAFaxiomatic:alternative}
	Let $Y, X, U_1, U_2, U_3, \ldots$ be pair\-wise-disjoint, countable sets of first-order variables.
	The \emph{separated Ackermann fragment (SAF)} comprises all relational first-order sentences without equality having the shape
		$\varphi := \forall \vx_1 \exists \vy_1 \vu_1 \ldots \forall \vx_n$ $\exists \vy_n \vu_n.\, \psi$ 
	with quantifier-free $\psi$ where
	(a) $\vx_i \subseteq X$, $\vy_i \subseteq Y$, and $\vu_i \subseteq U_1 \cup U_2 \cup U_3 \cup \ldots$ for every $i$,
	(b) $X = \{x_1, x_2, \ldots\}$,
	(c) for every $j$ and every $u \in U_j$ we have $\idx(u) \geq \idx(x_j)$, and
	(d) for every atom $A$ in $\varphi$ 
		either 
			(d.1) $\vars(A) \subseteq Y$
		or (d.2) there exists some $j$ such that
			$\vars(A) \subseteq \vy_1 \cup \ldots \cup \vy_{\idx(x_j)-1} \cup \{x_j\} \cup U_j$.

	(Notice that this entails separateness of the sets $\{x_1\} \cup U_1, \{x_2\} \cup U_2, \{x_3\} \cup U_3, \ldots$ in $\varphi$ and every atom $A$ in $\varphi$ either contains exclusively variables from $Y$ or there is some $j$ such that $\vars(A) \subseteq \vy_1 \cup \ldots \cup \vy_{\idx(x_j) - 1} \cup \{x_j\} \cup \vu_{\idx(x_j)} \cup \ldots \cup \vu_n$.)
\end{definition}
The tuples $\vx_i$ and $\vy_i, \vu_i$ in $\varphi$ may be empty, i.e.\ $\varphi$'s quantifier prefix does not have to start with a universal quantifier and it does not have to end with an existential quantifier. 
Moreover, variable $u \in \vu$ that occurs in $\varphi$ is associated with exactly one \emph{reference variable} $x \in \vx$, determined by the set $U_i$ to which $u$ belongs.
Intuitively speaking, using suitable equivalence-preserving transformations, any quantifier $\exists u$ with $u \in \vu$ can be shifted out of the scope of any universal quantifier but the one binding $u$'s reference variable.
This is the essence of the first step of the effective translation procedure from SAF into AF.

The following example gives a first impression of SAF sentences and how they can be translated into AF. 

\begin{example}\label{example:FirstSAFsentence}
	Consider the first-order sentence 
	\begin{align*}
		\varphi :=\;\; &\exists y \forall x_1 \exists u_1 \forall x_2 \exists u_2 u_3 .\, 
				\bigl( \neg P(y, x_1) \vee \bigl(Q(x_1,u_1) \wedge R(y,x_2,u_2) \bigr) \bigr)\\
				 &\hspace{18ex} \wedge\; \bigl( P(y, x_1) \vee \bigl( \neg Q(x_1,u_1) \wedge \neg R(y,x_2,u_3) \bigr) \bigr) ~.
	\end{align*}		
	The partition of the variables in $\varphi$ into 
		$Y := \{ y \}$, $X := \{x_1, x_2\}$, and $U_1 := \{ u_1 \}$, $U_2 := \{u_2, u_3\}$
	is a witness for the belonging of $\varphi$ to SAF.
	Due to the Boolean structure of $\varphi$, the quantifiers $\exists u_3$, $\exists u_2$, and $\forall x_2$ can be shifted inwards immediately but $\exists u_1$ cannot.
	This yields the equivalent sentence
	\begin{align*}
			&\exists y \forall x_1 \exists u_1.\, 
				\bigl( \neg P(y, x_1) \vee \bigl(Q(x_1,u_1) \wedge \forall x_2 \exists u_2.\, R(y,x_2,u_2) \bigr) \bigr) \\
			 &\hspace{10ex} \wedge\; \bigl( P(y, x_1) \vee \bigl( \neg Q(x_1,u_1) \wedge \forall x_2 \exists u_3.\, \neg R(y,x_2,u_3) \bigr) \bigr) ~.
	\end{align*}		
	Because of the two universal quantifiers $\forall x_1$ and $\forall x_2$, which are even interspersed with an existential one, $\varphi$ does not belong to AF. 
	However, there exists an equivalent sentence $\varphi'$ in which every atom lies within the scope of at most one universal quantifier:
	\begin{align*}
			\exists &y.\, 
			\Bigl(\forall x_1.\, \bigl(\neg P(y, x_1) \vee \exists u_1.\, Q(x_1,u_1) \bigr)\Bigr) 
				\wedge \Bigl(\bigl(\forall x_1.\, \neg P(y, x_1)\bigr) \vee \forall x_2 \exists u_2.\, R(y,x_2,u_2) \Bigr)\\
			&\wedge \Bigl(\forall x_1.\, \bigl(\exists u_1.\, \neg Q(x_1,u_1)\bigr) \vee P(y,x_1) \Bigr) 
				\wedge \Bigl(\bigl(\forall x_1 \exists u_1.\, \neg Q(x_1,u_1)\bigr) \vee \forall x_2 \exists u_2.\, R(y,x_2,u_2) \Bigr)\\
			&\wedge \Bigl(\bigl(\forall x_2 \exists u_3.\, \neg R(y,x_2,u_3)\bigr) \vee \forall x_1.\, P(y,x_1) \Bigr)
				\wedge \Bigl(\bigl(\forall x_2 \exists u_3.\, \neg R(y,x_2,u_3)\bigr) \vee \forall x_1 \exists u_1.\, Q(x_1,u_1) \Bigr)\\
			&\wedge \Bigl(\bigl(\forall x_2 \exists u_3.\, \neg R(y,x_2,u_3)\bigr) \vee \forall x_2 \exists u_2.\, R(y,x_2,u_2) \Bigr)
	\end{align*}	
	Transforming $\varphi$ into $\varphi'$ requires only basic logical laws and is very similar to approaches we have seen before:  
	first, we shift the quantifiers $\exists u_3, \exists u_2, \forall x_2$ inwards as far as possible.
	Then, we construct a disjunction of conjunctions of certain subformulas using distributivity.
	This allows us to shift the quantifier $\exists u_1$ inwards. 
	Afterwards, we apply the laws of distributivity again to obtain a conjunction of disjunctions of certain subformulas.
	This step enables us to shift the universal quantifier $\forall x_1$ inwards.
	Although the resulting sentence is not in AF, it is reasonably close to AF.
	The only difference is that we get more than only one universally quantified variable in the sentence as a whole, but at most one per atom. 
	We will show later (Lemma~\ref{lemma:ForallBehindOr}) that each such sentence is indeed equivalent to some AF sentence.

	Another example of a simple SAF sentence is the sentence \\
		\centerline{$\psi := \exists u \forall x \exists y \forall z. \bigl(P(u,z) \wedge Q(u,x)\bigr) \vee \bigl(P(y,z) \wedge Q(u,y) \bigr)$.}
	It belongs to SBSR and SAF at the same time, while it does not belong to AF, SF, BSR, or MFO.
	Hence, even the intersection of SBSR and SAF contains sentences which do not fall into the syntactic categories offered by the standard fragments.
\end{example}

Suppose we intend to check membership in SAF given an arbitrary first-order sentence without knowing a-priori how the set of occurring variables is to be partitioned.
Like in the SBSR setting, there is a deterministic polynomial-time procedure for this task based on a graph algorithm (see~\cite{Voigt2019PhDthesis}, Theorem~3.7.3).
	
The next proposition confirms that SAF extends AF and MFO.
Since the sentence $\varphi$ from Example~\ref{example:FirstSAFsentence} belongs to SAF but lies in neither of the other two fragments, it is immediately clear that SAF constitutes a proper extension of both.
\begin{proposition}\label{proposition:InclusionSAF}
	SAF properly contains AF and MFO.
\end{proposition}
\begin{proof}
	Let $\varphi := \exists \vz\, \forall x \exists \vv.\, \psi$ be an AF sentence with quantifier-free $\psi$. 
	We simply set $Y := \vz$, $X := \{x\}$, $U_1 = \emptyset$, $U_2 := \vv$, and $U_i := \emptyset$ for every $i \geq 3$. 
	Then, $\varphi$ trivially satisfies the conditions for SAF.
	
	Let $\varphi' := \forall \vx_1 \exists \vy_1 \ldots \forall \vx_n \exists \vy_n.\, \psi'$ be an MFO sentence. 
	We set $Y := \vy_1 \cup \ldots \cup \vy_n$, $X := \vx_1 \cup \ldots \cup \vx_n$, and $U_i := \emptyset$ for every $i \geq 1$.
	Obviously, this partition of $\varphi'$'s atoms meets all the conditions for SAF sentences.
\end{proof}


\subsection{Translation of SAF into the Ackermann Fragment}
\label{section:TranslationGAFintoAF}

In this section we devise an effective equivalence-preserving translation from SAF into AF that proceeds in two stages.
The first stage unfolds nestings of quantifiers that bind separated sets of variables.
This results in a sentence in which every subformula lies within the scope of at most one universal quantifier.
Such sentences can easily be converted into a special syntactic form, called \emph{SAF special form}.
Then, in the second stage of the translation process, a sentence in SAF special form is transformed into an equivalent AF sentence.

The next lemma provides the central tool (including technical invariants) for the first stage of the translation process.
\begin{lemma}\label{lemma:SAFquantifierShifting}
	Let $Y, X, U_1, U_2, \ldots$ be pairwise-disjoint, countable sets of first-order variables.
	Fix a quantifier sequence $\sigma := \forall \vx_1 \exists \vy_1 \vu_1 \ldots \forall \vx_n \exists \vy_n \vu_n$ (without double occurrences of variables) such that
		(a) $\vx_i \subseteq X = \{x_1, x_2, \ldots\}$, $\vy_i \subseteq Y$, and $\vu_i \subseteq U_1 \cup U_2 \cup \ldots$ for every $i$ and
		(b) every $u \in U_j$ that occurs in $\sigma$ occurs somewhere right of $x_j$.
	Moreover, fix a set of atoms $\At$ over the variables in $\vx_1, \vy_1, \vu_1, \ldots, \vx_n, \vy_n, \vu_n$ such that	
		for every atom $A \in \At$ 
		either 
			(a) $\vars(A) \subseteq Y$
		or (b) there exists some $j$ such that
			$\vars(A) \subseteq \vy_1 \cup \ldots \cup \vy_{\idx(x_j)-1} \cup \{x_j\} \cup U_j$.
		
	Let $\varphi := \cQ v.\, \psi$ be any first-order formula over the atoms from $\At$ such that
	\begin{enumerate}[label=(\roman{*}), ref=(\roman{*})]
		\item \label{enum:SAFtranslation:I} 
			$\varphi$ is in negation normal form, no variable is bound by two distinct quantifier occurrences, no variable occurs freely and bound,
		\item \label{enum:SAFtranslation:II} 
			any sequence of nested quantifiers $\cQ_1 v_1 \ldots \cQ_k v_k$ occurring in $\varphi$ (from left to right) is a subseqence of $\forall \vx_1 \exists \vy_1 \vu_1 \ldots \forall \vx_n \exists \vy_n \vu_n$,
		\item \label{enum:SAFtranslation:III} 
			for every subformula $\forall x_j.\, \chi$ of $\psi$ we have $\vars(\chi) \subseteq \vy_1 \cup \ldots \cup \vy_{\idx(x_j)-1} \cup \{x_j\} \cup U_j$,
		\item \label{enum:SAFtranslation:IV} 
			for every subformula $\exists u.\, \chi$ of $\psi$ with $u \in U_j$ we have $\vars(\chi) \subseteq \vy_1 \cup \ldots \cup \vy_{\idx(x_j)-1} \cup \{x_j\} \cup U_j$,
		\item \label{enum:SAFtranslation:V} 
			for every subformula $\exists y.\, \chi$ of $\psi$ with $y \in \vy_i$ we have that all variables occurring freely in $\chi$ stem from $Y$, and
		\item \label{enum:SAFtranslation:VI} 
			any sequence of nested quantifiers $\cQ_1 v_1 \ldots \cQ_k v_k$ occurring in $\psi$ (from left to right) is either
			a subseqence of $\exists \vy_1 \ldots \exists \vy_n$ or
			a subseqence of $\exists \vy_1 \ldots \exists \vy_{\idx(x_j)-1} \forall x_j\, \exists \bigl(\vu_{\idx(x_j)} \cap U_j\bigr) \ldots \exists \bigl(\vu_n \cap U_j\bigr)$ for some $j$.
	\end{enumerate}	
	Then, $\varphi$ can be transformed into an equivalent formula~$\varphi'$ such that Conditions~\ref{enum:SBSRtranslation:I} and~\ref{enum:SBSRtranslation:II} still apply to $\varphi'$ and Conditions~\ref{enum:SAFtranslation:III} to~\ref{enum:SAFtranslation:VI} apply to $\varphi'$ instead of $\psi$ only.
\end{lemma}
\begin{proof}[Proof Sketch]
	The proof works in analogy to the proof of Lemma~\ref{lemma:SBSRquantifierShifting}.
	At any step, we remove any quantifiers that bind variables which do not occur in their scope.
	A \emph{basic formula} is, again, any atom and any subformula $(\cQ v. \ldots)$ in $\psi$ that does not lie within the scope of any quantifier in $\psi$.
				
		Suppose $\cQ v$ is a universal quantifier and, hence, $v = x_{j_*} \in X$ for some $j_*$.
		We transform $\psi$ into an equivalent conjunction of disjunctions of negated or non-negated basic formulas.
		Due to Conditions~\ref{enum:SAFtranslation:III} and~\ref{enum:SAFtranslation:V}, the constituents of the $k$-th disjunction can be grouped into two parts:
		$\chi_{k} \vee \eta_{k}$ with $\vars(\chi_k) \subseteq \vy_1 \cup \ldots \cup \vy_{\idx(x_{j_*})-1} \cup \{x_{j_*}\} \cup U_{j_*}$ and $\vars(\eta_{k}) \cap \bigl( \{x_{j_*}\} \cup U_{j_*}\bigr)  = \emptyset$.
		Hence, we may shift the universal quantifier $\forall v$ inwards so that it only binds the (sub-)con\-junc\-tions $\chi_{k}$.
		The resulting formula 
			\begin{align*} 
				\bigwedge_k \bigl( \forall &x_{j_*}.\, \chi_{k}( \vy_1, \ldots, \vy_{\idx(x_{j_*})-1}, x_{j_*} ) \bigr) \\[-0.5ex]
					&\vee \eta_{k} \Bigl( \vx_1, \vy_1, \vu_1, \ldots, \vx_{\idx(x_{j_*})-1}, \vy_{\idx(x_{j_*})-1}, \vu_{\idx(x_{j_*})-1}, \bigl( \vx_{\idx(x_{j_*})} \setminus \{x_{j_*}\} \bigr) \Bigr) 
			\end{align*}	
		is the sought $\varphi'$.
		It is easy to see that $\varphi'$ still satisfies Conditions~\ref{enum:SBSRtranslation:I} and~\ref{enum:SBSRtranslation:II}.				
		Moreover, we now have established the Conditions~\ref{enum:SBSRtranslation:III},~\ref{enum:SBSRtranslation:IV}, and~\ref{enum:SBSRtranslation:V} for the whole formula $\varphi'$ instead of only the subformula $\psi$.
		
		The cases where $\cQ v$ is an existential quantifier and where either $v \in Y$ or $v \in U_j$ for some $j$ can be handled in a similar fashion.
\end{proof}

Lemma~\ref{lemma:SBSRquantifierShifting} provides a tool to shift single quantifiers in SAF sentences inwards in a fashion that, when applying the lemma iteratively, yields a sentence in which no universal quantifier lies inside the scope of another universal quantifier.
Any first-order sentence in this form can be further transformed into a particular shape to which we refer as \emph{SAF special form}: 
	\[ \exists \vz. \bigwedge_i \Bigl( \bigvee_j \forall x_{i, j} \exists \vy_{i,j}.\, \chi_{i, j}(\vz, x_{i,j}, \vy_{i,j}) \Bigr) \vee \eta_i(\vz) \]
	where the $\chi_{i, j}$ and the $\eta_i$ are quantifier free.
%
\begin{lemma}[SAF special form]\label{lemma:SAFPreSpecialForm}
	If $\varphi := \forall \vx_1 \exists \vy_1 \vu_1 \ldots \forall \vx_n \exists \vy_n \vu_n.\, \psi$ belongs to SAF, then we can effectively construct an equivalent sentence of the form \\
		\centerline{$\exists \vz. \bigwedge_i \bigl( \bigvee_j \forall x_{i, j} \exists \vy_{i,j}.\, \chi_{i, j}(\vz, x_{i,j}, \vy_{i,j}) \bigr) \vee \eta_{i}(\vz)$,}
	where the $\chi_{i, j}$ and the $\eta_i$ are quantifier free.
\end{lemma}
\begin{proof}
	Let $\cQ_1 v_1 \cQ_2 v_2$ be the two rightmost quantifiers in $\varphi$'s quantifier prefix.
	It is easy to check that $\cQ_2 v_2.\, \psi$ satisfies the prerequisites of Lemma~\ref{lemma:SAFquantifierShifting} and, hence, can be transformed into an equivalent formula $\psi'$ in accordance with that same lemma.
	Then, it is again easy to verify that $\cQ_1 v_1.\, \psi'$ satisfies the lemma again.
	Proceeding further, quantifier by quantifier, we in the end obtain a formula in which no subformula lies within the scope of two distinct universal quantifiers.
	We may restore the property that no two quantifiers in $\varphi'$ bind the same variables by appropriately renaming bound variables in $\varphi'$.
	Hence, we can effectively construct a sentence $\varphi'$ that is equivalent to $\varphi$ and that does not contain any nested occurrences of universal quantifiers.
	
	We now transform $\varphi'$ into a sentence $\varphi''$ in SAF special form.
	First, we shift all existential quantifiers in $\varphi'$ to the front that do not lie within the scope of any universal quantifier.
	In the resulting sentence $\exists \vz.\, \psi'$ every existential quantifier in $\psi'$ lies within the scope of exactly one universal quantifier.
	We treat every subformula of the form $\forall x.\, \chi$ in $\psi'$ as indivisible unit while transforming $\psi'$ into an equivalent conjunction of disjunctions of literals and such indivisible units.
	The resulting formula can be brought into the desired shape by shifting existential quantifiers that lie in the scope of a universal quantifier outwards until they form an existential quantifier block directly right of the corresponding universal quantifier.
\end{proof}
It is interesting to note that a sentence in SAF special form is not merely a Boolean combination of AF sentences.
The difference is that distinct subformulas $\forall x \exists \vy.\, \chi$ and $\forall x' \exists \vy'.\, \chi'$ may share existentially quantified variables.
However, one can show that every such sentence is indeed equivalent to some AF sentence.
Therefore, every SAF sentence is equivalent to an AF sentence.
Before we make this claim precise (cf.\ Lemma~\ref{lemma:SAFTransformationIntoAckermann}), we develop an auxiliary result that we will reuse later.

\begin{lemma}\label{lemma:ForallBehindOr}
	Let $\psi(\vz)$ be a first-order formula of the form $\psi(\vz) := \bigvee_j \forall \vx\, \exists \vy.\, \chi_j(\vz, \vx, \vy)$ with quantifier-free $\chi_j(\vz, \vx, \vy)$.
	Then, $\psi(\vz)$ is equivalent to the formula $\psi'(\vz)$ of the form
			\[ \exists \vv_1 \ldots \vv_q \exists \vy_1 \ldots \vy_q.\, \Bigl( \bigvee_j \bigwedge_{k = 1}^q \chi_j(\vz, \vv_k, \vy_k) \Bigr) 
						\wedge \forall \vx \exists \vy.\, \bigvee_{k = 1}^q \bigwedge_{A \in \At} \bigl( A(\vz, \vx, \vy) \leftrightarrow A(\vz, \vv_k, \vy_k) \bigr) ~, \]	
			where $\At$ denotes the set of all atoms occurring in $\psi(\vz)$ and $q := 2^{|\At|}$.
			In addition, we have $|\vv_k| = |\vx|$ for every $k$ and $|\vy_\ell| = |\vy|$ for every $\ell$.
\end{lemma}
\begin{proof}
		We first prove $\psi(\vz) \models \psi'(\vz)$.
		Let $\cA$ be any structure, $\ve$ be any tuple of elements form $\fUA$, and $j$ be any index such that
		$\cA \models \forall \vx\, \exists \vy.\, \chi_{j}(\ve, \vx, \vy)$.
		For every set $S \subseteq \At$ we define
			\begin{align*} 
				\fUD_S := \bigl\{ \<\va, \vc\,\> \bigm|
					\text{for every $A(\vz, \vx, \vy) \in \At$ we have $\cA \models A(\ve, \va, \vc)$ if and only if $A \in S$} \bigr\} ~.
			\end{align*}	
		We write $S \models \chi_j(\ve, \vx, \vy)$ if $\fUD_S$ is nonempty and if we have $\cA \models \chi_j(\ve, \va, \vc)$ for every tuple $\<\va,\vc\,\>$ in $\fUD_S$.
		Let $S_1, \ldots, S_r$ be an enumeration of all the sets $S_k$ with $S_k \models \chi_j(\ve, \vx, \vy)$.
		Notice that $1 \leq r \leq q$.
		Let $\<\vb_1, \vc_1\,\>, \ldots, \<\vb_r, \vc_r\,\>$ be some sequence with $\<\vb_k, \vc_k\,\> \in \fUD_{S_k}$ for every $k$.
		Then, for every $k$ the assumption $S_k \models \chi_j(\ve, \vx, \vy)$ entails $\cA \models \chi_j(\ve, \vb_k, \vc_k)$.
		Hence, 
			\begin{equation}
				\label{eqn:IfDirection:I}
				\cA \models \exists \vy_1 \ldots \exists \vy_q.\, \bigwedge_{k=1}^r \chi_j(\ve, \vb_k, \vy_k) \;\wedge\; \bigwedge_{k=r+1}^q \chi_j(\ve, \vb_1, \vy_k) ~.
			\end{equation}	
		Let $\va \in \fUA^{|\vx|}$ be any tuple of length $|\vx|$.
		Because of $\cA \models \forall \vx\, \exists \vy.\, \chi_j(\ve, \vx, \vy)$, there is some $S_k$, $1 \leq k \leq r$, and some $\vc$ such that $\<\va, \vc\,\> \in \fUD_{S_k}$ and $S_k \models \chi_j(\ve, \vx, \vy)$.
		Therefore, we get the following for $\<\vb_k, \vc_k\>$:
			\begin{equation}
				\label{eqn:IfDirection:II}
				\cA \models \exists \vy.\, \bigwedge_{A \in \At} \bigl( A(\ve, \va, \vy) \leftrightarrow A(\ve, \vb_k, \vc_k) \bigr) ~. 
			\end{equation}	
		
		Put together, (\ref{eqn:IfDirection:I}) and (\ref{eqn:IfDirection:II}) entail
			\begin{align*}
				\cA \models  
					\exists \vy_1 \ldots \exists \vy_q.\, &\Bigl( \bigvee_j \bigwedge_{k = 1}^r \chi_j(\ve, \vb_k, \vy_k) \;\wedge\; \bigwedge_{k = r+1}^q \chi_j(\ve, \vb_1, \vy_k) \Bigr) \\
					 &\wedge\; \forall \vx\, \exists \vy.\, \bigvee_{k = 1}^r \bigwedge_{A \in \At} \bigl( A(\ve, \vx, \vy) \leftrightarrow A(\ve, \vb_k, \vy_k) \bigl) ~. 
	 		\end{align*}
		This proves $\cA \models \psi'(\ve)$.
		Hence, we have shown that $\cA \models \psi(\ve)$ implies $\cA \models \psi'(\ve)$.
		
		\medskip 
		Next, we show $\psi(\vz)' \models \psi(\vz)$.
		Let $\cA$ be any structure and let $\ve, \vb_1, \ldots, \vb_q, \vc_1, \ldots, \vc_q$ be tuples for which 
			\begin{equation} 				
				\label{eqn:OnlyIfDirectionLeftHandSide}
				\cA \models 
					\Bigl( \bigvee_j \bigwedge_{k = 1}^q \chi_j(\ve, \vb_k, \vc_k) \Bigr) \;\wedge\; \forall \vx\, \exists \vy.\, \bigvee_{k = 1}^q \bigwedge_{A \in \At} \bigl( A(\ve, \vx, \vy) \leftrightarrow A(\ve, \vb_k, \vc_k) \bigr) ~.
			\end{equation}	
		Then, there is some $j$ such that
			$\cA \models
					\bigwedge_{k = 1}^q \chi_j(\ve, \vb_k, \vc_k)$.
		Let $\fUD_1, \ldots, \fUD_q$ be sets defined by 
			\begin{align*}
				\fUD_k := \bigl\{ \va \mid \;
					&\text{there is some $\vc$ such that for every $A(\vz, \vx, \vy) \in \At$} \\
					&\text{we have $\cA \models A(\ve, \va, \vc)$ if and only if $\cA \models A(\ve, \vb_k, \vc_k)$} \bigr\} ~.		
			\end{align*}
		Note that the sets $\fUD_k$ are all nonempty but not necessarily pairwise disjoint.
		Then, because of Assumption (\ref{eqn:OnlyIfDirectionLeftHandSide}), for every $\va \in \fUA^{|\vx|}$ there is some $k$, $1 \leq k \leq q$, such that $\va \in \fUD_k$.
		Because of $\cA \models \chi_j(\ve, \vb_k, \vc_k)$, we therefore have $\cA \models \chi_j(\ve, \va, \vc)$ for some $\vc$.
		In other words, we have $\cA \models \forall \vx\, \exists \vy.\, \chi_j(\ve, \vx, \vy)$ which entails $\cA \models \psi(\ve)$.
		Hence, $\cA \models \psi'(\ve)$ implies $\cA \models \psi(\ve)$.
\end{proof}
Lemma~\ref{lemma:ForallBehindOr} is essential for the second stage in the transformation process between SAF and AF.
With this tool at hand, the following lemma is easy to prove.

\begin{lemma}\label{lemma:SAFTransformationIntoAckermann}
	For every SAF sentence $\varphi$ we can effectively construct an equivalent sentence $\varphi'$ over the same vocabulary that has the shape $\exists \vv \forall x \exists \vw.\, \psi$ with quantifier-free $\psi$. 
\end{lemma}
\begin{proof}[Proof sketch]
	By virtue of Lemma~\ref{lemma:SAFPreSpecialForm}, we can transform $\varphi$ into an equivalent sentence $\varphi''$ in SAF special form, i.e.\ 
		$\varphi'' = \exists \vz. \bigwedge_i \bigl( \bigvee_j \forall x_{i, j} \exists \vy_{i,j}.\, \chi_{i, j}(\vz, x_{i,j}, \vy_{i,j}) \bigr) \vee \eta_i(\vz)$,
	where the $\chi_{i, j}$ and the $\eta_i$ are quantifier free.	
	Consider any subformula of the form $\psi'(\vz) := \bigvee_j \forall x\, \exists \vy.\, \chi_j(\vz, x, \vy)$.
	By Lemma~\ref{lemma:ForallBehindOr}, $\psi'(\vz)$ is equivalent to some formula $\exists \vv' \vy'.\, \chi'(\vz, \vv', \vy') \wedge \forall x \exists \vy.\, \chi''(\vz, x, \vy, \vv', \vy')$ with quantifier-free $\chi', \chi''$.	
	Hence, $\varphi''$ is equivalent to some sentence that, after shifting quantifiers outwards, is of the form 
		$\exists \vz. \bigwedge_i \bigl( \exists \vu_i \forall x_i \exists \vw_i.\, \psi''_i(\vz, \vu_i, x_i, \vw_i) \bigr) \vee \eta_{i}(\vz)$, 
	where the $\psi''_i$ and the $\eta_{i}$ are quantifier free.
	A prenex version of this sentence yields the sought $\varphi'$, since the universal quantifiers distribute over the topmost conjunction.
\end{proof}

Notice that the proofs of Lemmas~\ref{lemma:SAFquantifierShifting} to~\ref{lemma:SAFTransformationIntoAckermann} still work in the presence of the equality predicate or function symbols.
Therefore, we obtain the following result.
\begin{theorem}\label{theorem:EquivalenceOfSAFwithExtendedAF}
	Every SAF sentence $\varphi$ is equivalent to some AF sentence $\psi$.
	Moreover, we get the following for relaxed restrictions on the syntax.
	\begin{enumerate}[label=(\alph{*}), ref=(\alph{*})]
		\item \label{enum:EquivalenceOfSAFwithExtendedAF:I}
			 Every SAF sentence with equality is equivalent to some AF sentence with equality.
		\item \label{enum:EquivalenceOfSAFwithExtendedAF:II}
			Every SAF sentence with arbitrary function symbols and without equality is equivalent to some Gurevich--Maslov--Orevkov sentence (cf.\ page~\pageref{para:DefinitionAF}).
		\item \label{enum:EquivalenceOfSAFwithExtendedAF:III}
			Every SAF sentence with equality and with a single unary function symbol is equivalent to some Shelah sentence (cf.\ page~\pageref{para:DefinitionAF}).
	\end{enumerate}
	In addition, constant symbols are admissible in all of the above cases.
\end{theorem}
Since AF possesses the finite model property, so does SAF, even in the first two syntactically extended cases mentioned in Theorem~\ref{theorem:EquivalenceOfSAFwithExtendedAF}.
On the other hand, it is known that the Shelah fragment contains infinity axioms.
One example is the sentence $\forall x \exists y.\, f(f(y)) \approx f(x) \;\wedge\; f(y) \not\approx x$ (\cite{Borger1997}, proof of Proposition~6.5.5).
Still, the satisfiability problem for the Shelah fragment is known to be decidable (cf.~\cite{Borger1997}, Section~7.3).
Therefore, we get the following positive results regarding the decidability of SAF-Sat.

\begin{corollary}\label{corollary:DecidabilityOfSAF}
	SAF-Sat is decidable, even in the syntactically more liberal cases given in Theorem~\ref{theorem:EquivalenceOfSAFwithExtendedAF}.
	The syntactic extensions of SAF described in items~\ref{enum:EquivalenceOfSAFwithExtendedAF:I} and~\ref{enum:EquivalenceOfSAFwithExtendedAF:II} of Theorem~\ref{theorem:EquivalenceOfSAFwithExtendedAF} enjoy the finite model property.
\end{corollary}

\begin{remark}
	The L\"ob--Gurevich fragment (cf.\ page~\pageref{para:DefinitionMFO}) is subsumed by SAF when we in addition allow unary function symbols.
	By Lemma~\ref{lemma:SAFTransformationIntoAckermann}, every such sentence is equivalent to some $\exists^* \forall \exists^*$-sentence over the same vocabulary.
	The latter kind of sentences constitutes a subclass of the Gurevich--Maslov--Orevkov fragment.
	Hence, Lemmas~\ref{lemma:SAFquantifierShifting} to~\ref{lemma:SAFTransformationIntoAckermann} also establish a translation from the L\"ob--Gurevich fragment to the Gurevich--Maslov--Orevkov fragment.
\end{remark}
At this point we have settled the question concerning decidability of SAF-Sat, also under certain syntactic extensions.
In fact, decidability of SAF-Sat without any syntactic extensions is already a corollary of the decidability of the satisfiability problem for Maslov's fragment~K~\cite{Maslov1968, Hustadt1999}.
The reaon is that the latter syntactically subsumes SAF.

\begin{definition}[Maslov's fragment K, adapted from~\cite{Hustadt1999}]\label{definition:MaslovsFragmentK}
	Let $\varphi$ be any relational sentence in negation normal form and let $\psi(u_1, \ldots, u_m)$ be any subformula of $\varphi$ in which $u_1, \ldots, u_m$ are exactly the freely occurring variables.
	The \emph{$\varphi$-prefix of $\psi$} is the sequence $\cQ_1 v_1 \ldots \cQ_m v_m$ of quantifiers in $\varphi$ (read from left to right) that bind the free variables of $\psi$, in particular, we have $\{v_1, \ldots, v_m\} = \{u_1, \ldots, u_m\}$.
	The \emph{terminal $\varphi$-prefix of $\psi$} is the longest contiguous suffix of $\cQ_1 v_1 \ldots \cQ_m v_m$ starting with a universal quantifier.
	Put differently, if $\cQ_1 v_1 \ldots \cQ_m v_m$ is of the form $\exists v_1 \ldots v_k \forall v_{k+1} \cQ_{k+2} v_{k+2} \ldots \cQ_m v_m$, then the terminal $\varphi$-prefix of $\psi$ is $\forall v_{k+1} \cQ_{k+2} v_{k+2} \ldots \cQ_m v_m$.
	Notice that the terminal prefix may be empty.
	The sentence $\varphi$ belongs to \emph{Maslov's fragment K} if there are $k \geq 0$ universal quantifiers $\forall x_1, \ldots, \forall x_k$ in $\varphi$ that are not interspersed with existential quantifiers such that for every atom $A$ in $\varphi$ the terminal $\varphi$-prefix of $A$ either
	(a) is at most of length one, or
	(b) ends with an existential quantifier, or
	(c) is of the form $\forall x_1 \ldots \forall x_k$.
\end{definition}

\begin{proposition}\label{proposition:MaslovsKcontainsSAF}
	SAF is contained in Maslov's fragment K.
\end{proposition}
\begin{proof}
	Let $\varphi := \forall \vx_1 \exists \vy_1 \vu_1 \ldots \forall \vx_n \exists \vy_n  \vu_n.\, \psi$ be any SAF sentence with quantifier-free $\psi$.
	Recall that $\varphi$ is relational and does not contain equality.
	In the terminology of the definition of Maslov's fragment~K, the terminal $\varphi$-prefix of any atom $A$ in $\varphi$ is either empty or it is a subsequence of $\forall x_j \exists \bigl( \vu_{\idx(x)} \cap U_j \bigr) \ldots \bigl( \vu_n \cap U_j \bigr)$ for some $j$.
	Therefore, the terminal $\varphi$-prefix of $A$ either is either empty or it ends with an existential quantifier or it is of length one.
\end{proof}
Of course, Proposition~\ref{proposition:MaslovsKcontainsSAF} fails for any extensions of SAF with either equality or non-constant function symbols.
We will see in the next section, how SAF can be extended in such a way that we obtain a generalization of the G\"odel--Kalm\'ar--Sch\"utte fragment (GKS).
Although GKS is syntactically contained in Maslov's fragment~K as well, its extension SGKS is not contained (cf.\ Proposition~\ref{proposition:SGKSisIncomparableToMaslovsK}).

We have not yet given any lower bounds on the blowup that we incur when translating SAF sentences into AF.
However, the computational complexity of AF-Sat (in \EXPTIME, cf.~\cite{Borger1997}, Theorems~6.3.26 and~6.3.1) and MFO-Sat (\NEXPTIME-complete, cf.~\cite{Borger1997}, Theorem~6.2.13) provide some evidence that this blowup is at least exponential.
Since MFO is a subfragment of SAF, this entails the following conditional lower bound.
\begin{proposition}\label{proposition:SuccinctnessGapForSAFvsAF}
	In the worst case, there is at least a super-polynomial blowup in formula length when translating SAF sentences into equivalent AF sentences in a uniform algorithmic way, unless $\EXPTIME = \NEXPTIME$
\end{proposition}


\section{The Separated G\"odel--Kalm\'ar--Sch\"utte Fragment (SGKS)}
\label{section:generalizedGKS}

The \emph{G\"odel--Kalm\'ar--Sch\"utte fragment (GKS)} comprises all relational first-order sentences in prenex normal form with an $\exists^*\forall\forall\exists^*$ quantifier prefix and without equality.
G\"odel~\cite{Godel1932, Godel1933}, Kalm\'ar~\cite{Kalmar1933}, and Sch\"utte~\cite{Schutte1934a, Schutte1934b} independently showed that the satisfiability problem for GKS is decidable.
G\"odel and Kalm\'ar established the finite model property.
A probabilistic proof was later given by Gurevich and Shelah~\cite{Gurevich1983}, see also Section~6.2.3 in~\cite{Borger1997}.
Although G\"odel claimed that his proof methods could also be applied for GKS sentences with equality, Goldfarb refuted this claim~\cite{Goldfarb1984a}.
However, decidable subclasses are known, e.g.\ the syntactic subfragments described in~\cite{Goldfarb1984b} and in~\cite{Wirsing1976}, Section~12.
A decidable subclass described in semantic terms is mentioned in Section~6.2.3 in~\cite{Borger1997} --- AF with equality falls into this category, for instance.
Satisfiability for GKS is \NEXPTIME-complete (\cite{Borger1997}, Section~6.2).

It is only a tiny step from the Ackermann fragment to the G\"odel--Kalm\'ar--Sch\"utte fragment: simply allow two consecutive universal quantifiers in the quantifier prefix instead of only one.
If one views the definition of SAF from the right angle, it is a similarly small step to go from SAF to a separated generalization of GKS, which we will call the \emph{separated G\"odel--Kalm\'ar--Sch\"utte fragment (SGKS)}.
Intuitively speaking, an SGKS sentence is of the form $\varphi := \forall \vx_1 \exists \vy_1 \vu_1 \ldots \forall \vx_n \exists \vy_n \vu_n.\, \psi$ with quantifier-free $\psi$ and satisfies the following properties.
Each atom in $\varphi$ contains only variables from some subsequence of $\varphi$'s quantifier prefix of the form $\exists^* \forall \forall \exists^*$.
We allow only fixed pairs of universally quantified variables to co-occur in atoms.
Any two atoms that are associated with the same pair have the same $\exists^* \forall \forall \exists^*$-subsequence as source of all their variables.
The same applies to any two atoms that share some variable from the trailing $\exists^*$-block of their respective quantifier subsequence.
We use the same notion of \emph{index of a variable} like in Definition~\ref{definition:SAFaxiomatic:alternative} (cf.\ page~\pageref{page:IndexOfVariables}).
\begin{definition}[Separated G\"odel--Kalm\'ar--Sch\"utte fragment (SGKS)]\label{definition:SGKSaxiomatic:alternative}
	Let $Y, X, U_1, U_2,$ $U_3, \ldots$ be pair\-wise-disjoint, countable sets of first-order variables.
	The \emph{separated G\"odel--Kalm\'ar--Sch\"utte fragment (SGKS)} comprises all relational first-order sentences without equality having the shape
		$\varphi := \forall \vx_1 \exists \vy_1 \vu_1 \ldots \forall \vx_n \exists \vy_n \vu_n.\, \psi$ 
	with quantifier-free $\psi$ where the following conditions are satisfied.
	(a) $\vx_i \subseteq X$, $\vy_i \subseteq Y$, and $\vu_i \subseteq U_1 \cup U_2 \cup U_3 \cup \ldots$ for every $i$,
	(b) $\{ x_1, x'_1 \}, \{x_2, x'_2\}, \ldots$ is a partition of $X$ into sets of one or two variables each ($x_j = x'_j$ is allowed) such that for every $j$ we have $\idx(x_j) \leq \idx(x'_j)$,
	(c) for every $j$ and every $u \in U_j$ we have $\idx(u) \geq \idx(x'_j)$, and
	(d) for every atom $A$ in $\varphi$ 
		either 
			(d.1) $\vars(A) \subseteq Y$
		or (d.2) there exists some $j$ such that
			$\vars(A) \subseteq \vy_1 \cup \ldots \cup \vy_{\idx(x_j)-1} \cup \{x_j, x'_j\} \cup U_j$.

	(Notice that this entails separateness of the sets $\{x_1, x'_1\} \cup U_1, \{x_2, x'_2\} \cup U_2, \{x_3, x'_3\} \cup U_3, \ldots$ in $\varphi$ and every atom $A$ in $\varphi$ either contains exclusively variables from $Y$ or there is some $j$ such that $\vars(A) \subseteq \vy_1 \cup \ldots \cup \vy_{\idx(x_j) - 1} \cup \{x_j, x'_j\} \cup \vu_{\idx(x'_j)} \cup \ldots \cup \vu_n$.)
\end{definition}

Like for the other fragments the quantifier prefix of SGKS sentences does not have to start with a universal quantifier and it does not have to end with an existential quantifier either. 
Another analogy to SAF sentences is that every variable $u \in \vu_1 \cup \ldots \cup \vu_n$ that occurs in $\varphi$ is associated with a set $\{x_i, x'_i\} \subseteq X$ containing at least one and at most two \emph{reference variables} determined by the set $U_i$ in which $u$ occurs.
Intuitively speaking, like in the case of SAF, any quantifier $\exists u$ with $u \in \vu_1 \cup \ldots \cup \vu_n$ can be shifted out of the scope of any universal quantifier that does not bind one of $u$'s reference variables.

Like for SAF and SBSR, deciding membership in SGKS for a given first-order sentence can be done deterministically in polynomial time (see~\cite{Voigt2019PhDthesis}, Theorem~3.9.2 for details).

SGKS obviously contains sentences that SAF does not, e.g.\ $\forall x_1 x_2.\, P(x_1, x_2)$.
It is also easy to see that SGKS is an extension of SAF: if we restrict the sets $\{x_i, x'_i\}$ to singleton sets, then we essentially obtain SAF.
Hence, also AF and MFO are subsets of SGKS.
Finally, consider any GKS sentence $\exists \vy \forall x x' \exists \vu.\, \psi$ with quantifier-free $\psi$.
We define $Y := \vy$, $X := \{x, x'\}$, $U_1 = \emptyset$, $U_2 := \vu$, and $U_i := \emptyset$ for every $i \geq 3$.
Then, the sets $Y, X, U_1, U_2, \ldots$ satisfy the conditions of Definition~\ref{definition:SGKSaxiomatic:alternative} and thus witness that $\varphi$ belongs to SGKS.
\begin{proposition}\label{proposition:InclusionSGKS}
	SGKS properly contains GKS, SAF, AF, and MFO.
\end{proposition}

In the previous section we have seen that SAF is contained in Maslov's fragment~K.
We will see now that we have left the realm of the latter class when going from SAF to SGKS.
\begin{proposition}\label{proposition:SGKSisIncomparableToMaslovsK}
	SGKS and Maslov's fragment~K are syntactically incomparable.
\end{proposition}
\begin{proof}
	The following sentence witnesses that SGKS is not contained in Maslov's fragment~K: 
		$\forall x_1 x_1' x_2 x'_2.\, P(x_1, x'_1) \vee Q(x_2, x'_2)$ belongs to SGKS but not to Maslov's K.
	On the other hand, it is easy to find sentences that belong to K but not to SGKS, e.g.\
		$\forall x_1 x_2 x_3 \exists y.\, P(x_1, x_2, y) \wedge Q(x_1, x_3, y) \wedge R(x_1, x_2, x_3)$.
\end{proof}


Next, we sketch an equivalence-preserving translation from SGKS into GKS, similar to the one in Section~\ref{section:TranslationGAFintoAF}.
Again, we proceed in two stages, first transforming a given SGKS sentence into \emph{SGKS special form} and, afterwards, into an equivalent GKS sentence.
\begin{lemma}\label{lemma:SGKSTransformation}
	If $\varphi$ belongs to SGKS, we can effectively construct an equivalent sentence $\varphi'$ in which every subformula lies within the scope of at most two universal quantifiers, and the scope of every universal quantifier contains at most one more universal quantifier.
	Moreover, all literals in $\varphi'$ occur in $\varphi$ (modulo variable renaming).
\end{lemma}
The transformation mentioned in the lemma is based on a straightforward adaptation of Lemma~\ref{lemma:SAFquantifierShifting} (the singly occurring variables $x_j$ have to be replaced with $\{x_j, x'_j\}$ and the rest needs to be adapted accordingly).
The sentence $\varphi'$ can easily be further transformed into a particular shape to which we will refer as \emph{SGKS special form}:
	\[ \exists \vz. \bigwedge_i \bigl( \bigvee_j \forall x_{i, j} x'_{i, j} \exists \vy_{i,j}.\, \chi_{i, j}(\vz, x_{i,j}, x'_{i, j}, \vy_{i,j}) \bigr) \vee \eta_i(\vz) \]
	where the $\chi_{i, j}$ and the $\eta_i$ are quantifier free.
	
\begin{lemma}\label{lemma:SGKSTransformationIntoGKS}
	Every SGKS sentence $\varphi$ in SGKS special form can be effectively transformed into an equivalent sentence $\varphi'$ that has the shape $\exists \vz\, \forall x x' \exists \vy.\, \psi$ with quantifier-free $\psi$. 
\end{lemma}
\begin{proof}
	Since $\varphi$ is in SGKS special form, it has the shape \\
		\centerline{$\varphi'' :=\; \exists \vz. \bigwedge_i \bigl( \bigvee_j \forall x_{i, j} x'_{i, j} \exists \vy_{i,j}.\, \chi_{i, j}(\vz, x_{i,j}, x'_{i, j}, \vy_{i,j}) \bigr) \vee \eta_i(\vz)$,}
	where the $\chi_{i, j}$ and the $\eta_i$ are quantifier free.	
	Consider any subformula of the form $\psi' := \bigvee_j \forall x x'\, \exists \vy.\, \chi_j(\vz, x, x', \vy)$, possibly containing free variables from $\vz$.
	By virtue of Lemma~\ref{lemma:ForallBehindOr}, $\psi'$ is equivalent to some formula of the form $\exists \vv' \vy'.\, \chi'(\vz, \vv', \vy') \wedge \forall x x' \exists \vy.\, \chi''(\vz, x, x', \vy, \vv', \vy')$ with quantifier-free $\chi', \chi''$.	
	Hence, $\varphi''$ is equivalent to some sentence that, after shifting some quantifiers outwards, is of the form 
		$\exists \vz. \bigwedge_i \bigl( \exists \vu_i \forall x_i x'_i \exists \vw_i.\, \psi''_i(\vz, \vu_i, x_i, x'_i, \vw_i) \bigr) \vee \eta_{i}(\vz)$, 
	where the $\psi''_i$ and the $\eta_{i}$ are quantifier free.	
	A prenex version yields the sought $\varphi'$.
\end{proof}

\begin{theorem}
	Every SGKS sentence is equivalent to some GKS sentence.
\end{theorem}

Since we know that GKS enjoys the finite model property, this result immediately entails decidability of the satisfiability problem form SGKS (\emph{SGKS-Sat}).
\begin{corollary}
	The satisfiability problem for SGKS is decidable, and SGKS enjoys the finite model property.
\end{corollary}


We finish the present section by proving that SGKS sentences can be substantially more succinct than equivalent GKS sentences.
The following theorem formulates a lower bound regarding the incurred blowup that comes along with any equivalence-preserving translation from SGKS to GKS.

\begin{theorem}\label{theorem:LengthSmallestGKSsentences}
	There is a class of SGKS sentences and some positive integer $n_0$ such that for every integer $n \geq n_0$ the class contains a sentence $\varphi$ with a length linear in $n$ for which any equivalent GKS sentence has a length that is at least exponential in $n$.
\end{theorem}
\begin{proof}[Proof sketch]
	Let $n \geq 1$ be some positive integer.
	Consider the following first-order sentence in which the sets $\{x_1, x_2\}$ and $\{y_1, y_2\}$ are separated:
		\[ \varphi := \forall x_2 \exists y_2 \forall x_1 \exists y_1. \bigwedge_{i=1}^{8n} \bigl( P_i(x_1, x_2) \leftrightarrow Q_i(y_1, y_2) \bigr) ~.\]

	In analogy to the proof of Theorem~\ref{theorem:LengthSmallestBSRsentences}, we construct the following model $\cA$ using the sets
		$\cS_1 := \bigl\{ S \subseteq [8n] \bigm| |S| = 2n \bigr\}$ 
	and 
		$\cS_2 := \bigl\{ S \subseteq \cS_1 \bigm| |S| = \tfrac{1}{2} |\cS_1| \bigr\}$.
	We observe that
		\begin{align*}
			|\cS_1| = {{8n} \choose {2n}} \geq \left( \frac{8n}{2n} \right)^{2n} = 2^{4n}
			\qquad \text{and} \qquad
			|\cS_2| = {{|\cS_1|} \choose {|\cS_1|/2}} \geq \left( \frac{|\cS_1|}{|\cS_1|/2} \right)^{|\cS_1|/2} \geq 2^{2^{4n-1}} .
		\end{align*}	 
	Having the sets $\cS_1, \cS_2$, we now define the structure $\cA$ as follows:
		\begin{description}
			\item $\fUA := \bigl\{ \fa^{(1)}_{S}, \fb^{(1)}_{S} \bigm| S \in \cS_1 \bigr\} \cup \bigl\{ \fa^{(2)}_{S}, \fb^{(2)}_{S} \bigm| S \in \cS_2 \bigr\}$, 
			\item $P_i^\cA := \bigl\{ \<\fa^{(1)}_{S_1}, \fa^{(2)}_{S_2}\> \in \fUA^2 \bigm| i \in S_1 \in S_2 \bigr\}$ for $i = 1, \ldots, 8n$, and
			\item $Q_i^\cA := \bigl\{ \<\fb^{(1)}_{S_1}, \fb^{(2)}_{S_2}\> \in \fUA^2 \bigm| i \in S_1 \in S_2 \bigr\}$ for $i = 1, \ldots, 8n$.
		\end{description}		
	Like in the proof of Theorem~\ref{theorem:LengthSmallestBSRsentences}, it is easy to show that $\cA$ is a model of $\varphi$.

	For every $S \in \cS_1 \cup \cS_2$ we define the structure $\cA_{-S}$ as the substructure of $\cA$ induced by the domain $\fUA_{-S} := \fUA \setminus \{\fb^{(k)}_S\}$, where $k = 1$ if $S \in \cS_1$ and $k = 2$ if $S \in \cS_2$.
	Like in the proof of Theorem~\ref{theorem:LengthSmallestBSRsentences} (Claim~II), we can prove the following claim.
	\begin{description}
		\item \underline{Claim I:} For every $S \in \cS_1 \cup \cS_2$ the substructure $\cA_{-S}$ of $\cA$ does not satisfy $\varphi$.
			\hfill$\Diamond$
	\end{description}	
		
	Let $\varphi_* := \exists \vz \forall x_1 x_2 \exists \vy.\, \psi_*$ with quantifier-free $\psi_*$ be a shortest GKS sentence equivalent to $\varphi$. 
	Suppose that the length of $\varphi_*$ is less than $2^n$.
	Let $\psi' := \bigvee_{i \in I} \chi_i(\vz, x_1, x_2, \vy)$ be a shortest disjunction of conjunctions $\chi_i$ of literals that is equivalent to $\psi_*$.
	We observe that the index set $I$ contains fewer than $2^{2^n}$ indices and that every conjunction $\chi_i$ contains fewer than $2^n$ literals, for otherwise we could find a shorter formula with the desired properties.
	
	Let $\vd$ be some tuple for which we have
		\[ \cA \models \forall x_1 x_2 \exists \vy.\, \bigvee_{i \in I} \chi_i(\vd, x_1, x_2, \vy) ~. \] 
	Let $\fUD := \{ \fb_S^{(2)} \mid \text{$S \in \cS_2$ and $\fb_S^{(2)} \not\in \vd$} \}$.
	Because of $|\vd| \leq |\vz| \leq \len(\varphi_*) \leq 2^n$ and $|\cS_2| \geq 2^{2^{4n-1}} \geq 2^{2^{3n}}$, we have $|\fUD| \geq 2^{2^{2n}}$ for sufficiently large $n$.
	By Claim~I, we observe 
		\[ \cA_{-S} \not\models \forall x_1 x_2 \exists \vy.\, \bigvee_{i \in I} \chi_i(\vd, x_1, x_2, \vy) \] 
	for every $S$ with $\fb_S^{(2)} \in \fUD$.
	Hence, for every $\fb_S^{(2)} \in \fUD$ there is some pair $\fc_1, \fc_2 \in \fUA \setminus \{ \fb_S^{(2)} \}$, some tuple $\vb$ containing $\fb_S^{(2)}$ and some index $i_S \in I$ such that
			$\cA \models \chi_{i_S}(\vd, \fc_1, \fc_2, \vb)$
			and	
			$\cA_{-S} \not\models \exists \vy.\, \chi_{i_S}(\vd, \fc_1, \fc_2, \vy)$.
	Because of $|I| < 2^{2^n}$ and $|\fUD| \geq 2^{2^{2n}}$, there must be some index $i_*$ that appears in the role of $i_S$ for at least
		\[ \frac{|\fUD|}{|I|} \geq \frac{2^{2^{2n}}}{2^{2^n}} = 2^{2^{2n} - 2^n} \geq 2^{(2^n)^2 / 2^n} = 2^{2^n} \]
	distinct $\fb_S^{(2)} \in \fUD$, if $n$ is sufficiently large.
	Let $\fUD_* \subseteq \fUD$ be the set that comprises exactly those elements.
	In other words, we have $|\fUD_*| \geq 2^{2^n}$ and for every $\fb_S^{(2)} \in \fUD_*$ there is some pair $\fc_1, \fc_2$ and some tuple $\vb$ containing $\fb_S^{(2)}$ such that
		\begin{align}
			\label{eqn:thmLengthSmallestGKSsentences}
			\cA \models \chi_{i_*}(\vd, \fc_1, \fc_2, \vb) 
			\quad \text{and} \quad
			\cA_{-S} \not\models \exists \vy.\, \chi_{i_*}(\vd, \fc_1, \fc_2, \vy) ~.
		\end{align}	
	Consider some $\fb_S^{(2)} \in \fUD_*$ with $S \in \cS_2$ and fix it.
	The only atoms in $\chi_{i_*}$ that could possibly contribute to the effect described in (\ref{eqn:thmLengthSmallestGKSsentences}) for $\fb_S^{(2)}$ have the form $Q_j(z, y')$, $Q_j(y,y')$, $Q_j(x_1, y')$, or $Q_j(x_2, y')$ for $z \in \vz$,  $y, y' \in \vy$ and, moreover, the variables $z, y, x_1, x_2$ need to be assigned values $\fb_{T}^{(1)}$ with $T \in \cS_1$.
	Let $\cS'_1$ be the set collecting all the $T$ from $\cS_1$ that are assigned to such variables occurring in atoms of the mentioned kind.
	As $\chi_{i_*}$ contains at most $2^n$ such variables, $|\cS'_1| \leq 2^n$.
	Recall that $S$ contains $\frac{1}{2}|\cS_1| \geq 2^{4n-1}$ sets of indices.
	By construction of $\cS_2$, there must be some $S' \in \cS_2$ such that for every $T \in \cS'_1$ we have $T \in S'$ if and only if $T \in S$.
	Let $\vb'$ be the tuple that results from $\vb$ by replacing every occurrence of $\fb_S^{(2)}$ in the tuple by $\fb_{S'}^{(2)}$.
	Then, we get 
		$\cA_{-S} \models \chi_{i_*}(\vd, \fc_1, \fc_2, \vb')$,
	which contradicts (\ref{eqn:thmLengthSmallestGKSsentences}).	
	Consequently, the length of the sentence $\varphi_*$ cannot be less than $2^n$.
\end{proof}


\section{Separateness and Guarded Quantification}
\label{section:SeparatenessAndGuardedFormulas}

The \emph{guarded fragment (GF)} comprises all relational first-order sentences with equality that satisfy the following properties.
An \emph{atomic guard} $\gamma(\vu, \vv)$ is an atom $A$ such that all $u \in \vu \cup \vv$ occur in $A$.
We define the \emph{guarded fragment (GF)} inductively:
(i) every relational atom is a GF formula (equality is allowed);
(ii) every Boolean combination of GF formulas is a GF formula;
(iii) for all tuples $\vu, \vv$, every atomic guard $\gamma(\vu, \vv)$, and every GF formula $\psi(\vu, \vv)$ the following formulas belong to GF: 
	$\forall \vu.\, \bigl( \gamma(\vu, \vv) \rightarrow \psi(\vu, \vv) \bigr)$, abbreviated by $\bigl( \forall \vu.\, \gamma(\vu, \vv) \bigr)\psi(\vu, \vv)$,
	and 
	$\exists \vu.\, \bigl( \gamma(\vu, \vv) \wedge \psi(\vu, \vv) \bigr)$, abbreviated by $\bigl( \exists \vu.\, \gamma(\vu, \vv) \bigr)\psi(\vu, \vv)$.
Notice that we assume in any GF formula $\bigl( \cQ \vu.\, \gamma(\vu, \vv) \bigr)\psi(\vu, \vv)$ that all variables that occur freely in $\psi$ also occur in $\gamma$.

The guarded fragment was introduced by Andr\'eka, N\'emeti, and van Benthem~\cite{Andreka1998} as one characterization of the fragment of first-order logic in which propositional modal logic can be embedded via the so-called \emph{standard translation} (cf.\ Section~2.4 in~\cite{Blackburn2002}).
Van Benthem~\cite{vanBenthem1997} also proposed a more liberal form of guards, \emph{loose guards}.
A \emph{loose guard} $\gamma(\vu,\vv)$ is a nonempty conjunction of atoms $\gamma(\vu,\vv) := A_1(\vu,\vv) \wedge \ldots \wedge A_k(\vu,\vv)$ such that all $u,v$ with $u \in \vu$ and $v \in \vu \cup \vv$ co-occur in at least one $A_j$.
The \emph{loosely guarded fragment (LGF)} is then defined by liberalizing (iii) such that loose guards are used instead of atomic guards.
In particular, we assume in any LGF formula $\bigl( \cQ \vu.\, \gamma(\vu, \vv) \bigr)\psi(\vu, \vv)$ that 
	(a) all variables that occur freely in $\psi$ also occur in $\gamma$ and
	(b) every variable that is bound by $\cQ \vu$ co-occurs with every free variable from $\psi$ in some atom in $\gamma$.
We will occasionally use sloppy language and speak of LGF \emph{formulas} when we mean loosely guarded formulas that are not necessarily closed.
Formally, LGF exclusively contains sentences.
The same applies to GF \emph{formulas}.

Gr\"adel~\cite{Gradel1999a} derived the tree-like model property for GF and LGF and the finite model property for GF.
Moreover, the computational complexity of the associated satisfiability problems is pinpointed in the same article: both are complete for deterministic doubly exponential time.
More variants of guards and guarded quantification have been proposed, which lead to the definition of the \emph{clique-guarded fragment}~\cite{Gradel1999b} and the \emph{packed guarded fragment}~\cite{Marx2001}, for instance.
Hodkinson~\cite{Hodkinson2002} showed that also the loosely guarded fragment, the clique-guarded fragment, and the packed guarded fragment enjoy the finite model property.

At first glance it seems that guarded quantification and separateness of quantified variables are two opposite properties.
In particular, any guard $\gamma(\vu, \vv)$ in a formula $\bigl( \forall \vu.\, \gamma(\vu, \vv) \bigr) \varphi(\vu, \vv)$ has to ensure that every $u \in \vu$ co-occurs with each $v \in \vv$ in at least one atom in $\gamma$.
This destroys any separateness of variables from $\vu$ and $\vv$ which might be separated in $\varphi$.
However, it turns out that guardedness and separateness can indeed be combined in a beneficial way.

\begin{definition}[Separated loosely guarded fragment (SLGF)]\label{definition:SGFandSLGF}
	Two tuples $\vx, \vy$ are \emph{guard-separated} in a formula $\psi$ with guarded quantification if $\vx$ and $\vy$ are separated in $\psi$ and, in addition, for every guard $\gamma$ in $\psi$ either $\vars(\gamma) \cap \vx = \emptyset$ or $\vars(\gamma) \cap \vy = \emptyset$.
	
	We define the set of \emph{SLGF formulas} inductively as follows.
	\begin{enumerate}[label=(\roman{*}), ref=(\roman{*})]
		\item Every relational atom is an SLGF formula, equality is admitted.
		\item Every Boolean combination of SLGF formulas is an SLGF formula.
		\item Let $\vu, \vv, \vz$ be tuples of variables and let $\gamma(\vu, \vv)$ be any loose guard.
			The following are SLGF formulas: 
			$\forall \vu.\, \bigl( \gamma(\vu, \vv) \rightarrow \psi(\vu, \vv, \vz) \bigr)$
			and 
			$\exists \vu.\, \bigl( \gamma(\vu, \vv) \wedge \psi(\vu, \vv, \vz) \bigr)$,
			where the sets $\vu$ and $\vz$ are guard-separated in $\psi$.
	\end{enumerate}		
		
	The \emph{separated loosely guarded fragment (SLGF)} is the class of all SLGF sentences.
	When we restrict guards to be atomic, we obtain the \emph{separated guarded fragment (SGF)}.
\end{definition}

\begin{remark}\label{remark:MoreLiberalSyntaxForSLGF}
	Notice that every formula of the form $\forall \vu \vx.\, \gamma(\vu, \vv) \wedge \delta(\vx, \vy) \rightarrow \psi(\vu, \vv, \vx, \vy, \vz)$ where $\vu$, $\vx$, and $\vz$ are pairwise distinct and guard-separated in $\psi$ is equivalent to the SLGF formula $\bigl( \forall \vu.\, \gamma(\vu, \vv) \bigr) \bigl( \bigl( \forall \vx.\, \delta(\vx, \vy) \bigr) \psi(\vu, \vv, \vx, \vy, \vz) \bigr)$. 
	A dual observation holds true for existential quantification.
	This means that, under certain restrictions, we may mix variables that are subject to distinct guards in a single quantifier block.
	One could incorporate this idea into the definition of SLGF and, hence, obtain a syntactically slightly extended version.
	However, for the sake of simplicity, we adhere to the simpler definition given above.
\end{remark}

It is easy to see that SLGF is indeed a proper syntactic extension of LGF, and that the same applies to SGF and GF.
A simple sentence witnessing the strictness of these containment relations is the sentence 
	$\bigr( \forall x.\, x \approx x \bigl) \bigl( \exists y.\, y \approx y \bigr) \bigl( P(y) \leftrightarrow \neg P(x) \bigr)$.
It belongs neither to GF nor to LGF, but to both SGF and SLGF.
Moreover, the sentence is a witness of the following observation:
Every MFO sentence can be easily turned into an equivalent SGF sentence with a length linear in the original. 
We just need to add trivial equations $v \approx v$ as guards to subformulas $\cQ v.\, \chi$. 
The result of this transformation lies in the intersection of SGF and \MFOeq. 
\begin{proposition}\label{proposition:SLGFextendsBaseFragments}
	SGF properly contains GF and SLGF properly contains SGF, LGF, and GF.
	Moreover, every MFO sentence can be turned into an equivalent SGF sentence with a length linear in the original.
\end{proposition}
For \MFOeq\ sentences the matter seems to be more complicated. 
The sentence $\forall x y.\, x \approx y$, for instance, is not an SLGF sentence and cannot be directly transformed into an equivalent SLGF sentence in the described way.

Like for all the other novel first-order fragments we have defined, there exists an effective translation procedure that transforms SLGF sentences into equivalent LGF sentences.
\begin{lemma}\label{lemma:EquivalenceSLGFandLGF}
	Every SLGF formula is equivalent to some LGF formula.
\end{lemma}
\begin{proof}
	We prove an auxiliary result from which the lemma follows:
	Consider any SLGF formula $\varphi := \bigl( \cQ \vu.\, \gamma(\vu, \vv) \bigr)\psi(\vu, \vv, \vz)$ where $\psi$ is an LGF formula, $\vu, \vv, \vz$ are pairwise disjoint, and $\varphi$'s free variables are exactly the ones in $\vv, \vz$.
	Then, $\varphi$ is equivalent to some LGF formula $\varphi'(\vv, \vz)$.
	Moreover, any two guard-separated sets of variables in $\varphi$ are also guard-separated in $\varphi'$.

	Suppose $\cQ$ is a universal quantifier (the case for existential quantification is dual).
	Recall that, by definition of SLGF, the tuples $\vu$ and $\vz$ need to be guard-separated in $\psi$.
	Since $\psi$ is an LGF formula and since we assume that no variable occurs freely and bound in $\varphi$ at the same time, we know that in every subformula $\chi := \bigl( \cQ \vx.\, \delta(\vx,\vy) \bigr) \eta(\vx, \vy)$ of $\psi$ we either have $\vars(\chi) \cap \vu = \emptyset$ or $\vars(\chi) \cap \vz = \emptyset$ (or both).
	Moreover, since $\varphi$ is an SLGF formula, we have $\vars(A) \cap \vu = \emptyset$ or $\vars(A) \cap \vz = \emptyset$ for every atom $A$ in $\psi$.
	Hence, $\varphi$ is equivalent to some formula of the form $\varphi'' := \forall \vu.\, \gamma(\vu, \vv) \rightarrow \bigwedge_i \bigl( \chi_i(\vu, \vv) \vee \eta_i(\vv, \vz) \bigr)$, where the $\chi_i$ and $\eta_i$ are disjunctions of literals or LGF formulas of the form $\bigl( \cQ \vx.\, \delta(\vx,\vy) \bigr) \eta(\vx, \vy)$.
	Applying distributivity and shifting the quantifier $\forall \vu$ in $\varphi''$, it is easy to show equivalence to $\varphi' := \bigwedge_i \bigl( \bigl( \forall \vu.\, \gamma(\vu, \vv) \rightarrow \chi_i(\vu, \vv) \bigr) \vee \eta_i(\vv, \vz) \bigr)$.
	This is the sought LGF formula.
\end{proof}
Notice that the proof works irrespectively of the structure of guards.
Hence, we also observe that every SGF formula is equivalent to some GF formula.

In connection with the fact that GF and LGF possess the finite model property~\cite{Gradel1999a, Hodkinson2002}, the obvious consequence of Lemma~\ref{lemma:EquivalenceSLGFandLGF} is that the satisfiability problems associated with SGF and SLGF (\emph{SGF-Sat} and \emph{SLGF-Sat}) are decidable.
\begin{theorem}
	Both SGF and SLGF possess the finite model property.
	Moreover, SGF-Sat and SLGF-Sat are decidable.
\end{theorem}

We conclude this section with an investigation of the succinctness gap between SLGF and LGF.
The following theorem entails that there is no elementary upper bound on the length of the LGF sentences that result from any equivalence-preserving transformation of SLGF sentences (and even SGF sentences) into LGF. 

\begin{theorem}\label{theorem:LengthSmallestLGFsentences}
	There is a class of SGF sentences such that for every integer $n \geq 3$ the class contains a sentence $\varphi$ with $n$ $\forall \exists$ alternations and with a length polynomial in $n$ for which any equivalent LGF sentence has at least $(n-1)$-fold exponential length in $n$.
\end{theorem}
\begin{proof}[Proof sketch]
	Let $n \geq 3$.
	Consider the following SGF sentence in which the sets $\{x_1, \ldots, x_n\}$ and $\{y_1, \ldots, y_n\}$ are separated:
		\begin{align*}
			\varphi := \bigl(\forall x_n.\,& R_n(x_n)\bigr) \bigl(\exists y_n.\, T_n(y_n) \bigr) 
			 		\ldots \\
			 			&\bigl(\forall x_1.\, R_1(x_1, \ldots, x_n) \bigr) \bigl(\exists y_1.\, T_1(y_1, \ldots, y_n) \bigr). 
			 				\bigwedge_{i=1}^{4n} \bigl( P_i(x_1, \ldots, x_n) \leftrightarrow Q_i(y_1, \ldots, y_n) \bigr) .
		\end{align*}	
	
	In order to construct a particular model of $\varphi$, we reuse the sets $\cS_1, \ldots, \cS_n$ from the proof of Theorem~\ref{theorem:LengthSmallestBSRsentences} and define the structure $\cA$ as follows:
		\begin{description}
			\item $\fUA := \bigl\{ \fa^{(k)}_{S}, \fb^{(k)}_{S} \bigm| 1 \leq k \leq n \text{ and } S \in \cS_k \bigr\}$, 
			\item $P_i^\cA := \bigl\{ \<\fa^{(1)}_{S_1}, \ldots, \fa^{(n)}_{S_n}\> \in \fUA^n \bigm| i \in S_1 \in S_2 \in \ldots \in S_n \bigr\}$ for $i = 1, \ldots, 4n$,
			\item $Q_i^\cA := \bigl\{ \<\fb^{(1)}_{S_1}, \ldots, \fb^{(n)}_{S_n}\> \in \fUA^n \bigm| i \in S_1 \in S_2 \in \ldots \in S_n \bigr\}$ for $i = 1, \ldots, 4n$,
			\item $R_j^\cA := \bigl\{ \<\fa^{(j)}_{S_j}, \ldots, \fa^{(n)}_{S_n}\> \in \fUA^n \bigm| S_j \in \ldots \in S_n \bigr\}$ for $j = 1, \ldots, n$, and
			\item $T_j^\cA := \bigl\{ \<\fb^{(j)}_{S_j}, \ldots, \fb^{(n)}_{S_n}\> \in \fUA^n \bigm| S_j \in \ldots \in S_n \bigr\}$ for $j = 1, \ldots, n$.
		\end{description}		
	Once again, it is easy to show that $\cA$ is a model of $\varphi$.
	
	In analogy to the proof of Theorem~\ref{theorem:LengthSmallestBSRsentences}, one may use a game-based argument to show the following claim.	
	For every $S \in \cS_k$, $1 \leq k \leq n$, we define the structure $\cA_{-S}$ as the substructure of $\cA$ induced by the domain $\fUA_{-S} := \fUA \setminus \{\fb^{(k)}_S\}$.
	\begin{description}
		\item \underline{Claim I:} For every $S \in \cS_k$, $1 \leq k \leq n$, the substructure $\cA_{-S}$ of $\cA$ does not satisfy $\varphi$.
			\hfill$\Diamond$
	\end{description}
	A detailed proof can be found in~\cite{Voigt2019PhDthesis}, Section~3.10 (Theorem~3.10.8).

	We have already analyzed the size of the sets $\cS_k$ in  the proof of Theorem~\ref{theorem:LengthSmallestBSRsentences}. 
	Due to the observed lower bounds, we know that $\fUA$ contains at least
		$\sum_{k=1}^{n}\twoup{k}{n}$
	elements $\fb^{(k)}_{S}$.
	
	Let $\varphi_\LGF$ be a shortest LGF sentence that is semantically equivalent to $\varphi$.
	Next, we argue that $\len(\varphi_\LGF)$ is at least $(n-1)$-fold exponential in $n$. 
	We start by introducing some additional notation.
	We divide the domain $\fUA$ into two disjoint parts $\fUA_\fa := \bigl\{ \fa^{(k)}_S \mid 1 \leq k \leq n \text{ and } S \in \cS_k \bigr\}$ and $\fUA_\fb := \bigl\{ \fb^{(k)}_S \mid 1 \leq k \leq n \text{ and } S \in \cS_k \bigr\}$.
	Moreover, we subdivide $\fUA_\fb$ into parts $\fUA_{\fb,k} := \bigl\{ \fb^{(k)}_S \mid S \in \cS_k \bigr\}$ with $1 \leq k \leq n$.
	We define the following vocabularies
	\begin{align*}
		\Sigma &:= \{P_i, Q_i \mid 1 \leq i \leq 4n\} \cup \{R_j, T_j \mid 1 \leq j \leq n\} ~, \\
		\Sigma_{PR} &:= \{P_i \mid 1 \leq i \leq 4n\} \cup \{R_j \mid 1 \leq j \leq n\} ~, \text{ and} \\
		\Sigma_{QT} &:= \{Q_i \mid 1 \leq i \leq 4n\} \cup \{T_j \mid 1 \leq j \leq n\} ~.
	\end{align*}
	Moreover, let $\Sigma'_{PR}$ and $\Sigma'_{QT}$ be disjoint extensions of the vocabularies $\Sigma_{PR}$ and $\Sigma_{QT}$, respectively, each extended by a countably infinite number of nullary predicate symbols.
	
	An atom is called \emph{linear} if every variable in it occurs at most once.
	Any occurrence of a variable $v$ in a non-equational $\Sigma$-atom $A$ is called a \emph{column-$k$-occurrence}, if $v$ is the $(n-k+1)$-st argument from the right in $A$.
	For example, if we fix $n$ to be $6$, then $v$ has a column-$5$-occurrence in each of the atoms $Q_i(x_1, x_2, x_3, x_4, v, x_6)$, $T_3(x_3, x_4, v, x_6)$, $T_5(v, x_6)$, but $v$ has no column-$5$-occurrence in the atoms $T_6(x_6)$ or $Q_i(v, v, v, v, x_5, v)$.
	
	Starting from $\varphi_\LGF$ we construct the sentence $\psi_\LGF$ with the following properties.
	\begin{enumerate}[label=(\alph{*}), ref=(\alph{*})]
		\item \label{enum:theoremLengthSmallestLGFsentences:I} 
			The sentence $\psi_\LGF$ is a Boolean combination of loosely-guarded $\Sigma'_{PR}$-sentences and loosely-guarded $\Sigma'_{QT}$-sentences.
			Moreover, $\psi_\LGF$ is in negation normal form and none of the constituent sentences of $\psi_\LGF$ properly contains a non-atomic subsentence.
		\item \label{enum:theoremLengthSmallestLGFsentences:II} 
			The vocabulary of $\psi_\LGF$ is $\Sigma'_{PR} \cup \Sigma'_{QT}$, i.e.\  $\Sigma$ plus fresh nullary predicate symbols.
		\item \label{enum:theoremLengthSmallestLGFsentences:III} 
			The structure $\cA$ can be uniquely expanded to some model $\cB$ of $\psi_\LGF$ over the same domain and conserving the interpretations of all predicate symbols in $\Sigma$;
			 for every $\cB_{-S}$
			 --- which is defined to be the substructure of $\cB$ induced by the domain $\fUB \setminus \{\fb^{(k)}_S\}$ ---
			 we have $\cB_{-S} \not\models \psi_\LGF$.
		\item \label{enum:theoremLengthSmallestLGFsentences:IV} 
			$\len(\psi_\LGF) \in \cO\bigl( \len(\varphi_\LGF) \bigr)$.
		
		\item \label{enum:theoremLengthSmallestLGFsentences:V} 
			Equations in $\Sigma'_{QT}$-subsentences of $\psi_\LGF$ occur only in guards, and these consist of a single trivial equation, say $v = v$,  and no other atoms, and belong to top-most quantifiers in the subsentence. 
			
		\item \label{enum:theoremLengthSmallestLGFsentences:VI} 
			Every non-equational $\Sigma_{QT}$-atom is linear and for every variable $v$ occurring in it there is some $k$ such that all occurrences of $v$ in non-equational $\Sigma_{QT}$-atoms are column-$k$-occurrences.
		
		\item \label{enum:theoremLengthSmallestLGFsentences:VII} 
			For all distinct variables $v, v'$ that occur freely in a $\Sigma'_{QT}$-subformula $\chi$ and have column-$k$-occurrences and column-$k'$-occurrences, respectively, we know that $k \neq k'$.
	\end{enumerate}
	Notice that every $\Sigma'_{QT}$-subsentence that is part of $\psi_\LGF$ is actually a variable-renamed version of a loosely guarded $\text{FO}^n$ sentence (see Section~\ref{section:FiniteVariableLogicAndSeparateness} for a definition of $n$-variable logics).

	The details of the construction of $\psi_\LGF$ can be found in~\cite{Voigt2019PhDthesis}, Theorem~3.10.8.
	It roughly amounts to the following steps.
	Certain subformulas are simply replaced with Boolean constants \texttt{true} or \texttt{false}, other subformulas $\chi$ are replaced with fresh nullary predicate symbols, say $M$, for which  definitions of the form $M \leftrightarrow \chi$ are conjoined.
	The final result has a length that is linear in the length of the original.
	
	Now, suppose that $\psi_\LGF$ has fewer than $\twoup{n-1}{n}$ subformulas.
	We observed earlier that $\cB \models \psi_\LGF$ and $\cB_{-S} \not\models \psi_\LGF$ for every $S \in \cS_n$.
	Hence, for every $S \in \cS_n$ there is some $\Sigma'_{QT}$-subformula $\psi_S$ in $\psi_\LGF$ of the form
		$\bigl( \exists \vy.\, \gamma_S(\vy, \vz) \bigr) \chi_S(\vy, \vz)$
	and some variable assignment $\beta_S$ such that the following properties hold.
	We have $\beta_S(y_*) = \fb^{(n)}_S$ for exactly one $y_* \in \vy$ and for every $v \in \vy \cup \vz$ different from $y_*$ we have $\beta_S(v) \in \fUA_\fb \setminus \fUA_{\fb,n}$.
	Moreover, we have
	\begin{itemize}
		\item [($*$)] $\cB, \beta_S \models \gamma_S(\vy, \vz) \wedge \chi_S(\vy, \vz)$ and $\cB, \beta' \not\models \gamma_S(\vy, \vz) \wedge \chi_S(\vy, \vz)$ for every $\beta'$ that differs from $\beta_S$ only in the value assigned to $y_*$.
	\end{itemize}
	The tuple $\beta_S(\vz)$ represents a sequence $\vc_S$ of domain elements from $\fUA_\fb$ that can be completed to a chain $\fb^{(1)}_{T_1}, \ldots, \fb^{(n-1)}_{T_{n-1}}, \fb^{(n)}_S$ with $T_1 \in \ldots \in T_{n-1} \in S$.
		
	Fix any $S_* \in \cS_n$ and consider the formula $\psi_{S_*}(\vz)$.
	There is a nonempty set $\hS_*$ such that $\psi_{S_*}(\vz)$ coincides with every $\psi_S(\vz)$ with $S \in \hS_*$.
	For any distinct $S, S' \in \hS_*$ the sequences $\vc_S := \beta_S(\vz)$ and $\vc_{S'} := \beta_{S'}(\vz)$ must differ, for otherwise ($*$) would be violated.
	As there are at most $\prod_{k=1}^{n-1} \twoup{k}{n}$ distinct sequences $\vc_{S}$, $\hS_*$ can contain at most $\prod_{k=1}^{n-1} \twoup{k}{n} < \bigl( \twoup{n-1}{n} \bigr)^n$ sets.
	Recall that there are fewer than $\twoup{n-1}{n}$ subformulas in $\psi_\LGF$.
	We have just inferred that each of these can only serve as $\psi_S$ for at most $\bigl( \twoup{n-1}{n} \bigr)^n$ sets $S \in \cS_n$.
	Hence, only 
		\[ \bigl( \twoup{n-1}{n} \bigr)^n \cdot \twoup{n-1}{n} \;=\; 2^{(n+1) \cdot \twoup{n-2}{n}} \;<\; 2^{\twoup{n-1}{n}} \;=\; \twoup{n}{n} \]
	sets $S$ have a corresponding subformula $\psi_S$.
	But then $|\cS_n| \geq \twoup{n}{n+1}$ implies that there are $S \in \cS_n$ such that $\cB_{-S} \models \psi_\LGF$, which contradicts our assumptions.
	Consequently, $\psi_\LGF$ must have more than $\twoup{n-1}{n}$ subformulas.
\end{proof}


\section{Separateness and Guarded Negation}
\label{section:SeparatenessAndGuardedNegationFormulas}

B\'ar\'any, ten Cate, and Segoufin~\cite{Barany2011, Barany2015} have discovered that guards can be shifted from quantification to negation, see also~\cite{Segoufin2017}.
This leads to the \emph{guarded-negation first-order fragment (GNFO)}.
GNFO comprises all relational first-order formula with equality over the Boolean connectives $\neg, \wedge, \vee$ and existential quantification. 
Universal quantification has to be simulated using negation, based on the equivalence $\forall x.\, \psi  \semequiv  \neg \exists x.\, \neg \psi$.
Every occurrence of the negation sign is accompanied by a guard, i.e.\ negation may only occur in the form $\gamma(\vu, \vv) \wedge \neg \varphi(\vv)$, where $\gamma$ is an atomic guard and $\varphi$ is a GNFO formula.
In terms of expressiveness, GNFO strictly subsumes GF~\cite{Barany2015} (Proposition~2.2 and Example~2.3).
Moreover, in the same article it is shown that GNFO enjoys the tree-like model property and the finite model property.
The associated satisfiability problem is complete for deterministic doubly exponential time.
Clique-guarded variants of GNFO have also been studied~\cite{Barany2015}.

Similar to guarded quantification, guarded negation can be made compatible with separateness of variables in a way that allows us to syntactically extend GNFO while retaining its expressive power and the decidability of the associated satisfiability problem (\emph{GNFO-Sat}).
\begin{definition}[Separated guarded-negation fragment (SGNFO)]
	Given any sequence $\vu_1, \ldots, \vu_n, \vv$ of pairwise-disjoint tuples of first-order variables, a \emph{separated negation guard} $\gamma(\vu_1, \ldots, \vu_n, \vv)$ is a conjunction of $n$ atoms $A_1(\vu_1, \vv) \wedge \ldots \wedge A_n(\vu_n, \vv)$ (possibly equations) such that for every $i$, $1 \leq i \leq n$, all variables from $\vu_i$ occur at least once in $A_i(\vu_i, \vv)$.

	We define the set of \emph{SGNFO formulas} inductively:
	\begin{enumerate}[label=(\roman{*}), ref=(\roman{*})]
		\item every relational atom is an SGNFO formula, equality is admitted;
		\item every $\wedge$-$\vee$-combination of SGNFO formulas is an SGNFO formula;
		\item for every tuple $\vy$ and every SGNFO formula $\psi(\vy)$ the formula $\exists \vy.\, \psi(\vy)$ is an SGNFO formula;
		\item for every separated negation guard $\gamma(\vu_1, \ldots, \vu_n, \vv)$, and every SGNFO formula $\psi(\vu_1,$ $\ldots, \vu_n)$ the formula $\gamma(\vu_1, \ldots, \vu_n, \vv) \wedge \neg \psi(\vu_1, \ldots, \vu_n)$ is an SGNFO formula if the following conditions are met.
			Let $Z$ be the set of variables that are quantified in $\psi(\vu_1, \ldots, \vu_n)$.
			We require that $Z$ can be divided into pairwise disjoint, possibly empty subsets $Z_1, \ldots, Z_n$ such that the sets $Z_1 \cup \vu_1, \ldots, Z_n \cup \vu_n$ are all pairwise separated in $\psi(\vu_1, \ldots, \vu_n)$.
	\end{enumerate}	
	The \emph{separated guarded-negation fragment (SGNFO)} is the set of all SGNFO sentences.
\end{definition}
It is obvious that GNFO is contained in SGNFO and that there are SGNFO sentences that do not belong to GNFO.
Moreover, every MFO sentence $\varphi$ can be easily turned into an equivalent SGNFO sentence with a length linear in the original.
We first transform $\varphi$ into negation normal form and add trivial equations $v \approx v$ as guards to negated atomic subformulas $\neg P(v)$. 
The result lies in the intersection of SGNFO and \MFOeq. 
For \MFOeq\ sentences the matter seems to be more complicated. 
The sentence $\exists x y.\, x \not\approx y$, for instance, is not an SGNFO sentence and does not seem to have an SGNFO equivalent.
\begin{proposition}\label{proposition:SGNFOextendsBaseFragments}
	SGNFO properly contains GNFO.
	Moreover, every MFO sentence $\varphi$ can be turned into an equivalent SGNFO sentence of length $\cO\bigl(\len(\varphi)\bigr)$.
\end{proposition}

After we have seen the results obtained for the other novel first-order fragments, it should not come as a surprise that there is an effective translation from SGNFO to GNFO.

\begin{definition}[Strict separateness]\label{definition:StrictSeparateness}
	Let $\varphi$ be any first-order formula and let $X, Y$ be two disjoint sets of first-order variables.
	We say that $X, Y$ are \emph{strictly separated in $\varphi$}
	if $X$ and $Y$ are separated in $\varphi$ and, in addition, for every subformula $\chi := (\cQ v.\, \ldots)$ of $\varphi$ we either have $\vars(\chi) \cap X = \emptyset$ or $\vars(\chi) \cap Y = \emptyset$.
\end{definition}

\begin{lemma}\label{lemma:EquivalenceSGNFOandGNFO}
	Every SGNFO formula is equivalent to some GNFO formula.
\end{lemma}
\begin{proof}
	We infer two auxiliary results from which the lemma follows:
	\begin{description}
		\item \underline{Claim~I:}
			Consider any SGNFO formula $\varphi(\vu_1, \ldots, \vu_n, \vv) := \gamma(\vu_1, \ldots, \vu_n, \vv) \wedge \neg \psi(\vu_1, \ldots, \vu_n)$ where $\psi(\vu_1, \ldots, \vu_n)$ is any GNFO formula, and the $\vu_1, \ldots, \vu_n$ are pairwise strictly separated in $\psi(\vu_1, \ldots, \vu_n)$.
			Then, $\varphi(\vu_1, \ldots, \vu_n, \vv)$ is equivalent to some GNFO formula $\varphi'(\vu_1, \ldots, \vu_n, \vv)$ in which $\vu_1, \ldots, \vu_n$ are pairwise strictly separated.

		\item \underline{Proof:}
			Let \emph{basic formulas} in $\psi$ be subformulas that 
			do not lie in the scope of any quantifier or negation sign in $\psi$ 
			and that are either 
				guarded negation formulas $\delta(\vx_1, \ldots, \vx_k, \vy) \wedge \neg \chi(\vx_1, \ldots, \vx_k)$, 
				quantified formulas $\exists \vy.\, \chi(\vy, \vx)$, or
				atoms.
			Transform $\psi$ into a conjunction $\psi' := \bigwedge_{i \in I} \eta_i(\vu_1, \ldots, \vu_n)$ of disjunctions $\eta_i$ of basic formulas.
			Since we assumed $\psi(\vu_1, \ldots, \vu_n)$ to be a GNFO formula in which the $\vu_1, \ldots, \vu_n$ are pairwise strictly separated, we conclude that every basic formula $\chi(\vx)$ in $\psi$ satisfies $\vx \cap \vu_\ell \neq \emptyset$ for at most one $\ell$, $1 \leq \ell \leq n$.
			Hence, the disjuncts $\eta_i$ in $\psi'$ can be regrouped such that $\psi'$ has the form
				\[ \bigwedge_{i \in I} \eta_{i, 1}(\vu_1) \vee \ldots \vee  \eta_{i, n}(\vu_n) ~. \]
			Therefore, $\gamma(\vu_1, \ldots, \vu_n, \vv) \wedge \neg \psi(\vu_1, \ldots, \vu_n)$ is equivalent to the following sentence, where $A_1(\vu_1, \vv), \ldots, A_n(\vu_n, \vv)$ is the list of atoms that $\gamma(\vu_1, \ldots, \vu_n, \vv)$ comprises:
				\begin{align*}
					&A_1(\vu_1, \vv) \wedge \ldots \wedge A_n(\vu_n, \vv) \wedge \neg \psi'(\vu_1, \ldots, \vu_n) \\
					&\semequiv\; A_1(\vu_1, \vv) \wedge \ldots \wedge A_n(\vu_n, \vv) \wedge \neg \bigwedge_{i \in I} \eta_{i, 1}(\vu_1) \vee \ldots \vee  \eta_{i, n}(\vu_n) \\
					&\semequiv\; A_1(\vu_1, \vv) \wedge \ldots \wedge A_n(\vu_n, \vv) \wedge \bigvee_{i \in I} \neg \eta_{i, 1}(\vu_1) \wedge \ldots \wedge \neg \eta_{i, n}(\vu_n) \\
					&\semequiv\; \bigvee_{i \in I} \bigl( A_1(\vu_1, \vv) \wedge \neg \eta_{i,1}(\vu_1) \bigr) \wedge \ldots \wedge \bigl( A_n(\vu_n, \vv) \wedge \neg \eta_{i,n}(\vu_n) \bigr)
				\end{align*}
			This is the sought GNFO formula.			
			\hfill$\Diamond$	

		\item \underline{Claim~II:}
			Consider any SGNFO formula $\varphi(\vx, \vv) := \exists \vy.\, \psi(\vy, \vx, \vv)$ where $\psi(\vy, \vx, \vv)$ is any GNFO formula in which the sets $\vy \cup \vx$ and $\vv$ are strictly separated.
			Then, $\varphi(\vx, \vv)$ is equivalent to some GNFO formula $\varphi(\vx, \vv)$ in which $\vy \cup \vx$ and $\vv$ are strictly separated.

		\item \underline{Proof:}
			Let \emph{basic formulas} in $\psi$ be defined like in the proof of Claim~I.
			Transform $\psi$ into a disjunction $\psi' := \bigvee_{i \in I} \eta_i(\vy, \vx, \vv)$ of conjunctions $\eta_i(\vy, \vx, \vv)$ of basic formulas.
			Since we assumed $\psi$ to be a GNFO formula in which the sets $\vy \cup \vx$ and $\vv$ are strictly separated, every basic formula $\chi(\vu)$ in $\psi$ satisfies $\vu \cap (\vy \cup \vx) = \emptyset$ or $\vu \cap \vv = \emptyset$.
			Hence, the conjuncts $\eta_i$ in $\psi'$ can be regrouped such that $\psi'$ has the form
				\[ \bigvee_{i \in I} \eta_{i,1}(\vy, \vx) \wedge \eta_{i,2}(\vv) ~. \]
			Therefore, $\exists \vy.\, \psi(\vy, \vx, \vv)$ is equivalent to the sentence
				\begin{align*}
					\exists \vy.\, \bigvee_{i \in I} \eta_{i,1}(\vy, \vx) \vee \eta_{i,2}(\vv) 
					\quad \semequiv \quad 
					\bigvee_{i \in I} \bigl( \exists \vy.\, \eta_{i,1}(\vy, \vx) \bigr) \wedge \eta_{i,2}(\vv) ~.
				\end{align*}
			This is the sought GNFO formula.			
			\hfill$\Diamond$	
	\end{description}		
	
	Now consider any SGNFO formula $\varphi$ that is not a GNFO formula.	
	Let $\chi(\vu_1, \ldots, \vu_n, \vv) := \gamma(\vu_1, \ldots, \vu_n, \vv) \wedge \neg \chi'(\vu_1, \ldots, \vu_n)$ be a smallest subformula of $\varphi$ that violates the conditions of guarded negation in GNFO.
	Hence, the set $Z$ of variables quantified in $\chi'(\vu_1, \ldots, \vu_n)$ can be divided into pairwise disjoint sets $Z_1, \ldots, Z_n$ such that $Z_1 \cup \vu_1, \ldots, Z_n \cup \vu_n$ are pairwise separated in $\chi'$.
	Further suppose that in $\chi'$ the sets $Z_1 \cup \vu_1, \ldots, Z_n \cup \vu_n$ are not strictly separated.
	Let $\eta(\vx) := \exists \vy.\, \eta'(\vy, \vx)$ be a smallest subformula of $\chi'$ that violates this strict-separateness condition.
	Then, we can subdivide $\vy$ into pairwise disjoint parts $\vy_1, \ldots, \vy_n$ such that $\vy_i \subseteq Z_i$ for every $i$.
	Moreover, we can subdivide $\vx$ into pairwise disjoint parts $\vx_1, \ldots, \vx_n$ such that $\vx_i \subseteq Z_i \cup \vu_i$ for every $i$.
	Then, $\eta(\vx)$ can be rewritten into $\exists \vy_1 \exists \vy_2 \ldots \exists \vy_n.\, \eta'(\vy_1, \vx_1, \ldots, \vy_n, \vx_n)$.
	Since we assume $\eta(\vx)$ to be minimal, the sets $\vy_1 \cup \vx_1, \ldots, \vy_n \cup \vx_n$ are pairwise strictly separated in $\eta'(\vy_1, \vx_1, \ldots, \vy_n, \vx_n)$.
	By Claim~II, $\eta(\vx)$ is equivalent to some $\eta''(\vx)$ in which the $Z_1 \cup \vu_1, \ldots, Z_n \cup \vu_n$ are pairwise strictly separated.
	Therefore, the formula $\chi'(\vu_1, \ldots, \vu_n)$ can be transformed into an equivalent formula $\chi''(\vu_1, \ldots, \vu_n)$ in which $Z_1 \cup \vu_1, \ldots, Z_n \cup \vu_n$ are pairwise strictly separated.
	By Claim~I, $\gamma(\vu_1, \ldots, \vu_n, \vv) \wedge \neg \chi''(\vu_1, \ldots, \vu_n)$ can be transformed into an equivalent formula that belongs to GNFO and in which the sets $\vu_1, \ldots, \vu_n$ are pairwise strictly separated.
	
	By iterative and exhaustive application of the outlined transformation, we can derive a GNFO formula that is equivalent to the SGNFO formula $\varphi$.
\end{proof}

Since GNFO is known to possess the finite model property~\cite{Barany2015}, Lemma~\ref{lemma:EquivalenceSGNFOandGNFO} entails the same for SGNFO.
Of course, this also means that \emph{SGNFO-Sat} is decidable.
\begin{theorem}
	SGNFO possesses the finite model property and, hence, SGNFO-Sat is decidable.
\end{theorem}

We conclude this section with an investigation of the succinctness gap between SGNFO and GNFO.
The following theorem entails that there is no elementary upper bound on the length of the GNFO sentences that result from any equivalence-preserving transformation of SGNFO sentences into GNFO. 

\begin{theorem}\label{theorem:LengthSmallestGNFOsentences}
	There is a class of SGNFO sentences such that for every integer $n \geq 3$ the class contains a sentence $\varphi$ with a length polynomial in $n$ for which any equivalent GNFO sentence has at least $(n-1)$-fold exponential length in $n$.
\end{theorem}
\begin{proof}[Proof sketch]
	Let $n \geq 3$.	
	The following SGNFO sentence is equivalent to the sentence $\varphi$ given in the proof of Theorem~\ref{theorem:LengthSmallestLGFsentences}.
	In the sentence the sets $\{x_1, \ldots, x_n\}$ and $\{y_1, \ldots, y_n\}$ and $\{z\}$ are separated:
		\begin{align*}
				&\exists z.\, z = z  \\
					&\;\;\wedge\; \neg \exists x_n.\, R_n(x_n) \;\wedge\; \neg \exists y_n.\, R_n(x_n) \wedge T_n(y_n) \\
					&\quad\;\; \wedge\; \neg \exists x_{n-1}.\, R_{n-1}(x_{n-1}, x_n) \wedge T_n(y_n) \;\wedge\; \neg \exists y_{n-1}.\, R_{n-1}(x_{n-1}, x_n) \wedge T_{n-1}(y_{n-1}, y_n) \\
			 		&\qquad\;\; \ldots \\
			 			&\qquad\;\; \wedge\; \neg \exists x_1.\, R_1(x_1, \ldots, x_n) \wedge T_2(y_2, \ldots, y_n) \;\wedge\; \neg \exists y_1.\,  R_1(x_1, \ldots, x_n) \wedge T_1(y_1, \ldots, y_n) \\
			 				&\qquad\quad\;\; \wedge\; \bigwedge_{i=1}^{4n} \bigl( R_1(x_1, \ldots, x_n) \wedge T_1(x_1, \ldots, x_n) \;\wedge\; P_i(x_1, \ldots, x_n) \wedge Q_i(y_1, \ldots, y_n) \bigr) \\
			 				&\qquad\qquad\qquad\;\; \vee\; \bigl( R_1(x_1, \ldots, x_n) \wedge T_1(x_1, \ldots, x_n) \;\wedge\; \neg \bigl( P_i(x_1, \ldots, x_n) \vee Q_i(y_1, \ldots, y_n) \bigr) \bigr) .
		\end{align*}	
	The subformulas $z=z$, $R_n(x_n)$, and the more complex $R_i(x_i, \ldots, x_n) \wedge T_i(y_i, \ldots, y_n)$ and $R_i(x_i, \ldots,$ $x_n) \wedge T_{i+1}(y_{i+1}, \ldots, y_n)$ serve as separated negation guards for negated subformulas.
	We need to introduce a bit of redundancy in order to meet the syntactic requirements of SGNFO.
	Nevertheless, it is easy to see that the above sentence is equivalent to the sentence $\varphi$ used in the proof of Theorem~\ref{theorem:LengthSmallestLGFsentences}.
	The rest of the proof works in analogy to that proof.
	Details can be found in Appendix~\ref{appendix:LengthSmallestGNFOsentences}.
\end{proof}


\section{Separateness and Finite-Variable First-Order Logic}
\label{section:FiniteVariableLogicAndSeparateness}

Classes of first-order formulas over a fixed finite set of variables have been studied extensively (see, e.g., \cite{Otto1997, Grohe1998, Dawar1999, Gradel1999c, Kieronski2018} and also~\cite{Libkin2004}, Section~11, and~\cite{Gradel2007}, Sections~1.1.3, 2.7, and 2.8).
If the number of variables is restricted to two, we even obtain a decidable fragment.
The \emph{two-variable fragment (FO$^\text{2}$)} comprises all relational first-order sentences with equality that are build up using at most two variables.
It is important to understand that this restriction allows reusing variable names in quantifiers --- be they nested or not.
Therefore, in the formulas in the present section we explicitly allow variables to occur free and bound in a formula, and to reappear in distinct occurrences of quantifiers in the same formula.
For example, the sentence $\forall x \exists y.\, \bigl( E(x,y) \wedge \exists x.\, \bigl( E(y,x) \wedge \exists y.\, E(x,y) \bigr) \bigr)$ belongs to \FOtwo.

Scott gave a reduction of \emph{\FOtwo\-Sat} to GKS-Sat~\cite{Scott1962}.
This reduction works only for sentences without equality.
In 1975 Moritmer~\cite{Mortimer1975} proved that \FOtwo\ with equality possesses the finite model property.
The computational complexity of the satisfiability problem for \FOtwo\ has been determined by Gr\"adel, Kolaitis, and Vardi~\cite{Gradel1997}: it is \NEXPTIME-complete.
A recent survey of \FOtwo\ and various extensions is~\cite{Kieronski2018}.

We will see in the present section that also in the context of finite-variable logics separateness can give us more syntactic freedom and the ability to express certain properties in a substantially more succinct way.

\begin{definition}[Separated finite-variable formulas]\label{definition:SFOk}
	For every positive integer $k$ we define $\FOk$ to be the set of all relational first-order formulas with equality in which all variables are taken from a finite sequence $x_1,\ldots, x_k$ of pairwise distinct first-order variables.
	
	For every $k \geq 1$ we define the class $\SFOk$ of relational first-order formulas as follows.
	Let $V_1, V_2, V_3, \ldots$ be a sequence of pairwise disjoint sets $V_i$ of first-order variables, each containing exactly $k$ pairwise distinct variables.
	For every $m \geq 1$ we define the set $\SFOkm$ to be the set of all relational first-order formulas $\varphi$ with equality in which all variables are taken from $V_1 \cup \ldots \cup V_m$ and in which all sets $V_1, \ldots, V_m$ are pairwise separated.
	The class $\SFOk$ is the union $\bigcup_{m \geq 1} \SFOkm$.
\end{definition}

It is easy to see that \FOk\ is a special case of \SFOk.
Moreover, MFO is a proper subset of $\text{SFO}^{k}$ for $k=1$.
In contrast, for every positive integer $k$ the \MFOeq\ sentence $\forall x_1 \ldots x_k \exists y.\, \bigwedge_k y \not\approx x_k$ does not belong to \SFOk.
\begin{proposition}
	For every positive $k$, \SFOk\ contains \FOk\ and MFO.
\end{proposition}

In the following lemma we establish the equivalence between \SFOk\ and \FOk\ for every positive $k$ by devising an equivalence-preserving translation between the two fragments.
\begin{lemma}\label{lemma:EquivalenceSFOkandFOk}
	Every \SFOk\ sentence is equivalent to some \FOk\ sentence.
\end{lemma}
\begin{proof}
	Let $m$ be any positive integer and consider any sentence $\varphi$ from $\SFOkm$.
	Then, $\vars(\varphi) \subseteq V_1 \cup \ldots \cup V_m$ and all $V_1, \ldots, V_m$ are pairwise separated in $\varphi$.
	Without loss of generality, we assume that $\varphi$ is in negation normal form.
	
	We prove an auxiliary result from which the lemma follows.
	\begin{description}
		\item \underline{Claim~I:}
				Consider any subformula $\psi = \cQ \vv.\, \chi$ of $\varphi$ with $\vv \subseteq V_i$ for some $i$.
				If the sets $V_1, \ldots, V_m$ are pairwise strictly separated in $\chi$ (cf.\ Definition~\ref{definition:StrictSeparateness}), then we can construct a formula $\psi'$ that is equivalent to $\psi$ and in which all sets $V_1, \ldots, V_m$ are pairwise strictly separated.

		\item \underline{Proof:}
				A \emph{basic formula} is any atom and any subformula $(\cQ' v. \ldots)$ in $\chi$ that does not lie within the scope of any quantifier in $\chi$.
				Suppose $\cQ$ is an existential quantifier. 
				(The case of $\cQ = \forall$ can be treated in an analogous way.)
				
				Let $\vz$ be the tuple collecting all variables that occur freely in $\psi$.
				We first transform $\chi$ into an equivalent disjunction of conjunctions of negated or non-negated basic formulas.
				This is always possible.
				Since the sets $V_1, \ldots, V_m$ are pairwise strictly separated in $\chi$, the constituents of the $j$-th conjunction can be grouped into $m$ parts:
				$\eta_{j,1}\bigl( V_1 \cap (\vv \cup \vz) \bigr), \ldots, \eta_{j,m} \bigl( V_m \cap (\vv \cup \vz) \bigr)$ with $\vars(\eta_{j,\ell}) \subseteq V_\ell$.
				This is possible because of our assumption that the sets $V_1, \ldots, V_m$ are all pairwise strictly separated in $\chi$.
				Hence, since $\vv \subseteq V_i$, $\psi$ is equivalent to a formula of the form
					\[ \exists \vv. \bigvee_j \eta_{j,1}\bigl( V_1 \cap \vz \bigr) \wedge \ldots \wedge \eta_{j,i} \bigl( V_i \cap (\vv \cup \vz) \bigr) \wedge \ldots \wedge \eta_{j,m}(V_m \cap \vz) ~. \]
				We shift the existential quantifier block $\exists \vv$ inwards so that it only binds the (sub-)con\-junc\-tions $\eta_{j,i} \bigl( V_i \cap (\vv \cup \vz) \bigr)$.
				The resulting formula 
					\[ \bigvee_j \eta_{j,1}\bigl( V_1 \cap \vz \bigr) \wedge \ldots \wedge \Bigl( \exists \vv.\, \eta_{j,i} \bigl( V_i \cap (\vv \cup \vz) \bigr) \Bigr) \wedge \ldots \wedge \eta_{j,m}(V_m \cap \vz) ~. \]
				is the sought $\psi'$ in which the sets $V_1, \ldots, V_m$ are all pairwise strictly separated.
				\hfill$\Diamond$
	\end{description}	
	The sets $V_1, \ldots, V_m$ are pairwise strictly separated in any quantifier-free subformula of $\varphi$.
	Hence, applying Claim~I iteratively, we can transform $\varphi$ into an equivalent sentence $\varphi'$ in which the sets $V_1, \ldots, V_m$ are pairwise strictly separated.
	Since $\varphi'$ is a sentence, the strict separateness condition leads to the observation that for every subformula $\cQ \vv.\, \chi$ in $\varphi'$ there is some $j$ such that $\vars(\cQ \vv.\, \chi) \subseteq V_j$.
	As each of the $V_i$ contains exactly $k$ variables, we can rename the bound variables in $\varphi'$ such that $\varphi'$ is an $\wedge$-$\vee$-combination of \FOk\ sentences.
	Since \FOk\ is closed under Boolean combinations, $\varphi'$ is an \FOk\ sentence.
\end{proof}

Since \FOtwo\ possess the finite model property~\cite{Mortimer1975, Gradel1997}, Lemma~\ref{lemma:EquivalenceSFOkandFOk} entails the same for the class of \SFOtwo\ sentences and, hence, \emph{\SFOtwo-Sat} is decidable.
\begin{theorem}
	The class of \SFOtwo\ sentences possesses the finite model property and, hence, the satisfiability problem for \SFOtwo\ sentences is decidable.
\end{theorem}

Having established the equivalence between \FOk\ and \SFOk\ regarding expressiveness, it remains to investigate the succinctness gap between the two fragments.
We will do this for the special case of \SFOtwo\ sentences compared to the class of \FOtwo\ sentences.
\begin{theorem}\label{theorem:LengthSmallestSFOtwoSentences}
	There is a class of \SFOtwo\ sentences and some positive integer $n_0$ such that for every integer $n \geq n_0$ the class contains a sentence $\varphi$ with a length linear in $n$ for which any equivalent \FOtwo\ sentence has a length that is at least exponential in $n$.
\end{theorem}
\begin{proof}
	Let $n \geq 1$ be some positive integer.
	Consider the following first-order sentence in which the sets $\{x_1, x_2\}$ and $\{y_1, y_2\}$ are separated:
		\[ \varphi := \forall x_2 \exists y_2 \forall x_1 \exists y_1. \bigwedge_{i=1}^{2(n+1)} \bigl( P_i(x_1, x_2) \leftrightarrow Q_i(y_1, y_2) \bigr) ~.\]
		
	In analogy to the proof of Theorem~\ref{theorem:LengthSmallestBSRsentences}, we construct a model $\cA$ using the sets
		$\cS_1 := \bigl\{ S \subseteq [2(n+1)] \bigm| |S| = n+1 \bigr\}$ 
	and 
		$\cS_2 := \bigl\{ S \subseteq \cS_1 \bigm| |S| = \tfrac{1}{2} |\cS_1| \bigr\}$.
	We observe 
		\begin{align*}
			|\cS_1| &= {{2(n+1)} \choose {n+1}} \geq \left( \frac{2(n+1)}{n+1} \right)^{n+1} \!= 2^{n+1}
			\; \text{and} \;
			|\cS_2| = {{|\cS_1|} \choose {|\cS_1|/2}} \geq \left( \frac{|\cS_1|}{|\cS_1|/2} \right)^{|\cS_1|/2} \geq 2^{2^n} .
		\end{align*}	 
	\begin{description}
		\item \underline{Claim~I:}
			Let $\hcS$ be any subset of $\cS_2$  such that for every $S \in \hcS$ there is some $T \in S \subseteq \cS_1$ which does not belong to any $S' \in \hcS \setminus \{S\}$.
			Then, $\hcS$ contains at most $|\cS_1| \leq 2^{2(n+1)}$ sets as elements.
		\item \underline{Proof:} Obvious.
			\hfill$\Diamond$
	\end{description}	
	Let $\cA$ be the structure with
		\begin{description}
			\item $\fUA := \bigl\{ \fa^{(1)}_{S}, \fb^{(1)}_{S} \bigm| S \in \cS_1 \bigr\} \cup \bigl\{ \fa^{(2)}_{S}, \fb^{(2)}_{S} \bigm| S \in \cS_2 \bigr\}$, 
			\item $P_i^\cA := \bigl\{ \<\fa^{(1)}_{S_1}, \fa^{(2)}_{S_2}\> \in \fUA^2 \bigm| i \in S_1 \in S_2 \bigr\}$ for $i = 1, \ldots, 2(n+1)$, and
			\item $Q_i^\cA := \bigl\{ \<\fb^{(1)}_{S_1}, \fb^{(2)}_{S_2}\> \in \fUA^2 \bigm| i \in S_1 \in S_2 \bigr\}$ for $i = 1, \ldots, 2(n+1)$.
		\end{description}		
	
	Like in the proof of Theorem~\ref{theorem:LengthSmallestBSRsentences}, it is easy to show that $\cA$ is a model of $\varphi$, and that the following claim holds.
	\begin{description}
		\item \underline{Claim II:}
				For every $S \in \cS_2$ the substructure $\cA_{-S}$ of $\cA$ induced by $\fUA_{-S} := \fUA \setminus \bigl\{ \fb^{(2)}_{S} \bigr\}$ does not satisfy $\varphi$.
				\hfill$\Diamond$
	\end{description}		
	
	Let $\varphi_\FOtwo$ be a shortest \FOtwo\ sentence that is semantically equivalent to $\varphi$.
	Next, we argue that $\len(\varphi_\FOtwo)$ is at least exponential in $n$. 
	In~\cite{Scott1962} a normal form for \FOtwo\ sentences was introduced, which is referred to as \emph{Scott normal form} in the literature.
	For instance, Lemma~8.1.2 in~\cite{Borger1997} states that there is some relational \FOtwo\ sentence $\psi_\FOtwo$ that has the following properties:
	\begin{enumerate}[label=(\alph{*}), ref=(\alph{*})]
		\item $\psi_\FOtwo$ is of the form $\bigl( \forall u v.\, \chi(u,v) \bigr) \wedge \bigwedge_{i=1}^m \forall x \exists y.\, \eta_i(x,y)$ with quantifier-free $\chi$ and $\eta_i$,
		\item the vocabulary underlying $\psi_\FOtwo$ is that of $\varphi_\FOtwo$ extended with fresh unary predicate symbols $R_1, \ldots, R_\kappa$ with $\kappa \in \cO\bigl( \len(\varphi_\FOtwo) \bigr)$,
		\item $\psi_\FOtwo \models \varphi_\FOtwo$,
		\item every model of $\varphi_\FOtwo$ can be uniquely expanded to a model of $\psi_\FOtwo$ over the same domain and conserving the interpretations of all predicate symbols in $\varphi_\FOtwo$, and
		\item $\len(\psi_\FOtwo) \in \cO\bigl( \len(\varphi_\FOtwo) \bigr)$.
	\end{enumerate}
	Let $\cB$ be the unique expansion of $\cA$ for which $\cB \models \psi_\FOtwo$ and $\fUB = \fUA$.
	Claim~II can be extended to $\cB$, because of $\psi_\FOtwo \models \varphi_\FOtwo$.
	The set $\{ \fb^{(2)}_{S} \mid S \in \cS_2\}$ can be partitioned into at most $2^\kappa$ parts, each containing elements that are indistinguishable by their belonging to the sets $R_k^\cB$.
	Let $\hfUD$ be the largest of these parts and let $\hcS := \{ S \mid \fb^{(2)}_{S} \in \hfUD \}$.
	Hence, for all $\fb, \fb' \in \hfUD$ and every $k$ with $1 \leq k \leq \kappa$ we have $\fb \in R_k^\cB$ if and only if $\fb' \in R_k^\cB$.
	\begin{description}
		\item \underline{Claim~III:}
				Let $n$ be sufficiently large.
				If $\kappa$ is polynomial in $n$, then there is some $S_* \in \hcS$ such that $\fb^{(2)}_{S_*} \in \hfUD$ and for every $T \in S_*$ there is some $S' \in \hcS \setminus \{ S_* \}$ that also contains $T$ and we have $\fb^{(2)}_{S'} \in \hfUD$.
				
		\item \underline{Proof:}
				$\hfUD$ contains at least $2^{2^n} / 2^\kappa = 2^{2^n - \kappa}$ domain elements.
				Hence, $|\hcS| \geq 2^{2^n - \kappa}$.
				Moreover, we observe $2^{2^n - \kappa} > 2^{2(n+1)}$ for sufficiently large $n$, if $\kappa$ is polynomial in $n$.
				By Claim~I, there is some $S_* \in \hcS$ such that for every $T \in S_*$ there is some $S' \in \hcS \setminus \{S_*\}$ with $T \in S'$.
				Claim~III follows by definition of $\hfUD$ and $\hcS$.
				\hfill$\Diamond$
	\end{description}
	We fix some $S_* \in \hcS$ as specified in Claim~III.
	Let $\cB_{-S_*}$ be the substructure of $\cB$ induced by the domain $\fUB_* := \fUB \setminus \bigl\{ \fb^{(2)}_{S_*} \bigr\}$.
	By Claim~II (extended to $\cB$), there is some maximal nonempty set $J \subseteq [m]$ such that for every $j \in J$ we have $\cB_{-S_*} \not\models \forall x \exists y.\, \eta_j(x,y)$.
	Consequently, for every $j \in J$ there is some domain element $\fc \in \fUB_*$ such that $\cB \models \eta_j\bigl(\fc, \fb^{(2)}_{S_*}\bigr)$ and $\cB \not\models \eta_j(\fc,\fd)$ for every $\fd \in \fUB_*$.
	Regarding the domain element $\fc$, we distinguish two cases.
	
	Consider any $j \in J$ and any $\fc \in \{ \fb^{(2)}_S \mid S \in \cS_2, S \neq S_* \} \cup \{ \fa^{(1)}_{S} \mid S \in \cS_1\} \cup \{ \fa^{(2)}_{S} \mid S \in \cS_2\}$ for which we have $\cB \not\models \eta_j(\fc,\fd)$ for every $\fd \in \fUB_*$.
	Let $S'$ be some set from $\hcS$ that is different from $S_*$ and for which $\fc \neq \fb_{S'}^{(2)}$.
	Notice that $\eta_j$ is quantifier free and, hence, exclusively contains atoms over two variables $x,y$.
	Moreover, for every binary atom $A(x,y)$ of the form $P_i(x, y)$, $P_i(y, x)$, $Q_i(x, y)$, or $Q_i(y, x)$ we have $\cB \not\models A(\fc, \fd)$ for every $\fd \in \{\fb^{(2)}_S \mid S \in \cS_2\}$, including $\fb^{(2)}_{S_*}$ and $\fb^{(2)}_{S'}$.
	Since all other non-equational atoms occurring in $\eta_j$ are monadic and because of $\fb^{(2)}_{S_*}, \fb^{(2)}_{S'} \in \hfUD$, we conclude the following.
	For every non-equational atom $A$ occurring in $\eta_j$ we have $\cB \models A\bigl(\fc,\fb^{(2)}_{S_*}\bigr)$ if and only if $\cB \models A\bigl(\fc,\fb^{(2)}_{S'}\bigr)$ 
	(more precisely: $\cB \not\models A(\fc, \fb^{(2)}_{S_*})$ and $\cB \not\models A(\fc, \fb^{(2)}_{S'})$).
	Consider any equation $x \approx y$.
	Because of $\fc \in \fUB_{-S_*}$, we have $\cB \not\models \fc \approx \fb_{S_*}^{(2)}$.
	On the other hand, we also have $\cB_{-S_*} \not\models \fc \approx \fb_{S'}^{(2)}$.
	But then, we all in all get $\cB_{-S_*} \models \eta_j\bigl(\fc,\fb^{(2)}_{S'}\bigr)$, which entails $\cB_{-S_*} \models \exists y.\, \eta_j(\fc,y)$.
	This leads to a contradiction and, hence, there cannot be any pair $j, \fc$ as described.
	
	Consider any $j \in J$ and any $\fc = \fb^{(1)}_T$ with $T \in \cS_1$ for which $\cB \not\models \eta_j(\fc,\fd)$ for every $\fd \in \fUB_*$.
	Suppose $T \not\in S_*$.
	Then, for every binary atom $A(x,y)$ of the form $P_i(x, y)$, $P_i(y, x)$, $Q_i(x, y)$, $Q_i(y, x)$, or $x \approx y$ we have $\cB \not\models A(\fc, \fb^{(2)}_{S_*})$ but also $\cB \not\models A(\fc, \fb^{(2)}_{S})$ for every $S \in \cS_1 \setminus \{T\}$.
	Like in the above case we conclude $\cB_{-S_*} \models \exists y.\, \eta_j(\fc,y)$, yielding a contradiction.
	Suppose $T \in S_*$.
	By Claim~III, there is some $S' \in \hcS \setminus \{ S_* \}$ such that $T \in S'$ and $\fb^{(2)}_{S'} \in \hfUD \setminus \{\fb^{(2)}_{S_*}\} \subseteq \fUB_*$.
	Then, we have \\
		\centerline{$\cB \models Q_i\bigl(\fb^{(1)}_{T}, \fb^{(2)}_{S_*}\bigr) \;\text{ if and only if }\; i \in T$}
	and \\
		\centerline{$\cB \models Q_i\bigl(\fb^{(1)}_{T}, \fb^{(2)}_{S'}\bigr) \;\text{ if and only if }\; i \in T$.}
	For every other binary atom $A(x,y)$ of the form $Q_i(y, x)$, $P_i(x, y)$, $P_i(y, x)$, or $x \approx y$ we have 
		$\cB \not\models A\bigl(\fb^{(1)}_{T}, \fb^{(2)}_{S_*}\bigr)$ 
	and 
		$\cB \not\models A\bigl(\fb^{(1)}_{T}, \fb^{(2)}_{S'}\bigr)$.
	For every monadic atom $A(x,y)$ occurring in $\eta_j$ we have
		$\cB \models A\bigl(\fb^{(1)}_{T}, \fb^{(2)}_{S_*}\bigr) \quad\text{ if and only if }\quad \cB \models A\bigl(\fb^{(1)}_{T}, \fb^{(2)}_{S'}\bigr)$.
	All in all, this leads to \\
		\centerline{$\cB \models \eta_j\bigl(\fb^{(1)}_{T}, \fb^{(2)}_{S_*}\bigr) \quad\text{ if and only if }\quad \cB \models \eta_j\bigl(\fb^{(1)}_{T}, \fb^{(2)}_{S'}\bigr)$.}
	Therefore, we get $\cB_{-S_*} \models \exists y.\, \eta_j\bigl(\fb^{(1)}_{T}, y\bigr)$, which constitutes a contradiction.
	
	This means, the number $\kappa$ of unary predicate symbols occurring in $\psi_\FOtwo$ cannot be polynomial in $n$, for otherwise we get $\cB_{-S_*} \models \psi_\FOtwo$ and $\cA_{-S_*} \models \varphi_\FOtwo$.
	Since $\kappa \in \cO \bigl( \len(\varphi_\FOtwo) \bigr)$, it follows that $\len(\varphi_\FOtwo)$ cannot be polynomial in $n$ but must be at least exponential, in order to satisfy $2^{2^n} \leq 2^{2(n+1) + \kappa}$ for growing $n$.
\end{proof}


\section{Separateness and Fluted Formulas}
\label{section:SeparatenessAndFlutedFormulas}

The \emph{fluted fragment (FL)} comprises all relational first-order sentences without equality that satisfy the following properties.
Let $x_1, x_2, x_3, \ldots$ be a fixed ordered sequence of pairwise distinct variables.
For every nonnegative integer $k$ we define the set $\FL{k}$ inductively as follows.
Any atom $P(x_\ell, \ldots, x_k)$ with $\ell \geq 1$ belongs to $\FL{k}$ --- notice that $x_\ell, \ldots, x_k$ is asserted to be a  gap-free subsequence of $x_1, x_2, x_3, \ldots$.
The set $\FL{k}$ is closed under Boolean combinations, i.e.\ if $\varphi$ and $\psi$ belong to $\FL{k}$, then so do $\neg \varphi$, $\varphi \wedge \psi$, $\varphi \vee \psi$, $\varphi \rightarrow \psi$, $\varphi \leftrightarrow \psi$.
Given any $\FL{k+1}$ formula $\varphi(x_1, \ldots, x_{k+1})$, then $\forall x_{k+1}.\, \varphi$ and $\exists x_{k+1}.\, \varphi$ belong to $\FL{k}$.
The \emph{fluted fragment (FL)} is the set $\FL{0}$, which contains exclusively sentences.
Notice that every sentence $\varphi$ from $\FL{k}$ can be turned into an equivalent $\FL{0}$ sentence $\forall x_1 \ldots x_{k}.\, \varphi$ by adding vacuous quantifiers $\forall x_1 \ldots x_{k}$.

The fluted fragment was introduced by Quine in two steps~\cite{Quine1969, Quine1976}.
In an attempt to extrapolate an extension of MFO, Quine~\cite{Quine1969} considered so-called \emph{homogeneous $k$-adic} sentences, i.e.\ FL sentences in which all predicate symbols have arity $k$.
Decidability of \emph{FL-Sat} was shown via an extension of Herbrand's decidability proof for MFO~\cite{Herbrand1930}.
Later on, namely at the very end of~\cite{Quine1976}, Quine claimed that the same proof would work for full FL.
This claim turned out to be wrong~\cite{Noah1980}.
Only recently, a proof including a correct complexity analysis was published~\cite{PrattHartmann2016, PrattHartmann2019}, showing that FL-Sat is non-elementary.

Herzig~\cite{Herzig1990} considered a class of relational first-order sentences that is very similar to the fluted fragment.
\emph{Herzig's ordered fragment} consists of all relational first-order sentences without equality in which every atom $P(v_1, \ldots, v_m)$ satisfies the following property.
For every $i$ the (unique) quantifier $\cQ v_i$ binding $v_i$ lies within the scope of any quantifier $\cQ' u$ if and only if $\cQ' u$ binds one of the $v_j$ with $j < i$, i.e.\ $u \in \{v_1, \ldots, v_{i-1}\}$.
Notice that the definition implies that the $v_1, \ldots, v_m$ are pairwise distinct.
While atoms in fluted formulas $\varphi \in \FL{k}$ need to contain a contiguous suffix of the variable sequence $x_1, \ldots, x_k$, any atom $A$ in Herzig's ordered formulas must contain a contiguous prefix of the variables bound by the quantifier sequence governing $A$.

As the two fragments seem to be so similar, one could ask whether they are equivalent in expressiveness. 
Indeed, using the concept of separateness of variables, we can reconcile the two fragments while, at the same time, extending both of them to a common superclass, called the \emph{separated fluted fragment (SFL)}.

In the first-order formulas in this section we allow bound variables to reappear in distinct occurrences of quantifiers in the same formula.
Before we formulate the definition of SFL, we adapt the following notation from the definition of Maslov's fragment K (Definition~\ref{definition:MaslovsFragmentK}).
Let $\psi(u_1, \ldots, u_m)$ be any subformula of a first-order sentence $\varphi$.
We assume that $u_1, \ldots, u_m$ are exactly the variables occurring freely in $\psi$ and that they are pairwise distinct.
The \emph{$\varphi$-prefix of $\psi$} is the sequence $\cQ_1 v_1 \ldots \cQ_m v_m$ of quantifiers in $\varphi$ (read from left to right) that bind the free variables of $\psi$, in particular, we have $\{v_1, \ldots, v_m\} = \{u_1, \ldots, u_m\}$.

\begin{definition}[Separated fluted fragment (SFL)]\label{definition:SFL}
	Let $\cV_1, \cV_2, \cV_3, \ldots$ be disjoint ordered sequences of pairwise distinct variables $\cV_i = x^i_1, x^i_2, x^i_3, \ldots$.
	In what follows, we occasionally treat the sequences $\cV_i$ as sets.
	
	The \emph{separated fluted fragment (SFL)} comprises all relational first-order sentences $\varphi$ without equality in which every atom $A$ satisfies the following properties.
	\begin{enumerate}[label=(\alph{*}), ref=(\alph{*})]
		\item $A$ is of the form $P(x^i_\ell, \ldots, x^i_k)$ for some predicate symbol $P$ and certain integers $i, k, \ell$ with $i \geq 1$, $k \geq 0$, and $1 \leq \ell \leq k$. 
		\item The $\varphi$-prefix of $A$ is of the form $\cQ_\ell x^i_{\ell}, \ldots, \cQ_k x^i_{k}$ with $\cQ_j \in \{ \exists, \forall \}$.
	\end{enumerate}
\end{definition}

Although separateness is not explicitly mentioned in the definition of SFL, it implicitly plays an important role.
For every atom $A$ in any SFL sentence $\varphi$, we find one sequence $\cV_i$ from which all variables in $A$ stem, i.e.\ $\vars(A) \subseteq \cV_i$.
Since the $\cV_1, \cV_2, \cV_3, \ldots$ are pairwise disjoint, they are, hence, also pairwise separated in $\varphi$.

\begin{example}
	A fluted sentence:
	\[ \forall x_1 \exists x_2. \bigl( \forall x_3.\, P(x_1, x_2, x_3) \bigr) \wedge \bigl( \exists x_3 \forall x_4.\, Q(x_2, x_3, x_4 \bigr) ~.\]

	An SFL sentence that is not fluted:
	\[ \forall x_1 \exists x_2. \bigl( \forall x_2 \forall x_3.\, P(x_1, x_2, x_3) \bigr) \wedge \bigl( \exists x_3 \forall x_4.\, Q(x_2, x_3, x_4 \bigr) ~.\]
\end{example}

It is not hard to see that every FL sentence also belongs to SFL.
The simple monadic sentence
	$\forall x^1_1 \exists x^2_1.\, P(x^1_1) \leftrightarrow Q(x^2_1)$
is neither fluted nor does it belong to Herzig's fluted fragment.
However, it belongs to SFL.
Indeed, every MFO sentence can be turned into an SFL sentence by renaming bound variables.
Consider any MFO sentence $\varphi$ and suppose that all quantifiers in $\varphi$ bind distinct variables.
Let $u_1, \ldots, u_k$ be an enumeration of all the first-order variables occurring in $\varphi$.
Let $\varphi'$ be the sentence that results from $\varphi$ by renaming every $u_i$ into $x^i_1$.
This sentence $\varphi'$ clearly belongs to SFL.

Finally, consider any sentence $\psi$ that belongs to Herzig's ordered fragment.
Let $P(u_1, \ldots,$ $u_m)$ and $Q(v_1, \ldots, v_{m'})$ be two atoms in $\psi$.
Let $j, j'$ be any two indices with $1 \leq j \leq m$ and $1 \leq j' \leq m'$ such that $u_j$ and $v_{j'}$ are bound by the same quantifier $\cQ_j u_j = \cQ'_{j'} v_{j'}$ in $\psi$.
By definition of Herzig's ordered fragment, these quantifiers $\cQ_j u_j$ and $\cQ'_{j'} v_{j'}$ are exactly in the scopes of $\cQ_1 u_1, \ldots, \cQ_{j-1} u_{j-1}$ and $\cQ'_1 v_1, \ldots, \cQ'_{j'-1} v_{j'-1}$, respectively, and no other quantifier scopes.
As the quantifiers $\cQ_j u_j$ and $\cQ'_{j'} v_{j'}$ coincide, the sets $\{u_1, \ldots, u_j\}$ and $\{v_1, \ldots, v_{j'}\}$ must be equal.
Applying this argument iteratively, we infer $j = j'$ and that the sequences $u_1, \ldots, u_j$ and $v_1, \ldots, v_{j'}$ coincide.
Suppose $j_* \geq 1$ is the maximal index such that $u_{j_*}$ and $v_{j_*}$ are bound by the same quantifier.
For any indices $\ell, \ell' > j_*$ we have that neither of the quantifiers $\cQ u_\ell$ and $\cQ' v_{\ell'}$ binding the variables $u_\ell$ and $v_{\ell'}$, respectively, lies in the scope of the other.
For otherwise, assume that $\cQ u_\ell$ were in the scope of $\cQ' v_{\ell'}$.
Hence, $\ell' < \ell$ and there is some $u_{\ell''}$ with $\ell'' < \ell$ such that $u_{\ell''}$ is also bound by the quantifier $\cQ' v_{\ell'}$.
By the above argument, we have that $\ell' = \ell''$ and that the sequences $u_1, \ldots, u_{\ell'}$ and $v_1, \ldots, v_{\ell'}$ must coincide.
But since $j_*$ is maximal and $j_* < \ell'$, we get a contradiction.  
Consequently, we can rename the bound variables in $\psi$ in such a way that every atom $A$ has the form $P(x^1_1, \ldots, x^1_k)$ for some $k$ and the $\psi$-prefix of $A$ is of the form $\cQ_1 x^1_1, \ldots, \cQ_k x^1_k$.

\begin{proposition}\label{proposition:SFLContainsFLandMFOandHerzigsFragment}
	SFL properly contains (modulo renaming of bound variables) FL, MFO, and Herzig's ordered fragment.
\end{proposition}

The following lemma stipulates that every SFL sentence has an equivalent in FL.
As usual, this result is established via an effective equivalence-preserving translation from SFL to FL.
\begin{lemma}\label{lemma:EquivalenceSFLandFL}
	Every SFL sentence is equivalent to some FL sentence.
\end{lemma}
\begin{proof}
	For every nonnegative integer $k$ and every positive integer $i$ we define the set $\FL{k}(\cV_i)$ inductively as follows.
	Any atom $P(x^i_\ell, \ldots, x^i_k)$ with $1 \leq \ell \leq k$ belongs to $\FL{k}(\cV_i)$.
	The set $\FL{k}(\cV_i)$ is closed under Boolean combinations.
	Given any $\FL{k+1}(\cV_i)$ formula $\varphi(x^i_\ell, \ldots, x^i_{k+1})$, then $\forall x_{k+1}.\, \varphi$ and $\exists x_{k+1}.\, \varphi$ belong to $\FL{k}(\cV_i)$.

	Consider any SFL sentence $\varphi$ and let $m$ be the smallest integer such that $\vars(\varphi) \subseteq \cV_1 \cup \ldots \cup \cV_m$.
	Then, all $\cV_1, \ldots, \cV_m$ are pairwise separated in $\varphi$.
	Without loss of generality, we assume that $\varphi$ is in negation normal form.
	\begin{description}
		\item \underline{Claim~I:}
				Consider any subformula $\psi = \cQ x^i_k.\, \chi$ of $\varphi$ with $\cQ \in \{\forall, \exists\}$ that satisfies the following properties:
				\begin{enumerate}[label=(\alph{*}), ref=(\alph{*})]
				 	\item $\chi$ is a Boolean combination of formulas from $\bigcup_{k', i'} \FL{k'}(\cV_{i'})$ --- which we will call \emph{basic formulas} in what follows;
					\item each of these basic formulas that contains $x^i_k$ is an $\FL{k}(\cV_i)$ formula;
					\item every subformula of $\chi$ that is of the form $\cQ'' x^i_k.\, \chi''$ is an $\FL{k-1}(\cV_i)$ formula.
				\end{enumerate}
				\noindent	
				Then, we can construct a formula $\psi'$ such that
				\begin{enumerate}[label=(\arabic{*}), ref=(\arabic{*})]
					\item $\psi'$ is equivalent to $\psi$, 
					\item $\psi'$ is a Boolean combination of formulas from $\bigcup_{k', i'} \FL{k'}(\cV_{i'})$, and
					\item every subformula $\cQ x^i_k.\, \chi'$ occurring in $\psi'$ belongs to $\FL{k-1}(\cV_i)$.
				\end{enumerate}	

		\item \underline{Proof:}
				We treat the case where $\cQ$ is an existential quantifier; the case of $\cQ = \forall$ can be treated dually.
				
				First, we transform $\chi$ into an equivalent disjunction of conjunctions of basic formulas that is of the form
					\[ \bigvee_j \eta_{j,i,k}\bigl( x^i_1, \ldots, x^i_k \bigr) \wedge \bigwedge_{i' \neq i} \bigwedge_{k'} \eta'_{j,i',k'}(x^{i'}_1, \ldots, x^{i'}_{k'}) ~, \]
				where we group the basic formulas in accordance with their belonging to the sets $\FL{k'}(\cV_{i'})$.
				More precisely, the conjunctions $\eta_{j,i,k}$ contain exactly those basic formulas from the $j$-th disjunct in which the variable bound by $\cQ x^i_k$ occurs freely.
				Moreover, any basic formula from $\FL{k'}(\cV_{i'})$ that occurs in the $j$-th disjunct and does not contain $x^i_k$ as free variable is a conjunct of $\eta'_{j,i',k'}$.
				By assumption, each $\eta_{j,i,k}$ belongs to $\FL{k}(\cV_i)$, as we assumed that every basic formula in which $x^i_k$ occurs is an $\FL{k}(\cV_i)$ formula.
				
				Hence, $\psi$ is equivalent to a formula of the form
					\[ \exists x^i_k. \bigvee_j \eta_{j,i,k}\bigl( x^i_1, \ldots, x^i_k \bigr) \wedge \bigwedge_{i' \neq i} \bigwedge_{k'} \eta_{j,i',k'}(x^{i'}_1, \ldots, x^{i'}_{k'}) ~. \]
				We shift the existential quantifier $\exists x^i_k$ inwards so that it only binds the (sub-)conjunctions $\eta_{j,i,k}$.
				The emerging subformula $\exists x^i_k.\, \eta_{j,i,k}$ belongs to $\FL{k-1}(\cV_i)$.
				The result
					\[ \bigvee_j \Bigl( \exists x^i_k.\, \eta_{j,i,k}\bigl( x^i_1, \ldots, x^i_k \bigr) \Bigr) \wedge \bigwedge_{i' \neq i} \bigwedge_{k'} \eta_{j,i',k'}(x^{i'}_1, \ldots, x^{i'}_{k'}) \]
				is the sought $\psi'$ that is a Boolean combination of formulas from $\bigcup_{k', i'} \FL{k'}(\cV_{i'})$.
				\hfill$\Diamond$
	\end{description}	
	By Definition~\ref{definition:SFL}, every atom in $\varphi$ is an $\FL{k'}(\cV_{i'})$ formula for certain $k', i'$.
	Hence, every subformula $\cQ x^i_k.\, \chi$ of $\varphi$ with quantifier-free $\chi$ satisfies the conditions of Claim~I.
	Consider any subformula $\psi := \cQ x^i_k.\, \chi$ of $\varphi$ such that $\chi$ is a Boolean combination of atoms and of formulas $\psi' := \cQ' x^{i'}_{k'}.\, \chi'$ that satisfy the preconditions of Claim~I.
	By Claim~I, we can transform all these $\psi'$ into equivalent formulas $\psi''$ in such a way that $\psi$, after all these transformations, satisfies the preconditions of Claim~I.
	Due to this observation, we can iteratively apply Claim~I to transform the sentence $\varphi$ into an equivalent sentence $\varphi'$ that is a Boolean combination of sentences from $\bigcup_{k', i'} \FL{k'}(\cV_{i'})$.
	Since every sentence $\chi \in \FL{k'}(\cV_{i'})$ is equivalent to the sentence $\forall x^{i'}_1 \ldots x^{i'}_{k'}.\, \chi$, we can transform $\varphi'$ into an equivalent sentence $\varphi''$ that is a Boolean combination of sentences from $\bigcup_{i'} \FL{0}(\cV_{i'})$.
	In $\varphi''$ the sets $\cV_1, \ldots, \cV_m$ are pairwise strictly separated.
	Hence, we can rename bound variables in $\varphi''$ in such a way that the result $\varphi'''$ is a Boolean combination of sentences from $\FL{0}(\cV_1)$.
	This sentence $\varphi'''$ belongs to the fluted fragment.
\end{proof}

Since FL enjoys the finite model property~\cite{PrattHartmann2016}, Lemma~\ref{lemma:EquivalenceSFLandFL} implies that the same holds true for SFL.
Hence, \emph{SFL-Sat} is decidable.
\begin{theorem}
	SFL possesses the finite model property and, hence, SFL-Sat is decidable.
\end{theorem}


\section{Conclusion}
\label{section:Conclusion}

In the present paper we have treated separateness of first-order variables in the context of the classical decision problem.
Separateness turned out to be an enabler for the definition of significant syntactic extensions of nine of the best-known decidable first-order fragments ---
see Figure~\ref{figure:KnownAndNovelFOLfragments} in Section~\ref{section:Introduction}.
Hence, separateness opens a new perspective on the landscape that research activity around the classical decision problem has revealed over the course of about a hundred years.
It seems likely that separateness facilitates more extensions of decidable first-order fragments.
For instance, Maslov's fragment K may be an interesting candidate for being extended, as may be the more recent unary-negation fragment~\cite{Segoufin2013} and the uniform one-dimensional fragment~\cite{Kieronski2014}.

Interestingly, each and every of the novel fragments discussed in the present paper subsumes MFO.
The inclusion of MFO yields a litmus test for the generality of definitions based on separateness: if MFO is not covered, then the definition is not yet liberal enough.
Another peculiarity is that each of our extended fragments exhibits the same expressiveness as the underlying original fragment does.
This is witnessed by the equivalence-preserving translation procedures that we have devised for each extended fragment, say $F$, into the respective original fragment, say $G$.
From this perspective, the syntax of $G$ could be conceived as a normal form of $F$-sentences: there is a procedure bringing any $F$-sentence into $G$-normal form, so to speak.
Furthermore, we have seen that this translation in some cases inevitably leads to a super-polynomial or even non-elementary blowup of the formula length in the worst case --- take a look at Table~\ref{table:SuccinctnessGapsSummary} again for a summary.
This shows that separateness provides the ability to express certain logical properties in significantly more succinct ways. 

In the present paper the main method for proving decidability of the newly introduced  fragments is based on syntactic transformations.
In~\cite{Voigt2019PhDthesis}, Chapter~4, this syntactic point of view is complemented with a semantic perspective, based on an investigation of dependences between existentially and universally quantified variables in sentences.
In general, nested first-order quantification introduces such dependences.
For example, consider the first-order sentence $\varphi := \forall x \exists y.\, P(x) \leftrightarrow Q(y)$.
Skolemization of the quantifier $\exists y$ removes the quantifier and replaces every occurrence of $y$ with the \emph{Skolem term} $f(x)$.
The result is the equivalent second-order sentence 
	$\varphi_\Sk := \exists f.\, \forall x.\, P(x) \leftrightarrow Q(f(x))$,
where the dependence of the argument of $Q$ on $x$ is explicit.
We distinguish two kinds of dependences that occur between existentially quantified variables and universally quantified variables.
Roughly speaking, such a dependence in a sentence $\varphi$ is \emph{strong}, if there are models (with an infinite domain) where the range of the Skolem function $f$ has to be infinite.
On the other hand, we speak of \emph{weak} dependence of existentially quantified $y$ on universally quantified  $x$ in the following case.
If the values of all other variables are fixed and only $x$ does not have a fixed value, then the range of $y$ (and the Skolem function $f$) can always, i.e.\ in all models, be restricted to a finite set of values.
The following definition makes this intuition more precise.

\begin{definition}[Weak dependence]\label{definition:WeakDependenceGeneral}
	Consider any satisfiable relational first-order sentence $\psi$ that contains some subformula $\chi := \exists y.\, \chi'(\vu, \vv, \vx, y)$
	such that 
		the variables from $\vu$ and $\vx$ are universally quantified in $\psi$, and
		the variables from $\vv$ are existentially quantified in $\psi$.
	Let $\psi_\Sk$ be the result of replacing every occurrence of $y$ in $\psi$ with the Skolem term $f(\vu, \vx)$ for some fresh Skolem function~$f$.
	Then, $y$ \emph{depends weakly} on the variables in $\vx$, if every model $\cA \models \psi_\Sk$ can be turned into a model $\cB \models \psi_\Sk$ by replacing $f^\cA$ with some mapping $f^\cB$ satisfying the following property.
	There exists a \emph{finite} family of mappings $\bigl(g_i : \fUA^{|\vu|} \to \fUA \bigr)_{i \in I}$ and some mapping $h : \fUA^{|\vx|} \to I$ such that $f^\cB(\va, \vb) := g_{h(\vb)}(\va)$ for all $\va \in \fUA^{|\vu|}$ and $\vb \in \fUA^{|\vx|}$.
\end{definition}
Regarding weak dependences, BSR, SF, and SBSR are special fragments, since all dependences in such sentences are weak.
In addition, and less obvious, every sentence in which all dependences are weak is equivalent to some BSR sentence~\cite{Voigt2019PhDthesis} (Theorem~4.2.1 and Remark~4.2.12).
Hence, it is fair to say that the property of containing exclusively weak dependences between existential and universal variables is a semantic characterization of BSR (and also of its syntactic extensions SF and SBSR).

Integrating the analysis of weak dependences into Skolemization strategies could make them sensitive to weak dependences.
This might offer interesting and valuable directions for research in the fields of automated reasoning and proof complexity, where improvements to Skolemization strategies can have tremendous effects on the length of proofs~\cite{Baaz1994a, Egly1994}.

Possible connections of weakness of dependences in the sense of Definition~\ref{definition:WeakDependenceGeneral} to the field of dependence logic (broadly construed)~\cite{Vaananen2007, Abramsky2016} remain to be investigated.

Separateness may turn out to be even more versatile in future investigations.
Some hints are given in~\cite{Voigt2019PhDthesis} (e.g.\ in Chapter~7).
We will conclude with a sketch of one more idea.
We have emphasized time and again that, compared to BSR sentences, SF sentences can express certain logical properties much more succinctly.
This holds true in particular for properties that exhibit a high degree of \emph{structural regularity}.
An example for such a property is the one described by the family of SF sentences $\bigl(\varphi_n\bigr)_{n\geq 1}$ with \\
	\centerline{$\varphi_n := \forall x_n \exists y_n \ldots \forall x_1 \exists y_1.\, \bigwedge_{i=1}^{n} \bigl( P_i(x_1, \ldots, x_n) \leftrightarrow Q_i(y_1, \ldots, y_n) \bigr)$.}
We have already encountered variants of this class of sentences in all the proofs of succinctness gaps (Theorems~\ref{theorem:LengthSmallestBSRsentences}, \ref{theorem:LengthSmallestGKSsentences}, \ref{theorem:LengthSmallestLGFsentences}, \ref{theorem:LengthSmallestGNFOsentences}, \ref{theorem:LengthSmallestSFOtwoSentences}).
Although the domain size of the following family of models $\bigl(\cA_n\bigr)_{n\geq 1}$ with $\cA_n \models \varphi_n$ grows massively with increasing $n$, its interpretation of the predicate symbols $P_i$ and $Q_i$ is given by a rather simple pattern and, hence, each $\cA_n$ is very regular in an intuitive sense --- the latter is witnessed by the shortness of the following definition of $\cA_n$:\\
\strut\quad $\fUA_n := \bigcup_{k = 1}^{n} \bigl\{ \fa^{(k)}_{S}, \fb^{(k)}_{S} \bigm| S \in \cP^k([n]) \bigr\}$, where $\cP^k$ is the $k$-fold power set operator,\\
\strut\quad $P_i^{\cA_n} := \bigl\{ \<\fa^{(1)}_{S_1}, \ldots, \fa^{(n)}_{S_n}\> \in \fUA^n \bigm| i \in S_1 \in S_2 \in \ldots \in S_n \bigr\}$ for $i = 1, \ldots, n$, and\\
\strut\quad $Q_i^{\cA_n} := \bigl\{ \<\fb^{(1)}_{S_1}, \ldots, \fb^{(n)}_{S_n}\> \in \fUA^n \bigm| i \in S_1 \in S_2 \in \ldots \in S_n \bigr\}$ for $i = 1, \ldots, n$.\\
Any of the structures $\cA_n$ neatly captures the essence of the logical property described by $\varphi_n$, as every domain element $\fa^{(k)}_{S}$ has a corresponding twin element $\fb^{(k)}_{S}$ that mirrors in the predicates $Q_i^{\cA_n}$ exactly the role that $\fa^{(k)}_{S}$ plays in the predicates $P_i^{\cA_n}$.

More generally, consider any logical property $\pi_n$ that is parameterized by some positive integer $n$ and that can be expressed by a (uniform) family of BSR sentences.
Let $f(n)$ be the function representing the length of a shortest BSR sentence $\psi$ that describes $\pi_n$. 
Let $g(n)$ be the function that denotes the length of a shortest SF sentence describing $\pi_n$.
By Theorem~\ref{theorem:LengthSmallestBSRsentences}, we know that there are properties $\pi_n$ such that $g(n)$ can be bounded from above by some polynomial but we cannot find any integer $k$ such that $f(n)$ is bounded from above by any $k$-fold exponential function.
In such a case we would say that $\pi_n$ is structurally fairly regular, as we can describe it with an SF sentence of polynomial length.
Now imagine a property $\pi'_n$ accompanied with corresponding functions $f'(n)$ and $g'(n)$ for which we have $g'(n) \in \Omega\bigl(f'(n)\bigr)$, i.e.\ the length of shortest SF sentences describing $\pi'_n$ is asymptotically of the same order as the length of shortest BSR sentences describing $\pi'_n$.
On an intuitive level, this means that the relaxed syntactic conditions of SF do not provide a significant edge over BSR when $\pi'_n$ is to be described.
It seems that $\pi'_n$ requires a more sophisticated and lengthy description than, for instance, $\pi_n$ does, or, viewed from the opposite angle, $\pi'_n$ exhibits a lower degree of structural regularity than $\pi_n$.
A possible measure for this lack of regularity might be provided by the gap between $f'(n)$ and $g'(n)$: the smaller the gap, the higher the \emph{structural irregularity} of $\pi'_n$.

Instead of the comparison SF versus BSR, one could also use the comparison between SF sentences and equivalent Gaifman-local sentences.
Of course, the above said is also relevant to other fragments and not exclusively to SF ---
for instance, to SLGF versus LGF, to \SFOtwo\ versus \FOtwo, or to the full class of relational first-order sentences versus relational Gaifman-local sentences.

The general idea of measuring structural regularity by means of the asymptotic length of shortest logical descriptions appears to have some similarity to concepts investigated in the field of algorithmic information theory and Kolmogorov complexity (see~\cite{Downey2010, Calude2002} for introductory material).
Potential connections and interrelations remain to be studied.



\appendix
\section{Proof of Theorem~\ref{theorem:LengthSmallestGNFOsentences}}\label{appendix:LengthSmallestGNFOsentences}

\begin{proof}[Proof sketch]
	Let $n \geq 3$.	
	The following SGNFO sentence is equivalent to the sentence $\varphi$ given in the proof of Theorem~\ref{theorem:LengthSmallestLGFsentences}.
	In the sentence the sets $\{x_1, \ldots, x_n\}$ and $\{y_1, \ldots, y_n\}$ and $\{z\}$ are separated:
		\begin{align*}
				&\exists z.\, z = z  \\
					&\;\;\wedge\; \neg \exists x_n.\, R_n(x_n) \;\wedge\; \neg \exists y_n.\, R_n(x_n) \wedge T_n(y_n) \\
					&\quad\;\; \wedge\; \neg \exists x_{n-1}.\, R_{n-1}(x_{n-1}, x_n) \wedge T_n(y_n) \;\wedge\; \neg \exists y_{n-1}.\, R_{n-1}(x_{n-1}, x_n) \wedge T_{n-1}(y_{n-1}, y_n) \\
			 		&\qquad\;\; \ldots \\
			 			&\qquad\;\; \wedge\; \neg \exists x_1.\, R_1(x_1, \ldots, x_n) \wedge T_2(y_2, \ldots, y_n) \;\wedge\; \neg \exists y_1.\,  R_1(x_1, \ldots, x_n) \wedge T_1(y_1, \ldots, y_n) \\
			 				&\qquad\quad\;\; \wedge\; \bigwedge_{i=1}^{4n} \bigl( R_1(x_1, \ldots, x_n) \wedge T_1(x_1, \ldots, x_n) \;\wedge\; P_i(x_1, \ldots, x_n) \wedge Q_i(y_1, \ldots, y_n) \bigr) \\
			 				&\qquad\qquad\qquad\;\; \vee\; \bigl( R_1(x_1, \ldots, x_n) \wedge T_1(x_1, \ldots, x_n) \;\wedge\; \neg \bigl( P_i(x_1, \ldots, x_n) \vee Q_i(y_1, \ldots, y_n) \bigr) \bigr) .
		\end{align*}	
	The subformulas $z=z$, $R_n(x_n)$, and the more complex $R_i(x_i, \ldots, x_n) \wedge T_i(y_i, \ldots, y_n)$ and $R_i(x_i, \ldots,$ $x_n) \wedge T_{i+1}(y_{i+1}, \ldots, y_n)$ serve as separated negation guards for negated subformulas.
	We need to introduce a bit of redundancy in order to meet the syntactic requirements of SGNFO.
	Nevertheless, it is easy to see that the above sentence is equivalent to the sentence $\varphi$ used in the proof of Theorem~\ref{theorem:LengthSmallestLGFsentences}.
	We will refer to it as $\varphi$ as well.

	The rest of the proof works in analogy to the proof of Theorem~\ref{theorem:LengthSmallestLGFsentences}.
	We therefore take over definitions and notation from that proof, in particular the sets $\cS_k$, the model $\cA$ with the domain elements $\fa^{(j)}_{S}, \fb^{(j)}_{S}$ and subdomains $\fUA_\fa, \fUA_\fb, \fUA_{\fb, k}$, the notation $\cA_{-S}$, the vocabularies $\Sigma, \Sigma_{PR}, \Sigma_{QT}$ and their extensions $\Sigma'_{PR}, \Sigma'_{QT}$ with nullary predicate symbols, and the notion of column-$k$-occurrence.
	
	Let $\varphi_\GNFO$ be a shortest GNFO sentence that is semantically equivalent to $\varphi$.
	In analogy to the proof of Theorem~\ref{theorem:LengthSmallestLGFsentences}, our goal is to show that $\len(\varphi_\GNFO)$ is at least $(n-1)$-fold exponential in $n$. 
	Indeed, we can transform $\varphi_\GNFO$ into an equivalent GNFO sentence $\psi_\GNFO$ with the following properties.
	
	\begin{enumerate}[label=(\alph{*}), ref=(\alph{*})]
		\item \label{enum:theoremLengthSmallestGNFOsentences:I} 
			The sentence $\psi_\GNFO$ is a Boolean combination of GNFO $\Sigma'_{PR}$-sentences and GNFO $\Sigma'_{QT}$-sentences.
			Moreover, none of the constituent sentences of $\psi_\GNFO$ properly contains a non-atomic subsentence.
			
		\item \label{enum:theoremLengthSmallestGNFOsentences:II} 
			The vocabulary of $\psi_\GNFO$ is $\Sigma'_{PR} \cup \Sigma'_{QT}$, i.e.\  $\Sigma$ plus fresh nullary predicate symbols.
			
		\item \label{enum:theoremLengthSmallestGNFOsentences:III} 
			The structure $\cA$ can be uniquely expanded to some model $\cB$ of $\psi_\GNFO$ over the same domain and conserving the interpretations of all predicate symbols in $\Sigma$;
			 for every $\cB_{-S}$
			 --- the substructure of $\cB$ induced by the domain $\fUB \setminus \{\fb^{(k)}_S\}$ ---
			 we have $\cB_{-S} \not\models \psi_\GNFO$.
			 
		\item \label{enum:theoremLengthSmallestGNFOsentences:IV} 
			$\len(\psi_\GNFO) \in \cO\bigl( \len(\varphi_\GNFO) \bigr)$.
		
		\item \label{enum:theoremLengthSmallestGNFOsentences:V} 
			Equations in $\Sigma'_{QT}$-subsentences of $\psi_\GNFO$ occur only in guards, and these consist of a single trivial equation, say $v = v$,  and no other atoms. 
			
		\item \label{enum:theoremLengthSmallestGNFOsentences:VI} 
			Every non-equational $\Sigma_{QT}$-atom is linear and for every variable $v$ occurring in it there is some $k$ such that all occurrences of $v$ in non-equational $\Sigma_{QT}$-atoms are column-$k$-occurrences.
		
		\item \label{enum:theoremLengthSmallestGNFOsentences:VII} 
			For all distinct variables $v, v'$ that occur freely in a $\Sigma'_{QT}$-subformula $\chi$ and have column-$k$-occurrences and column-$k'$-occurrences, respectively, we know that $k \neq k'$.
	\end{enumerate}
	The transformation of $\varphi_\GNFO$ into $\psi_\GNFO$ proceeds in analogy to the corresponding transformation in the proof of Theorem~\ref{theorem:LengthSmallestLGFsentences}, except that the transformation into negation normal form is postponed.
	When we next transform $\psi_\GNFO$ into negation normal form, we obtain the equivalent sentence $\psi$, which does not necessarily belong GNFO anymore, but still shares the properties that are essential for the rest of the proof.
	
	Now, suppose that $\psi$ has fewer than $\twoup{n-1}{n}$ subformulas.
	By~\ref{enum:theoremLengthSmallestGNFOsentences:III}, we have $\cB \models \psi$ and $\cB_{-S} \not\models \psi$ for every $S \in \cS_n$.
	Hence, for every $S \in \cS_n$ there is some $\Sigma'_{QT}$-subformula $\psi_S$ in $\psi$ of the form
		$\exists \vy.\, \chi_S(\vy, \vz)$
	and some variable assignment $\beta_S$ such that the following properties hold.
	We have $\beta_S(y_*) = \fb^{(n)}_S$ for exactly one $y_* \in \vy$ and for every $v \in \vy \cup \vz$ different from $y_*$ we have $\beta_S(v) \in \fUA_\fb \setminus \fUA_{\fb,n}$.
	Moreover, we have \\
	\centerline{
		\begin{tabular}{c@{\hspace{1ex}}l}
			($*$) & $\cB, \beta_S \models \chi_S(\vy, \vz)$ and $\cB, \beta' \not\models \chi_S(\vy, \vz)$ for every $\beta'$ that differs from $\beta_S$ only in the \\
				& value assigned to $y_*$.
		\end{tabular}
	}
	The tuple $\beta_S(\vz)$ represents a sequence $\vc_S$ of domain elements from $\fUA_\fb$ that can be completed to a chain $\fb^{(1)}_{T_1}, \ldots, \fb^{(n-1)}_{T_{n-1}}, \fb^{(n)}_S$ with $T_1 \in \ldots \in T_{n-1} \in S$.
		
	Fix any $S_* \in \cS_n$ and consider the formula $\psi_{S_*}(\vz)$.
	There is a nonempty set $\hS_*$ such that $\psi_{S_*}(\vz)$ coincides with every $\psi_S(\vz)$ with $S \in \hS_*$.
	For any distinct $S, S' \in \hS_*$ the sequences $\vc_S := \beta_S(\vz)$ and $\vc_{S'} := \beta_{S'}(\vz)$ must differ, for otherwise ($*$) would be violated.
	As there are at most $\prod_{k=1}^{n-1} \twoup{k}{n}$ distinct sequences $\vc_{S}$, $\hS_*$ can contain at most $\prod_{k=1}^{n-1} \twoup{k}{n} < \bigl( \twoup{n-1}{n} \bigr)^n$ sets.
	Recall that there are fewer than $\twoup{n-1}{n}$ subformulas in $\psi$.
	We have just inferred that each of these can only serve as $\psi_S$ for at most $\bigl( \twoup{n-1}{n} \bigr)^n$ sets $S \in \cS_n$.
	Hence, only 
		\[ \bigl( \twoup{n-1}{n} \bigr)^n \cdot \twoup{n-1}{n} \;=\; 2^{(n+1) \cdot \twoup{n-2}{n}} \;<\; 2^{\twoup{n-1}{n}} \;=\; \twoup{n}{n} \]
	sets $S$ have a corresponding subformula $\psi_S$.
	But $|\cS_n| \geq \twoup{n}{n+1}$ implies that there are $S \in \cS_n$ such that $\cB_{-S} \models \psi_\LGF$, which contradicts our assumptions.
	Consequently, there must be more than $\twoup{n-1}{n}$ subformulas in $\psi$ and in $\psi_\GNFO$ and, by~\ref{enum:theoremLengthSmallestGNFOsentences:IV}, also in $\varphi_\GNFO$.	
\end{proof}

\end{document}